\documentclass[11pt,fleqn]{article}

\usepackage{amsmath,amssymb,amsthm,enumerate,epsfig,graphicx,overpic,cite}

\newcommand{\p}{\partial}
\newcommand{\const}{\mathop{\rm const}\nolimits}

\newcommand{\sign}{\mathop{\rm sign}\nolimits}
\newcommand{\pr}{\mathop{\rm pr}\nolimits}
\newcommand{\thetbn}{\arabic{nomer}}
\newcommand{\CV}{\mathop{\rm CV}\nolimits}
\newcommand{\CL}{\mathop{\rm CL}\nolimits}
\newcommand{\Ch}{\mathop{\rm Ch}\nolimits}

\newcommand{\ord}{\mathop{\rm ord}\nolimits}
\newcommand{\sco}{\mathop{\rm sco}\nolimits}
\newcommand{\wsco}{\mathop{\rm wsco}\nolimits}
\newcommand{\Div}{\mathop{\rm Div}\nolimits}

\newcounter{tbn}
\newcounter{casetran}
\newcounter{clnumber}

\newtheorem{theorem}{Theorem}
\newtheorem{lemma}{Lemma}
\newtheorem{corollary}{Corollary}
\newtheorem{proposition}{Proposition}
{\theoremstyle{definition} \newtheorem{definition}{Definition}

\newtheorem{remark}{Remark}

\begin{document}

\begin{center}
\par\noindent {\LARGE\bf Group-theoretical analysis of variable coefficient\\ nonlinear
telegraph equations\par} {\vspace{6mm}\par\noindent {\bf Ding-jiang
Huang~$^{~*, ~\S,~\dag}$
and Shuigeng Zhou~$^{*, ~\S}$ }

{\vspace{2mm}\par\noindent {\it $^*$~School
of Computer Science, Fudan University, Shanghai 200433, China, \\
}}

{\par\vspace{2mm}\noindent\vspace{2mm} {\it $^\S$~Shanghai Key Lab of
Intelligent Information Processing, Fudan University,\\ Shanghai
200433, China }

\par\par} {\vspace{2mm}\par\noindent {\it
$^\dag$~Department of Mathematics, East China University of Science and Technology,\\  Shanghai
200237, China\\
}} {\noindent \vspace{2mm}{ $\phantom{^\dag}$~E-mail:
\{djhuang, sgzhou\}@fudan.edu.cn}\par
}}
\end{center}


{\vspace{5mm}\par\noindent\hspace*{8mm}\parbox{140mm}{\small Given a
class $\mathcal {F(\theta)}$ of differential equations with
arbitrary element $\theta$, the problems of symmetry group,
nonclassical symmetry and conservation law classifications are to
determine for each member $f\in \mathcal {F(\theta)}$ the structure
of its Lie symmetry group $G_f$, conditional symmetry $Q_f$ and
conservation law $\CL_f$ under some proper equivalence
transformations groups.

~~~~~~~In this paper, an extensive investigation of
these three aspects is carried out for the class of variable
coefficient (1+1)-dimensional nonlinear telegraph equations with
coefficients depending on the space variable
$f(x)u_{tt}=(g(x)H(u)u_x)_x+h(x)K(u)u_x$. The usual equivalence
group and the extended one including transformations which are
nonlocal with respect to arbitrary elements are first constructed.
Then using the technique of variable gauges of arbitrary elements
under equivalence transformations, we restrict ourselves to the
symmetry group classifications for the equations with two different
gauges $g=1$ and $g=h$. In order to get the ultimate classification,
the method of furcate split is also used and consequently a number
of new interesting nonlinear invariant models which have non-trivial
invariance algebra are obtained. As an application, exact solutions
for some equations which are singled out from the classification
results are constructed by the classical method of Lie reduction.

~~~~~~~The classification of nonclassical symmetries for the classes of
differential equations with gauge $g=1$ is discussed within the
framework of singular reduction operator. This enabled to obtain
some exact solutions of the nonlinear telegraph equation which are
invariant under certain conditional symmetries.

~~~~~Using the direct method, we also carry out two classifications of local conservation
laws up to equivalence relations generated by both usual and
extended equivalence groups. Equivalence with respect to these
groups and correct choice of gauge coefficients of equations play
the major role for simple and clear formulation of the final
results.
 }\par\vspace{5mm}}

\section{Introduction}

Since the notation of continuous group was introduced by Lie at the
end of 19th century, significant progress in application of
symmetries to analysis of concrete nonlinear differential equations
has been achieved. The classical Lie symmetries of nonlinear
differential equations allow us to find explicit solutions,
conservation laws, linearizing substitutions of the Hopf-Cole type,
etc \cite{Bluman&Kumei1989,
Ibragimov1985,Ibragimov1994V1,Olver1986,Ovsiannikov1982}. For most
of the application the search of explicit symmetry structure of the
corresponding differential equations is always a crucial step, which
consist of the cornerstone of group analysis of differential
equation. Consequently, there are two interrelated problems which
still remain to be solved in the traditionally group analysis of
differential equations. The first problem consists of finding the
maximal Lie symmetry group admitted by a given equation. The second
problem is one of classifying differential equations that admit a prescribed symmetry group. The principal
tool for handling both problems is the classical infinitesimal Lie
method \cite{Bluman&Kumei1989,
Ibragimov1985,Olver1986,Ovsiannikov1982}. It reduces the problem to
finding the corresponding Lie symmetry algebra of infinitesimal
operators whose coefficients are found as solutions of some
over-determined system of linear partial differential equations
(PDEs). However, if the equations under study contains with arbitrary
element (functions or parameters) , then one has to solve an intermediate classification
problem. Namely, it is necessary to describe all the possible forms
of the functions involved such that this equation admits a
nontrivial invariance group. Generally, the problem
can be described as follows: Given a class $\mathcal {F(\theta)}$ of
differential equations with arbitrary element $\theta$, the problem
of group classification is to determine for each member $f\in
\mathcal {F(\theta)}$ the structure of its Lie symmetry group $G_f$,
or equivalently of its Lie symmetry algebra $A_f$ under some proper
equivalence transformation groups. This description is also fit for
the problems of nonclassical symmetry and conservation law
classification by replacing Lie symmetry with these two different
notations.

Historically, the first classification of Lie symmetries was derived
by Lie, he proved that a linear two-dimensional second-order PDE may
admit at most a six-parameter invariance group (apart from the
trivial infinite parameter symmetry group, which is due to
linearity)\cite{Lie1881}. The modern formulation of the problem of
group classification of PDEs was suggested by Ovsiannikov in 1959
\cite{Ovsiannikov1959}, in which he present complete group
classification of a class of nonlinear heat conductivity equations
by using the technique of equivalence group and direct integration.
After that, the group classification of nonlinear PDEs became the
subject of intensive research. A detailed survey of the work done in
this area up to the beginning of the 1990's is given
in~\cite{Ibragimov1994V1}.

In the present paper we investigate a class of hyperbolic type
variable coefficient (1+1)-dimensional nonlinear telegraph equations
of the form
\begin{equation}\label{eqVarCoefTelegraphEq}
f(x)u_{tt}=(g(x)H(u)u_x)_x+h(x)K(u)u_x
\end{equation}
where $f=f(x)$, $g=g(x)$, $h=h(x)$, $H=H(u)$ and $K=K(u)$ are
arbitrary and sufficient smooth real-valued functions of their
corresponding variables, $f(x)g(x)H(u)\neq 0$. In what follows, we
assume that $(H_u, K_u)\neq(0, 0)$,
i.e.,~\eqref{eqVarCoefTelegraphEq} is a nonlinear equation. This is
because the linear case of~\eqref{eqVarCoefTelegraphEq} ($H,
K=\const$) was studied by Lie~\cite{Lie1881} in his classification
of linear second-order PDEs with two variables. (See also a modern
treatment of this subject in~\cite{Ovsiannikov1982}).

The study of equation~\eqref{eqVarCoefTelegraphEq} is strongly
stimulated not only by its intrinsic theoretical interest but also by its significant applications
in Mathematics and Engineering. In fact,
hyperbolic type second-order nonlinear PDEs in two independent
variables are usually used to describe different types of wave
propagation. They are also used in differential geometry, in various
fields of hydro- and gas dynamics, chemical technology, super
conductivity, crystal dislocation to mention only a few applications
areas. The corresponding models are comprised by the Liouville, sine/sinh-Gordon,
Goursat, d'Alembert, Tzitzeica and nonlinear telegraph equations and
a couple of others. From the group-theoretical viewpoint the
popularity of these models is due to the fact that they have
non-trivial Lie or Lie--B\"acklund
symmetry~\cite{Bluman&Anco2002,Bluman&Kumei1989,
Fushchych&Shtelen&Serov1993,Ibragimov1985,Ibragimov1994V1,
Ibragimov1999,Olver1986,Ovsiannikov1982,Stephani1994}. By this very
reason they are either integrable by the inverse problem methods or
are linearizable and completely integrable
~\cite{AblowitzB1991,Novikov&Manakov1980,FaddeevB1987}.

The investigation of Lie symmetry classification of the
$(1+1)$-dimensional hyperbolic type second-order nonlinear PDEs has
a long history. Probably, Barone {\it et
al}~\cite{Barone&Esposito&Magee&Scott1971} was the first study of
the following nonlinear wave equation $u_{tt}=u_{xx}+F(u),$ by means
of symmetry method, this equation was also studied by
Kumei~\cite{Kumei1975} and Pucci {\it  et
al}~\cite{Pucci&Salvatori1986} subsequently. Motivated by a number
of physical problems, Ames {\it  et al}~\cite{Ames1965/72,Ames1981}
investigated group properties of quasi-linear hyperbolic equations
of the form
\begin{equation}\label{eqWaveEq}
u_{tt}=[f(u)u_x]_x.
\end{equation}
Later, their investigation was generalized
in~\cite{Torrisi&Valenti1985,Donato1987,Ibragimov&Torrisi&Valenti1991}
to equations of the following forms respectively
\begin{gather*}
u_{tt}=[f(x, u)u_x]_x,\quad
u_{tt}=[f(u)u_x+g(x,u)]_x,\quad\mbox{and}\quad u_{tt}=f(x,
u_x)u_{xx}+g(x,u_x).
\end{gather*}
The alternative form of equation~\eqref{eqWaveEq} was also
investigated by Oron and Rosenau~\cite{Oron&Rosenau1986} and Suhubi
and Bakkaloglu~\cite{Suhubi&Bakkaloglu1991}.
Arrigo~\cite{Arrigo1991} classified the equations
$
u_{tt}=u_x^{m}u_{xx}+F(u).
$
Furthermore, classification results for the equation
$
u_{tt}+K(u)u_t=[F(u)u_x]_x
$
can be found in~\cite{Ibragimov1994V1,Oron&Rosenau1986}. An expand
form of the latter equation
$
u_{tt}+K(u)u_t=[F(u)u_x]_x+H(u)u_x
$
was studied by Kingston and
Sophocleous~\cite{Kingston&Sophocleous2001}. Recently, Zhdanov and
Lahno~\cite{Lahno&Zhdanov2005} presented the most extensive list
of symmetries of the equations
\[
u_{tt}=u_{xx}+F(t, x, u, u_x)
\]
by using the infinitesimal Lie method, the technique of
equivalence transformations and the theory of classification of
abstract low-dimensional Lie algebras. There are also some
papers~\cite{Oron&Rosenau1986,Chikwendu1981,Gandarias&Torrisi&Valenti2004,Pucci1987}
devoted to the group classification of the equation of the
following form
\begin{gather*}
u_{tt}=F(u_{xx}), \quad
u_{tt}=F(u_x)u_{xx}+H(u_x),\quad\mbox{and}\quad u_{tt}+\lambda
u_{xx}=g(u,u_x).
\end{gather*}
It is worthwhile mentioned that the constant-coefficient nonlinear
telegraph equations
\[
u_{tt}=(F(u)u_x)_{x}+H(u)u_x
\]
together with its equivalent potential systems have also been
studied by Bluman {\it et
al}~\cite{Bluman&Temuerchaolu&Sahadevan2005,Bluman&Temuerchaolu2005a,
Bluman&Temuerchaolu2005b,Bluman&Cheviakov&Ivanova2005}. In their a
series of papers, many interesting results (especially for case of
power nonlinearities) including Lie point and nonlocal symmetries
classification and conservation law of the four equivalent systems
were systematically investigated. Recently, Huang and Ivanova
present a strong complete group classification for a class of
variable coefficient (1+1)-dimensional nonlinear telegraph
equations of the form \cite{Huang&Ivanova2007}
\begin{equation}\label{eqVarCoefTelegraphEqgx0}
f(x)u_{tt}=(H(u)u_x)_x+K(u)u_x.
\end{equation}
Exact solutions and classifications of conservation law with
characteristics of order $0$ were also investigated
\cite{Huang&Ivanova2007}.

From the above introduction, we can see that
equation~\eqref{eqVarCoefTelegraphEq} is different from any
aforementioned ones and is a generalization of many well studied
equations. What's more, equations~\eqref{eqVarCoefTelegraphEq} can
be used to model a wide variety of phenomena in physics, chemistry,
mathematical biology etc (see Section 2 for detail). Thus there is
essential interest in investigating them from a unified and group
theoretical viewpoint.

In this paper, extended group analysis of class (1) is first carried
out. The usual equivalence group and the extended one including
transformations which are nonlocal with respect to arbitrary
elements are constructed for class (1) and its subclasses. The
structure of the extended equivalence group and non-trivial subgroup
of (nonlocal) gauge equivalence transformations are investigated. As
a result, group classification problems related to two different
gauges $g=1$ and $g=h$ are really solved for each form with respect
to the corresponding usual and extended equivalence group. Classical
Lie reduction of some classification models are described and exact
solutions are obtained by using the reduction. Nonclassical
symmetries classification of class (1) with the gauges $g=1$ is
discussed within the framework of singular reduction operator.
Several nonclassical symmetries for equation form class (1) are
constructed. This enabled to obtain some exact solutions of the
nonlinear telegraph equation which are invariant under certain
conditional symmetries. Using the most direct method, two
classifications of local conservation laws up to equivalence
relations are generated by both usual and extended equivalence
groups. Equivalence with respect to these groups and correct choice
of gauge coefficients of equations play the major role for simple
and clear formulation of the final results.

Problems of group classification, except for really trivial
cases, are very difficult. Generally, there are two main approaches in studying
group classification problems in the literature. The first one is
the algebraic methods and is based on subgroup analysis of the equivalence
group associated with a class of differential equations under
consideration. Its main idea rely on the description of
inequivalent realizations of Lie algebras in certain set of vector
fields of the equation under consideration
\cite{Basarab&Zhdanov&Lahno2001,Zhdanov&Lahno1999}, which was
original from S. Lie \cite{Ibragimov1994V1,Lie1881} and recently
rediscovered by Winternitz and Zhdanov et al
\cite{Gazeau&Winternitza1992,Zhdanov&Lahno1999}. The method has been
applied to classifying a number of nonlinear differential equations
\cite{Abramenko&Lagno&Samoilenko2002,Basarab&Zhdanov&Lahno2001,Gagnon&Winternitza1993,
Gazeau&Winternitza1992,Huang&Zhang2009,Huang&Lahno&Qu&Zhdanov2009,Huang&Qu&Zhdanov2009,
Lahno&Zhdanov2005,Lahno&Samoilenko2002,Zhdanov&Lahno1999,Zhdanov&Lahno2005,Zhdanov&Lahno2007},
including the class is normalized (see
\cite{Popovych&Kunzinger&Eshraghi2010} for rigorous definitions of
normalized classes and related notions). The second approach is
based on the investigation of compatibility and the direct
integration, up to the equivalence relation generated by the
corresponding equivalence group, of determining equations implied by
the infinitesimal invariance criterion \cite{Ovsiannikov1982}. This
method was suggested by L.V Ovsiannikov and referred as the
Lie-Ovsiannikov method. This is the most applicable approach but it
is efficient only for classes of a simple structure, e.g., which
have a few arbitrary elements of one or two same arguments or whose
equivalence groups are finite-dimensional. A number of results on
group classification problems investigated within the framework of
this approach are collected in \cite{Bluman&Kumei1989,
Ibragimov1994V1,Ovsiannikov1982} and other books on the subject.

Recently, based on the Lie-Ovsyannikov method and the investigation
of the specific compatibility of classifying conditions, Nikitin and
Popovych~\cite{Nikitin&Popovych2001} developed an effective tool (we
refer it as method of furcate split) for solving the group
classification problem of nonlinear Schr\"{o}dinger equation. In
2004, Popovych and Ivanova extended the method to complete group
classification of nonlinear diffusion-convection equations by
further considering the so called additional and conditional
equivalence transformations~\cite{Popovych&Ivanova2004NVCDCEs}. In
2007, Ivanova, Popovych and Sophocleous present the extended and
generalized equivalence transformation group, gauging of arbitrary
elements by equivalence transformations for further investigation of
nonlinear diffusion-convection
equations\cite{Ivanova&Popovych&Sophocleous2007}. Furthermore, Popovych and Ivanova et.al also
extended these new group classification idea to the nonclassical symmetries \cite{Kun&Pop2008,Vaneeva&Popovych&Sophocleous2010}
and conservation law classification\cite{Popovych&Ivanova2004ConsLaws,Ivanova&Popovych&Sophocleous2007-3}. Up to now, these
methods and different notations have been applied to investigating a
number of different symmetry group, nonclassical symmetry and conservation law classification
problems~\cite{Meleshko1994,Nikitin&Popovych2001,Popovych&Cherniha2001,
Boyko&Popovych2001,Popovych&Ivanova2004NVCDCEs,Popovych&Ivanova2004ConsLaws,Ivanova&Sophocleous2006,
Ivanova&Popovych&Sophocleous2007,Ivanova&Popovych&Sophocleous2007-3,Vasilenko&Yehorchenko2001,Popovych&Kunzinger&Eshraghi2010,
Vaneeva&Johnpillai&Popovych&Sophocleous2007,Kun&Pop2008,Vaneeva&Popovych&Sophocleous2010}.

However, almost all the research was concentrated on parabolic type
nonlinear diffusion-convection equations and few of hyperbolic type
nonlinear partial differential
\cite{Huang&Ivanova2007,Huang&Mei&Zhang2009,Huang&Zhou2010}. Therefore, the present paper is one of new extension
of the above mentioned method and different notations to this classes of equations.
The results of symmetry group, nonclassical symmetry and conservation law classification of
class~\eqref{eqVarCoefTelegraphEq} present in this work are
new. Hence, these will lead to some explicit applications in Physics
and Engineering.

The structure of the paper is as follows:

Some physical examples contained in class
\eqref{eqVarCoefTelegraphEq} is discussed in section
\ref{SectionOnPhysExam}.

In section \ref{SectiononEquivaTrans}, the complete group of usual
equivalence transformations for class \eqref{eqVarCoefTelegraphEq}
and the extended one including transformations which are nonlocal
with respect to arbitrary elements are constructed by using
Lie-Ovsiannikov method and direct method. Taking into account the
non-trivial subgroup of gauge equivalence transformations, we
strongly simplify the solving the group classification problem to
equation \eqref{eqVarCoefTelegraphEq} with two different gauges
$g=1$ and $g=h$.

The results of group classification under the extended equivalence
transformation group for the gauge $g = 1$ are contained in
section~\ref{subSectionOnGrClasRes1}. Then in section \ref{SectionOnGrClasResg=h}, the classification of
gauge $g=h$ are presented. The sketch of the proof of the obtained results are given in Section
\ref{SectionOnGrClasProof}. Classification with respect to the set of point transformations are presented in
\ref{SectionOnAdditEquivTr}. We note that for both gauges two
essentially different classifications are presented: the
classification with respect to the (extended) equivalence group and
the classification with respect to all possible point
transformations.

In Section \ref{SectionOnLieRedu}, exact solutions of some
classification models are given  by using the method of classical
Lie reduction.

After making a brief review of notation of singular reduction
operator, we then carry out a preliminary analysis of nonclassical
symmetry of the class \eqref{eqVarCoefTelegraphEq} with the gauge
$g=1$ in section \ref{SectionOnNonclassicalSym}. As an example, we
also present several reduction operators of a special nonlinear
telegraph equation and constructed some non-Lie exact solutions for
them.

In Section \ref{SectionOnConsLaws}, the local conservation laws of
these equations are exhaustively described. Using the most direct
method, two classifications of local conservation laws up to
equivalence relations are generated by both usual and extended
equivalence groups. Equivalence with respect to these groups and
correct choice of gauge coefficients of equations play the major
role for simple and clear formulation of the final results.

Finally, some conclusion and discussion are given in
Section~\ref{SectionOnConclusion}.

In the Appendix, classification results for the gauge $g = 1$ under
the usual equivalence transformation group can be found.

\section{Physical examples}\label{SectionOnPhysExam}
Class~\eqref{eqVarCoefTelegraphEq} is a unified form of many
significant second-order hyperbolic type nonlinear PDEs in Physics,
Mechanics and Engineering Science. Physical examples corresponding to the
case $f(x)=g(x)=h(x)=1$ and $K(u)=0$ are collected in the well known
paper \cite{Ames1965/72,Ames1981}, which describe the flow of
one-dimensional gas, longitudinal wave propagation on a moving
threadline and dynamics of a finite nonlinear string and so on. In
what follows, we review several important physical models related
with the coefficient  $f(x)\neq 0$ or $h(x)K(u)\neq 0$
\cite{Torrisi&Valenti1990}.

{\bf Example 1.} Two-conductor transmission lines telegraph equation.
The waves in two-conductor transmission lines having small transverse
dimensions (in comparison with the characteristic
wavelength) can often be described by the telegraph
equation \cite{Katayev1966}
\begin{equation}\label{EqOnTwoConTranLineTele_1}
I_x = U_t ,\quad I_t = F(U)U_x + G(U),
\end{equation}
where $t$ is a spatial variable and $x$ is time; $I, U, F(U)$, and $G(U)$ are respectively the current in the
conductors, the voltage between the conductors, the leakage current per unit length,
and the differential capacitance. The form of $F(U)$, and $G(U)$  depend both on the configuration of
the conductors, and on the properties of the medium
filling it.

Setting $U=u$, we obtain the telegraph equation
\begin{equation}\label{EqOnTwoConTranLineTele_2}
u_{tt} = (F(u)u_x)_x + G'(u)u_x,
\end{equation}
which fall into  \eqref{eqVarCoefTelegraphEq}.

{\bf Example 2.} Longitudinal vibrations of elastic and non-homogeneous taut strings or bars. Suppose a string is taut along the $x-$axis. The equation giving the balance of momentum is
\begin{equation}\label{EqOnLongVibraElaBar_1}
\rho \omega_{tt}=T_x,
\end{equation}
where $x$ is the coordinate of the point $P$ in the present reference system, and
\[
x=x(y,t),
\]
where $y$ represents the coordinate of the corresponding point $P_0$ of $P$ in the reference shape, where $\omega=x-y, \rho$ is mass per unit length and $T$ is the tension. To equation \eqref{EqOnLongVibraElaBar_1} we associate the following constitutive relations already considered in \cite{Keller&Ting1966}:
\begin{equation}\label{EqOnLongVibraElaBar_2}
T=T(\omega_{x}), \rho=\rho(x).
\end{equation}
The balance law \eqref{EqOnLongVibraElaBar_1}, with \eqref{EqOnLongVibraElaBar_2}, transforms to the following second order partial differential equation
\begin{equation}\label{EqOnLongVibraElaBar_3}
u_{tt}=[\frac{T'(u)}{\rho(x)}u_x]_x,.
\end{equation}
where $\omega_x=u$. Of course \eqref{EqOnLongVibraElaBar_3} is particular case of \eqref{eqVarCoefTelegraphEq}.

{\bf Example 3.} Bar with variable cross section. The equation of
motion of a hyperelastic homogeneous bar, whose cross sectional area
is variable along the bar, is \cite{Cristescu1967}
\begin{equation}\label{EqOnVarCroSecBar_1}
\rho \omega_{tt}=T_x+\frac{S'(x)}{S(x)}T,
\end{equation}
where $\rho$ is the (constant) mass density, $\omega=y-x$ is the
displacement, $y$ is the coordinate of the point $P$ in the present
reference system, $x$ represents the coordinate of the corresponding
point $P_0$ of $P$ in the reference frame, $T$ is the tension and
$S(x)$ is the cross sectional area.

Taking into account the constitutive relation $T=T(\omega_x)$
\cite{Cristescu1967}, the equation \eqref{EqOnVarCroSecBar_1}
becomes
\begin{equation}\label{EqOnVarCroSecBar_2}
u_{tt}=[\frac{T'(u)}{\rho}u_x]_x+(\frac{S'(x)}{S(x)})_x\frac{T(u)}{\rho}+\frac{S'(x)}{S(x)}\frac{T'(u)}{\rho}u_x,
\end{equation}
where $u=\omega_x$. Obviously when $(\frac{S'(x)}{S(x)})_x=0$ the
equation \eqref{EqOnVarCroSecBar_2} is included in
\eqref{eqVarCoefTelegraphEq}.

{\bf Example 4.} One dimensional propagation of visco-elastic stress
waves. In \cite{Varley&Seymour1985}, the following constitutive laws
were adopted for a nonlinear homogeneous visco-elastic model
\begin{equation}\label{EqOnVisElaStrewave_1}
e_t=\Phi(T)T_t+\Omega(T),
\end{equation}
with
\begin{equation}\label{EqOnVisElaStrewave_2}
\Phi(T)=\frac{d_0}{(1+kc_0T)^2},\quad\Omega(T)=\frac{c_0T}{1+kc_0T},\quad
e_t=\omega_{xt}
\end{equation}
where $T$ is the stress, $\omega$ is the displacement, $d_0>0, k$
and $c_0$ are constants. When the one-dimensional propagation of
visco-elastic stress waves is investigated, combining the momentum
equation
\begin{equation}\label{EqOnVisElaStrewave_3}
\rho \omega_{tt}=T_x
\end{equation}
with \eqref{EqOnVisElaStrewave_1} and taking into account
\eqref{EqOnVisElaStrewave_2}, we obtain
\begin{equation}\label{EqOnVisElaStrewave_4}
T_{xx}=[\frac{d_0}{(1+kc_0T)^2}T_t]_t+\frac{c_0}{(1+kc_0T)^2}T_t.
\end{equation}

So, we fall into \eqref{eqVarCoefTelegraphEq} when $x$ and $t$ are
exchanged.

{\bf Example 5.} Hyperbolic heat equation. Many models for the heat
propagation  with finite speed give a hyperbolic equation which can
be reduced to \eqref{eqVarCoefTelegraphEq}. In fact, quite recently,
in order to describe one-dimensional heat conduction, the following
partial differential equation has been considered in
\cite{Rogers&Ruggeri1985}
\begin{equation}\label{EqOnHyperboHeat_1}
\theta_{tt}-\frac{q_0a^2}{\gamma^*}(\frac{\delta}{\theta}+\epsilon)^2\theta_{xx}-\frac{2\epsilon}{\delta+\epsilon
\theta}\theta_t^2+\frac{1}{\gamma^*}\theta_t
+\frac{2q_0a^2}{\gamma^*\theta}(\frac{\delta}{\theta}+\epsilon)^2\theta_{x}^2=0,
\end{equation}
where $\theta$ denote the (absolute) temperature while $q_0, a,
\gamma^*, \delta$ and $\epsilon$ are suitable constants. The
equation \eqref{EqOnHyperboHeat_1} is based on a nonlinear model
with relaxation.

If we set $\theta=1/u$, we obtain
\begin{equation}\label{EqOnHyperboHeat_2}
u_{xx}=[\frac{\gamma^*}{q_0a^2(\delta
u+\epsilon)^2}u_t]_t+\frac{1}{q_0a^2(\delta u+\epsilon)^2}u_t.
\end{equation}

This equation was obtained in \cite{Donato&Fusco1987} in the
framework of M\"{u}ller's theory for heat propagation in rigid
bodies.

{\bf Example 6.} One dimensional heat propagation in
a rigid body. The models like Cattaneo's
\cite{Cattaneo1948}, describing one dimensional heat propagation in
a rigid body, are governed by a nonlinear equation of the type
\cite{Engelbrecht1983}
\begin{equation}\label{EqOnCattaneo_1}
\theta_{xx}=[\frac{\tau_0}{\chi}C(\theta)\theta_t]_t+\frac{1}{\chi}[\int
C(\theta) d \theta]_\theta \theta_t.
\end{equation}
where $C(\theta)$ is the special heat, $\tau_0$ is the thermal
relaxation time and $\chi$ is the thermal conductivity.

The equation \eqref{EqOnHyperboHeat_2} and  \eqref{EqOnCattaneo_1}
fall into \eqref{eqVarCoefTelegraphEq} when $x$ and $t$ are
exchanged.

\section{Equivalence transformations and choice
of investigated class}\label{SectiononEquivaTrans}

In order to perform group classification of class~\eqref{eqVarCoefTelegraphEq},
we should first find its group of equivalence transformations.
The {\sl usual equivalence group} ~$G^{\sim}$  of class \eqref{eqVarCoefTelegraphEq} is formed by the nondegenerate point transformations
in the space of $(t, x, u, f, g, h, H, K)$, which are projectible on the space of $(t, x, u)$, i.e.
they have the form
\[
\begin{array}{ll}
(\tilde t, \tilde x, \tilde u)=(T^t, T^x, T^u)(t,x,u),\\
(\tilde f, \tilde g, \tilde h, \tilde H, \tilde K)=(T^f,T^g,T^h,T^H,T^K)(t,x,u,f,g,h,H,K),
\end{array}
\]
and transform any equation from class \eqref{eqVarCoefTelegraphEq} for the function $u = u(t, x)$ with the arbitrary
elements $(f, g, h, H, K)$ to an equation from the same class for the function $\tilde u = \tilde u(\tilde t, \tilde x)$ with the
new arbitrary elements $(\tilde f, \tilde g, \tilde h, \tilde H, \tilde K)$.
To find the connected component of the unity of~$G^{\sim}$, we have
to investigate Lie symmetries of the system that consists of
equation~\eqref{eqVarCoefTelegraphEq} and some additional
conditions, that is to say we must seek for an operator of the
$G^{\sim}$ in the form
\begin{equation}\label{OperatorEquivTr}
X=\tau \partial_t+\xi \partial_x+\eta \partial_u+\pi
\partial_f+\varphi \partial_g+\phi \partial_h+\rho \partial_H+\theta \partial_K
\end{equation}
from the invariance criterion of the following system:
\begin{equation}\label{sysForEquivTr}
\begin{array}{l}
f(x)u_{tt}=(g(x)H(u)u_x)_x+h(x)K(u)u_x,\\[1ex]
f_t=f_u=0,\quad g_t=g_u=0,\quad h_t=h_u=0,\quad H_t=H_x=0,\quad
K_t=K_x=0.
\end{array}
\end{equation}
Here $u$, $f$, $g$, $h$, $H$ and $K$ are considered as differential
variables: $u$ on the space $(t,x)$ and $f$, $g$, $h$, $H$, $K$ on
the extended space $(t, x, u)$. The coordinates $\tau$, $\xi$,
$\eta$ of the operator~\eqref{OperatorEquivTr} are sought as
functions of $t$, $x$, $u$ while the coordinates $\pi$, $\varphi$,
$\phi$, $\rho$ and $\theta$ are sought as functions of $t$, $x$,
$u$, $f$, $g$, $h$, $H$, $K$.

The invariance criterion of system~\eqref{sysForEquivTr} yields the
following determining equations for $\tau$, $\xi$, $\eta$, $\pi$,
$\varphi$, $\phi$, $\rho$ and $\theta$:
\begin{gather}\nonumber
\tau_x=\tau_u=\xi_t=\xi_u=\eta_x=\eta_t=0, \quad
\tau_{tt}=\eta_{uu}=0,\\ \nonumber
\pi_t=\pi_u=\pi_H=\pi_K,\quad
\varphi_t=\varphi_u=\varphi_H=\varphi_K=0,\\ \nonumber
\phi_t=\phi_u=\phi_H=\phi_K=0,\quad \rho_t=\rho_x=\rho_u=\rho_f=\rho_g=\rho_h=\rho_K=0,\\
\label{sysDetEqEquivTr} \theta_t=\theta_x=\theta_f=\theta_g=\theta_h=0,\\
\nonumber
\frac{\pi}{f}+2\xi_x-2\tau_t-\frac{\varphi}{g}-\rho_{H}=0,\quad
\frac{\pi}{f}+2\xi_x-2\tau_t-\frac{\varphi}{g}=\frac{\rho}{H},\\
\nonumber
g_x\rho+[-g\xi_{xx}+(2\tau_t-\xi_x-\frac{\pi}{f})g_x+\varphi_x+f_x\varphi_f+g_x\varphi_g+h_x\varphi_h-\xi_xg_x]H\\
\nonumber +h\theta+(2\tau_t-\xi_x-\frac{\pi}{f}+\frac{\phi}{h})hK=0.
\end{gather}
After easy calculations from~\eqref{sysDetEqEquivTr}, we can find
the connected component of the unity of $G^{\sim}$ for the class
\eqref{eqVarCoefTelegraphEq}.

\begin{theorem}\label{TheorOnEquivGrouUsaual}
The usual equivalence transformation group $G^{\sim}$ for the class
\eqref{eqVarCoefTelegraphEq} consists of the transformations
\[
\begin{array}{ll}
\tilde t =\delta_1 t +\delta_2,\quad \tilde x =X(x),\quad
\tilde u =\delta_3 u + \delta_4,\\
\tilde f =\frac{\epsilon_1\delta_1^2}{X_x}f,\quad \tilde g =
\epsilon_1\epsilon_2^{-1}X_xg,\quad \tilde h =
\epsilon_1\epsilon_3^{-1}h,\quad \tilde H = \epsilon_2 H ,\quad
\tilde K=\epsilon_3K,
\end{array}
\]
where $\delta_j~(j=1,\ldots, 4)$ and $\epsilon_i~(i=1,2,3)$ are
arbitrary constants,
$\delta_1\delta_3\epsilon_1\epsilon_2\epsilon_3\neq 0$, $X$ is an
arbitrary smooth function of $x$, $X_x\neq 0$.
\end{theorem}

It is shown that class \eqref{eqVarCoefTelegraphEq} admits other
equivalence transformations which do not belong to $G^{\sim}$ and
form, together with usual equivalence transformations, an {\sl extended
equivalence group}. We demand for these transformations to be point
with respect to $(t, x, u)$. The explicit form of the new arbitrary
elements $(\tilde f, \tilde g, \tilde h, \tilde H, \tilde K)$ is
determined via $(t, x, u, f, g, h, H, K)$ in some nonfixed
(possibly, nonlocal) way. We can construct the complete (in this sense)
extended equivalence group $\hat G^{\sim}$ of class
\eqref{eqVarCoefTelegraphEq}  by using the direct method.

\begin{theorem}\label{TheorOnEquivGroupExtended}
The extended equivalence transformation group  $\hat G^{\sim}$ for
the class \eqref{eqVarCoefTelegraphEq} is formed by the
transformations
\[
\begin{array}{ll}
\tilde t =\delta_1 t +\delta_2,\quad \tilde x =X(x),\quad
\tilde u =\delta_3 u + \delta_4,\\
\tilde f =\frac{\epsilon_1\delta_1^2}{X_x}f\int e^{-\epsilon_4\int\!
\frac{h}{g}}dx,\quad \tilde g = \epsilon_1\epsilon_2^{-1}X_xg\int
e^{-\epsilon_4\int\! \frac{h}{g}}dx,\quad \tilde h =
\epsilon_1\epsilon_3^{-1}h\int e^{-\epsilon_4\int\! \frac{h}{g}}dx,\\
\tilde H = \epsilon_2 H ,\quad \tilde K=\epsilon_3(K+\epsilon_4H),
\end{array}
\]
where $\delta_j~(j=1,\ldots, 4)$ and $\epsilon_i~(i=1,\ldots,4)$ are
arbitrary constants,
$\delta_1\delta_3\epsilon_1\epsilon_2\epsilon_3\neq 0$, $X$ is an
arbitrary smooth function of $x$, $X_x\neq 0$, $\int\!
\frac{h}{g}=\int \frac{h(x)}{g(x)} dx$.
\end{theorem}

\begin{remark}
It should be noted that the existence of such equivalence transformations
can be explained in many respects by features of representation of equations in the form
\eqref{eqVarCoefTelegraphEq}. This form usually leads to an ambiguity because
the same equation has an infinite series of different
representations. In fact, two representations
\eqref{eqVarCoefTelegraphEq} with the arbitrary element tuples $(f,
g, h, H, K)$ and $(\tilde f, \tilde g, \tilde h, \tilde H, \tilde
K)$ determine the same equation if and only if
\begin{equation}\label{TransofDepentVa}
\begin{array}{ll}
\tilde f =\frac{\epsilon_1\delta_1^2}{X_x}f\int e^{-\epsilon_4\int\!
\frac{h}{g}}dx,\quad \tilde g = \epsilon_1\epsilon_2^{-1}X_xg\int
e^{-\epsilon_4\int\! \frac{h}{g}}dx,\quad \tilde h =
\epsilon_1\epsilon_3^{-1}h\int e^{-\epsilon_4\int\! \frac{h}{g}}dx,\\
\tilde H = \epsilon_2 H ,\quad \tilde K=\epsilon_3(K+\epsilon_4H),
\end{array}
\end{equation}
where $\delta_1$ and $\epsilon_i~(i=1,\ldots,4)$ are arbitrary constants,
$\epsilon_1\epsilon_2\epsilon_3\neq 0$ (the variables $t, x$ and $u$
do not transform!). The transformations \eqref{TransofDepentVa} act only on arbitrary
elements and do not really change equations. In general,
transformations of such type can be considered as trivial \cite{Lisle1992}
(``gauge'') equivalence transformations and form the ``gauge''
(normal) subgroup $\hat G^{\sim g}$ of the extended equivalence
group $\hat G^{\sim}$. Application of ``gauge'' equivalence
transformations is equivalent to rewrite equations in another form.
In spite of really equivalence transformations, their role in group
classification comes not to choice of representatives in equivalence
classes but to choice of form of these representatives.
\end{remark}

\begin{remark}
We note that transformations \eqref{TransofDepentVa} with
$\epsilon_4\neq 0$ are nonlocal with respect to arbitrary elements,
otherwise they belong to  $G^{\sim}$ and form the ``gauge'' (normal)
subgroup $G^{\sim g}$ of the equivalence group $G^{\sim}$.
The factor-group $\hat G^{\sim}/\hat G^{\sim g}$ for class
\eqref{eqVarCoefTelegraphEq} coincides  with $G^{\sim}/G^{\sim g}$ and can be
assumed to consist of the transformations
\begin{equation}\label{FactorTransofDepentVa}
\begin{array}{ll}
\tilde t =\delta_1 t +\delta_2,\quad \tilde x =X(x),\quad
\tilde u =\delta_3 u + \delta_4,\\
\tilde f =\frac{\delta_1^2}{X_x}f,\quad \tilde g = X_xg,\quad \tilde
h = h, \tilde H =H ,\quad \tilde K=K,
\end{array}
\end{equation}
where $\delta_j~(j=1,\ldots, 4)$ are arbitrary constants,
$\delta_1\delta_3\neq 0$, $X$ is an arbitrary smooth function of
$x$, $X_x\neq 0$.
\end{remark}

Based on the this idea, we can gauge the parameter-function $g$ in equation \eqref{eqVarCoefTelegraphEq} to $1$.
More exactly, using theorem \ref{TheorOnEquivGrouUsaual}, we deduce that
the transformation
\[
\tilde t = t, \tilde x =\int \frac{dx}{g(x)},
\tilde u = u
\]
from $G^{\sim}/G^{\sim g}$ reduce equation
\eqref{eqVarCoefTelegraphEq} to
\[
\tilde f(\tilde x)\tilde u_{\tilde t\tilde t}=(H(\tilde u)\tilde
u_{\tilde x})_{\tilde x}+\tilde h(\tilde x)K(\tilde u)\tilde
u_{\tilde x},
\]
where $\tilde f(\tilde x) = g(x)f(x), \tilde g(\tilde x) = 1$ and
$\tilde h(\tilde x) = h(x)$. (Likewise any equation of form \eqref{eqVarCoefTelegraphEq} can
be reduced to the same form with $\tilde f(\tilde x) = 1$). That is
why, without loss of generality we can restrict ourselves to
investigation of the equation
\begin{equation}\label{g=1eqVarCoefTelegraphEq}
f(x)u_{tt}=(H(u)u_x)_x+h(x)K(u)u_x,
\end{equation}
Below, we denote $G_1^{\sim}$ and $\hat G_1^{\sim}$ as the usual and
extended equivalence transformation group of equation
\eqref{g=1eqVarCoefTelegraphEq}.

\begin{theorem}\label{TheorOnEquivGroupUsualg=1}
The usual equivalence transformation group $G_1^{\sim}$ for
class~\eqref{g=1eqVarCoefTelegraphEq} consists of the
transformations
\[
\begin{array}{ll}
\tilde t =\epsilon_4 t +\epsilon_1,\quad \tilde x =\epsilon_5 x
+\epsilon_2,\quad
\tilde u =\epsilon_6 u + \epsilon_3,\\
\tilde f = \epsilon_4^2\epsilon_5^{-2}\epsilon_7 f ,\quad \tilde h=
\epsilon_8^{-1} h, \quad \tilde H = \epsilon_7 H ,\quad \tilde K=
\epsilon_5^{-1}\epsilon_7\epsilon_8K,
\end{array}
\]
where $\epsilon_1, \ldots\epsilon_8$ are arbitrary constants and
$\epsilon_4\epsilon_5\epsilon_6\epsilon_7\epsilon_8\neq0$.
\end{theorem}

Note that for class~\eqref{g=1eqVarCoefTelegraphEq} there also
exists a non-trivial group of discrete equivalence transformations
generated by four involutive transformations of alternating sign in
the sets $\{t\}, \{x,K\},$ $ \{u\}$, $\{f,H,K\}$ and $\{h,K\}$.
Class~\eqref{g=1eqVarCoefTelegraphEq} admits other equivalence
transformations being nonlocal with respect to arbitrary elements
which do not belong to~$G_1^{\sim}$. We demand for these
transformations to be point with respect to $(t,x,u)$. In such way,
using the direct method we can find a generalized equivalence
group~$\hat G^{\sim}$ of class~\eqref{eqVarCoefTelegraphEq}.

\begin{theorem}\label{TheorOnEquivGroupExtendg=1}
The extended equivalence transformation group $\hat G_1^{\sim}$ for
class~\eqref{g=1eqVarCoefTelegraphEq} is formed by the
transformations
\begin{equation} \label{EquivTransformationsGeng=1}\arraycolsep=0em
\begin{array}{l}
\tilde t =\epsilon_4 t +\epsilon_1,\quad \tilde x =\epsilon_5 \int
e^{\epsilon_9\int\! h}dx +\epsilon_2,\quad
\tilde u =\epsilon_6 u + \epsilon_3,\\
\tilde f = \epsilon_4^2\epsilon_5^{-1}\epsilon_7 f \int
e^{-2\epsilon_9\int\! h}dx, \quad \tilde h =
\epsilon_7\epsilon_8^{-1} h\int e^{-\epsilon_9\int\! h}dx,\quad
\tilde H = \epsilon_5\epsilon_7 H ,\quad \tilde K=
\epsilon_8(K+\epsilon_9H),
\end{array}
\end{equation}
where $\epsilon_1, \ldots\epsilon_9$ are arbitrary constants and
$\epsilon_1\epsilon_3\epsilon_5\epsilon_7\epsilon_9\neq0$, $\int\!
h=\int\! h(x)\,dx$.
\end{theorem}
The group $\hat G_1^{\sim}$ is a subgroup of $\hat G^{\sim}$. Its
transformation can be considered as from $\hat G^{\sim}$, which
preserves the condition $g = 1$. The
transformations~\eqref{EquivTransformationsGeng=1} with
non-vanishing values of the parameter $\epsilon_9$ are also nonlocal
in the arbitrary element~$h$. There exists a way to avoid operations
with nonlocal equivalence transformations. More exactly, we can
assume that the parameter-function $K$ is determined up to an
additive term proportional to $H$ and subtract such term from $K$
before applying equivalence transformations
\eqref{FactorTransofDepentVa}.

At the same time, there is another possible generalization of the
gauge $g=1$ to the general case of $h$, namely
the gauge $g=h$. Any equation of the form
\eqref{eqVarCoefTelegraphEq} can be reduced to the equation
\begin{equation}\label{g=heqVarCoefTelegraphEq}
f(x)u_{tt}=(h(x)H(u)u_x)_x+h(x)K(u)u_x.
\end{equation} by the transformation $\tilde t = t, \tilde x =\int \frac{f(x)}{g(x)}dx,
\tilde u = u$ from $G^{\sim}/G^{\sim g}$.

\begin{theorem}\label{TheorOnEquivGroupUsualg=h}
The usual equivalence transformation group $G_h^{\sim}$ for
class~\eqref{g=heqVarCoefTelegraphEq} consists of the
transformations
\[
\begin{array}{ll}
\tilde t =\epsilon_1 t +\epsilon_2,\quad \tilde x =\epsilon_5 x
+\epsilon_6,\quad
\tilde u =\epsilon_3 u + \epsilon_4,\\
\tilde f = \epsilon_1^2\epsilon_5^{-2}\epsilon_7 \epsilon_8f ,\quad
\tilde h= \epsilon_8 h, \quad \tilde H = \epsilon_7 H ,\quad \tilde
K= \epsilon_5^{-1}\epsilon_7K,
\end{array}
\]
where $\epsilon_1, \ldots\epsilon_8$ are arbitrary constants and
$\epsilon_4\epsilon_5\epsilon_6\epsilon_7\epsilon_8\neq0$.
\end{theorem}

\begin{theorem}\label{TheorOnEquivGroupExtendg=h}
The extended equivalence transformation group $\hat G_h^{\sim}$ for
class~\eqref{g=heqVarCoefTelegraphEq} is formed by the
transformations
\begin{equation} \label{EquivTransformationsGeng=h}\arraycolsep=0em
\begin{array}{l}
\tilde t =\epsilon_1 t +\epsilon_2,\quad \tilde x =\epsilon_5  x
+\epsilon_6,\quad
\tilde u =\epsilon_3 u + \epsilon_4,\\
\tilde f = \epsilon_1^2\epsilon_5^{-1}\epsilon_9 f e^{\epsilon_8x},
\quad \tilde h = \epsilon_9\epsilon_7^{-1} h e^{\epsilon_8x},\quad
\tilde H = \epsilon_5 H ,\quad \tilde K= \epsilon_7(K-\epsilon_8H),
\end{array}
\end{equation}
where $\epsilon_1, \ldots\epsilon_9$ are arbitrary constants and
$\epsilon_1\epsilon_3\epsilon_5\epsilon_7\epsilon_9\neq0$, $\int\!
h=\int\! h(x)\,dx$.
\end{theorem}

\begin{remark}
If $H=0$, we assume $h=1$ for determinacy.
\end{remark}

\section{Group Classification of nonlinear telegraph equations}
\label{SectionOnGrClasRes}

In this section, we will present the group classification for the class
\eqref{eqVarCoefTelegraphEq} with the gauges $g=1$ and $g=h$  under the extended
equivalence transformations group $\hat G^{\sim}$.

Following the algorithm in
\cite{Ovsiannikov1982,Akhatov&Gazizov&Ibragimov1987,Akhatov&Gazizov&Ibragimov1989,
Popovych&Ivanova2004NVCDCEs,Ivanova&Popovych&Sophocleous2007}
we are looking for an infinitesimal operator in the form
\begin{equation}\label{OperatorLieSym}
Q=\tau(t, x, u)\partial_t+\xi(t, x, u)\partial_x+\eta(t,
x,u)\partial_u
\end{equation}
which corresponds to a one-parameter Lie group of local
transformation and keep the equation~\eqref{eqVarCoefTelegraphEq}
invariant. The classical infinitesimal Lie invariance criterion for
equation~\eqref{eqVarCoefTelegraphEq} to be invariant with respect
to the operator~\eqref{OperatorLieSym} read as
\begin{equation}\label{sysDetEqTelEqGenForm}
\pr^{(2)}Q(\triangle)\mid_{\triangle=0}=0,\qquad
 \triangle=f(x)u_{tt}-(g(x)H(u)u_x)_x-h(x)K(u)u_x.
\end{equation}
Here $\pr^{(2)}Q$ is the usual second order
prolongation~\cite{Olver1986,Ovsiannikov1982} of the
operator~\eqref{OperatorLieSym}. Substituting the coefficients
of~$\pr^{(2)}Q$ into~\eqref{sysDetEqTelEqGenForm} yields the
following determining equations for $\tau$, $\xi$ and $\eta$:
\begin{equation}\label{sysDetEqTelEqComp}
\begin{array}{ll}
\tau_x=\tau_u=\xi_t=\xi_u=\eta_{uu}=0,\\
H(g\eta_x)_x+hK\eta_x-f\eta_{tt}=0,\\
\frac{f_x}{f}H\xi-\frac{g_x}{g}H\xi-2\tau_tH-\eta
H_u+2H\xi_x=0,\\
(g_x\eta+2g\eta_x)H_u+[(2\eta_{xu}-\xi_{xx})g+(2\tau_t-\xi_x-\xi\frac{f_x}{f})g_x+\xi
g_{xx}]H+h\eta K_u\\
+(2\tau_t-\xi_x-\xi\frac{f_x}{f}+\xi\frac{h_x}{h})hK=0,\\
2\eta_{tu}-\tau_{tt}=0,\\
2(\xi_x-\eta_u)H_u-2\tau_tH_u-\eta
H_{uu}+\frac{f_x}{f}H_u\xi-\frac{g_x}{g}H_u\xi+\eta_uH_u=0.
\end{array}
\end{equation}
Investigating the compatibility of system~\eqref{sysDetEqTelEqComp}
we find that the final equation of system~\eqref{sysDetEqTelEqComp}
is an identity (substituting the third equation of
system~\eqref{sysDetEqTelEqComp} to the final one can yield this
conclusion). With this condition, system~\eqref{sysDetEqTelEqComp}
can be rewritten in the form
\begin{gather}
\tau_x=\tau_u=\xi_t=\xi_u=\eta_{uu}=0,\quad
2\eta_{tu}=\tau_{tt},\label{sysDetEqTelEq_1}\\
2(\xi_x-\tau_t)+(\frac{f_x}{f}-\frac{g_x}{g})\xi=\frac{H_u}{H}\eta,\label{sysDetEqTelEq_2}\\
(g\eta_x)_xH+hK\eta_x-f\eta_{tt}=0,\label{sysDetEqTelEq_3}\\
(g_x\eta+2g\eta_x)H_u+[(2\eta_{xu}-\xi_{xx})g+(2\tau_t-\xi_x-\xi\frac{f_x}{f})g_x+\xi
g_{xx}]H+h\eta K_u  \nonumber \\
+(2\tau_t-\xi_x-\xi\frac{f_x}{f}+\xi\frac{h_x}{h})hK=0.\label{sysDetEqTelEq_4}
\end{gather}
Equations~\eqref{sysDetEqTelEq_1} do not contain arbitrary elements.
Integration of them yields
\begin{equation}\label{CoefsSymIntegr}
\tau =\tau(t),\quad \xi=\xi(x),\quad \eta= \eta^1(t, x)u + \eta^0(t,
x), \quad \eta^1(t, x)=\frac{1}{2}\tau_t+\alpha(x).
\end{equation}

Thus, group classification of~\eqref{eqVarCoefTelegraphEq} reduces
to solving classifying
conditions~\eqref{sysDetEqTelEq_2}--\eqref{sysDetEqTelEq_4}.

Splitting system~\eqref{sysDetEqTelEq_2}--\eqref{sysDetEqTelEq_4}
with respect to the arbitrary elements and their non-vanishing
derivatives gives the equations $\tau_t=0$, $\xi = 0$, $\eta= 0$ for
the coefficients of the operators from~$A^{\ker}$
of~\eqref{eqVarCoefTelegraphEq}. As a result, we obtain the following assertion.

\subsection{Classification under the gauge $g=1$}
\label{subSectionOnGrClasRes1}

\begin{theorem}\label{TheorOnKernelg=1}
The Lie algebra of the kernel of principal groups
of~\eqref{eqVarCoefTelegraphEq} with the gauge $g=1$ is
$A^{\ker}=\langle\partial_t\rangle$.
\end{theorem}

\begin{theorem}\label{TheorOnGrClasResg=1}
A complete set of inequivalent
equations~\eqref{eqVarCoefTelegraphEq} with the gauge $g=1$ with
respect to the transformations from $\hat G_1^{\sim}$ with $A^{\max}
\neq A^{\ker}$ is exhausted by cases given in
tables~\ref{TableGrClasForAllHg=1}--\ref{TableGrClasHpowerg=1}.
\end{theorem}

\setcounter{tbn}{0}

\begin{center}\footnotesize\renewcommand{\arraystretch}{1.15}
Table~\refstepcounter{table}\label{TableGrClasForAllHg=1}\thetable. Case of $\forall H(u)$ (gauge $g=1$) \\[1ex]
\begin{tabular}{|l|c|c|c|l|}
\hline
N & $K(u)$ & $f(x)$ & $h(x)$ & \hfil Basis of A$^{\max}$ \\
\hline
\refstepcounter{tbn}\label{CaseForAllHForAllKForallF}\thetbn & $\forall$ & $\forall$ & $\forall$ & $\p_t$ \\
\refstepcounter{tbn}\label{CaseForAllHForAllKFexp1}\thetbn a
&$\forall$ & $e^{px}$ & $1$ & $\partial_t,\, pt\partial_t+2\partial_x$ \\
\thetbn a'
&$1$ & $|x|^p$ & $x^{-1}$ & $\partial_t,\, (p+2)t\partial_t+2x\partial_x$ \\
\refstepcounter{tbn}\label{CaseForAllHForAllKFexp11}\thetbn
&$1$ & $x^{-2}$ & $x^{-1}$ & $\partial_t, t\partial_t+x\partial_x$ \\
\refstepcounter{tbn}\label{CaseForAllHForAllKFexp111}\thetbn
&$0$ & $1$ & $1$ & $\partial_t,\, \partial_x, t\partial_t+x\partial_x$ \\
\hline
\end{tabular}
\end{center}
{\footnotesize Here $p\in\{0,1\}$ mod $G_1^{\sim}$ in case
\ref{CaseForAllHForAllKFexp1}, $p \neq -2$ case
\ref{CaseForAllHForAllKFexp1}$a'$.\\
Additional equivalence transformations:\\ \setcounter{casetran}{0}
\refstepcounter{casetran}\thecasetran.
\ref{CaseForAllHForAllKFexp1}a $(p=0, K=1)$ $\to$
\ref{CaseForAllHForAllKFexp1}a $(p=0, K=0)$:\quad
$\tilde t=t$, $\tilde x=x+t$, $\tilde u=u;$\\
\refstepcounter{casetran}\thecasetran.
\ref{CaseForAllHForAllKFexp1}a' $(p=-2, K=1)$ $\to$
\ref{CaseForAllHForAllKFexp1}a' $(p=-2, K=0)$:\quad $\tilde t=t$,
$\tilde x=xe^t$, $\tilde u=u.$ }

\bigskip

\setcounter{tbn}{0}

\begin{center}\footnotesize\renewcommand{\arraystretch}{1.15}
Table~\refstepcounter{table}\label{TableGrClasHexpg=1}\thetable. Case of $H(u)=e^{\mu u}$ (gauge $g=1$) \\[1ex]
\begin{tabular}{|l|c|c|c|c|l|}
\hline
N & $K(u)$ & $f(x)$ & $h(x)$ & \hfil Basis of A$^{\max}$ \\
\hline \refstepcounter{tbn}\label{CaseHexpKexpnuuFxlambda}\thetbn &
 $0$ & $\forall$ & $1$ &
$\partial_t,\, t\partial_t -2\partial_u $ \\
\refstepcounter{tbn}\label{CaseHexpKexpnuuF1}\thetbn & $e^{\nu u}$ &
$|x|^{p}$& $|x|^{q}$ & $\partial_t,\, [p(\mu-\nu)+(\mu-2\nu)-q\mu]t\partial_t+2(\mu-\nu)\partial_x+2(q+1)\partial_u$  \\
\thetbn $^{*}$ & $e^{\nu u}$ &
$e^{px}$& $\epsilon e^{qx}$ & $\partial_t,\, [p(\mu-\nu)-q\mu]t\partial_t+2(\mu-\nu)\partial_x+2q\partial_u$  \\
\refstepcounter{tbn}\label{CaseHexpK1Fxlambda}\thetbn & $ue^u$ &
$h^2e^{q\int h} $ & $(h^{-1})^{''}=-2ph$ & $\partial_t,\,
(2p+q)t\partial_t+2h^{-1}\partial_x-4p\partial_u$  \\
\refstepcounter{tbn}\label{CaseHexpK1F1}\thetbn &
 $e^{\nu u}$ &
$1$ & $1$ & $\partial_t, \partial_x, (\mu-2\nu)t\partial_t+2(\mu-\nu)x\partial_x+2\partial_u$  \\
\refstepcounter{tbn}\label{CaseHexpK0expuForAllF}\thetbn & $0$ &
$f^{1}(x)$ & $1 $ & $\partial_t, t\partial_t-2\partial_u, \alpha t\partial_t+2(\beta x^2+\gamma_1x+\gamma_0)\partial_x
+2\beta x\partial_u$  \\
\refstepcounter{tbn}\label{CaseHexpK0Ff1}\thetbn a  & $0$ & $1$ &
$1$ & $\partial_t,\partial_x, t\partial_t-2\partial_u,
t\partial_t+x\partial_x$  \\
\thetbn b  & $0$ & $x^{-3}$ & $1$ & $\partial_t,
x\partial_x-\partial_u, x^2\partial_x+x\partial_u,
t\partial_t-2\partial_u$  \\
\hline
\end{tabular}
\end{center}
{\footnotesize Here $(\mu,~\nu)\in\{(0,1),(1,\nu)\}, \nu\neq\mu$ in
cases \ref{CaseHexpKexpnuuF1}, \ref{CaseHexpKexpnuuF1}$^{*}$ and
\ref{CaseHexpK1F1}; $\mu = 1$ and $\nu \neq 1$ in the other cases;
$q\neq-1$ in case \ref{CaseHexpKexpnuuF1}$^{*}$ (otherwise it is
subcase of the case 1.2a); $\epsilon=\pm 1$ in case
\ref{CaseHexpKexpnuuF1}$^{*}$; $\alpha, \beta, \gamma_1,
\gamma_0=$const and
\[
f^{1}(x)=\exp\left\{\int\frac{-3\beta x-2\gamma_1+\alpha}{\beta
x^2+\gamma_1 x+\gamma_0}dx\right\}
\]
Additional equivalence transformations:\\ \setcounter{casetran}{0}
\refstepcounter{casetran}\thecasetran. \ref{CaseHexpK0Ff1}b  $\to$
\ref{CaseHexpK0Ff1}a :\quad $\tilde t=t\sign x$, $\tilde x=1/x$,
$\tilde u=u-\ln|x|.$ }

\bigskip

\setcounter{tbn}{0}

\begin{center}\footnotesize\renewcommand{\arraystretch}{1.15}
Table~\refstepcounter{table}\label{TableGrClasHpowerg=1}\thetable. Case of $H(u)=u^{\mu}$ (gauge $g=1$) \\[1ex]
\begin{tabular}{|l|c|c|c|c|l|}
\hline
N & $\mu$ & $K(u)$ & $f(x)$ & $h(x)$  & \hfil Basis of A$^{\max}$ \\
\hline \refstepcounter{tbn}\label{CaseHpowerKpowerFpower}\thetbn &
$\neq -4$ & $0$ &
$\forall$ & $1 $ & $\partial_t,\, \mu t\partial_t-2u\partial_u$  \\
\refstepcounter{tbn}\label{CaseHpowerKpowerF1}\thetbn & $\forall$ &
$|u|^{\nu }$ & $|x|^p$ & $|x|^q$ & $\partial_t,\,
[(p-q)\mu-p\nu+\mu-2\nu]t\partial_t$\\
&$~$& $~$ & $~$ & $~$ & $+2(\mu-\nu)x\partial_x+2(q+1)u\partial_u$\\
\thetbn $^*$ & $\forall$ & $|u|^{\nu }$ & $e^{px}$ & $\epsilon
e^{qx}$ & $\partial_t,\,
[(p-q)\mu-q\nu]\partial_t$\\
&$~$& $~$ & $~$ & $~$ & $+2(\mu-\nu)\partial_x+2qu\partial_u$\\
\refstepcounter{tbn}\label{CaseHpowerK1Fpower} \thetbn & $\forall$ &
$|u|^{\mu}\ln|u|$ & $h^2e^{q\int h dx}$ & $(h^{-1})''=-2ph$  &
$\partial_t,
(2p\mu+q) t\partial_t+2h^{-1}\partial_x-4pu\partial_u$  \\
\refstepcounter{tbn}\label{CaseHpowerK2Fpower} \thetbn & $0$ &
$\forall$ & $h^2$ & $(h^{-1})''=0$  & $\partial_t,
h^{-1}\partial_x$  \\
\refstepcounter{tbn}\label{CaseHpowerK0F1}\thetbn & $0$ & $u$ &
$h^2e^{\int h}$& $(h^{-1})''=-2ph$
& $\partial_t, t\partial_t+2h^{-1}\partial_x-4p\partial_u$  \\
\refstepcounter{tbn}\label{CaseHu-4K0HForallF} \thetbn & $\forall$ &
$|u|^{\nu}$ & $1$ & $1$
& $\partial_t,\,\partial_x,(\mu-2\nu) t\partial_t+2(\mu-\nu)x\partial_x+2u\partial_u$  \\
\refstepcounter{tbn}\label{CaseHu-4K0Ff3} \thetbn & $\neq-4$ & $0$ &
$f^3(x)$ & $1$ & $\partial_t,
\mu t\partial_t-2u\partial_u,$\\
&  & & & & $ \alpha t\partial_t+2[(\mu+1)\beta x^2+\gamma_1x+\gamma_0]\partial_x+2\beta xu\partial_u $  \\
\refstepcounter{tbn}\label{CaseHu-4K0F1} \thetbn a & $\neq-4,
-\frac{4}{3}$ & $0$ & $1$ & $1$ & $\partial_t,
\mu t\partial_t-2u\partial_u, \partial_x, t\partial_t+x\partial_x$  \\
\thetbn b & $\neq-4, -\frac{4}{3}, -1$ & $0$ &
$|x|^{-\frac{3\mu+4}{\mu+1}}$ & $1$ & $\partial_t, \mu
t\partial_t-2u\partial_u, (\mu+2)t\partial_t-2(\mu+1)x\partial_x,$\\
&  & & & & $ (\mu+1)x^2\partial_t+xu\partial_u$  \\
\thetbn c & $-1$ & $0$ & $e^x$ & $1$ & $\partial_t,
t\partial_t+2u\partial_u, \partial_x-u\partial_u,$\\
&  & & & & $ t\partial_t+x\partial_x-xu\partial_u$  \\
\refstepcounter{tbn}\label{CaseH1-43K1} \thetbn & $-4$ & $0$
&$f^3(x)|_{\mu=-4}$ & $1$ & $\partial_t, 2t\partial_t+u\partial_u, t^2\partial_t+tu\partial_u$  \\
&  & & & & $\alpha t\partial_t+2(-3\beta x^2+\gamma_1x+\gamma_0)\partial_x+2\beta xu\partial_u$  \\
\refstepcounter{tbn}\label{CaseHu-43K0F3} \thetbn & $-4$ & $0$ & $1$
& $1$ & $\partial_t,
2t\partial_t+u\partial_u, \partial_x, t^2\partial_t+tu\partial_u,$\\
&  & & & & $ 2x\partial_x-u\partial_u$  \\
\refstepcounter{tbn}\label{CaseHu-43K0F2} \thetbn & $-\frac{4}{3}$ &
$0$ & $1$ & $1$ & $\partial_t,
2t\partial_t+3u\partial_u, \partial_x, t\partial_t+x\partial_x,$\\
&  & & & & $ x^2\partial_x-3xu\partial_u$  \\
\hline
\end{tabular}
\end{center}
{\footnotesize Here $\nu\neq \mu$; $\epsilon=\pm 1$ in case
\ref{CaseHpowerKpowerF1}. (otherwise it is subcase of the case
1.2a); $\alpha, \beta, \gamma_1, \gamma_0=$const, and
\[
f^{3}(x)=\exp\left\{\int\frac{-(3\mu+4)\beta
x-2\gamma_1+\alpha}{(\mu+1)\beta x^2+\gamma_1 x+\gamma_0}dx\right\}
\]
Additional equivalence transformations:\\ \setcounter{casetran}{0}
\refstepcounter{casetran}\thecasetran. \ref{CaseHu-4K0F1}b $\to$
\ref{CaseHu-4K0F1}a :\quad
$\tilde t=t$, $\tilde x=-1/x$, $\tilde u=|x|^{-\frac{1}{1+\mu}}u;$\\
\refstepcounter{casetran}\thecasetran. \ref{CaseHu-4K0F1}c
 $\to$ \ref{CaseHu-4K0F1}a $(\mu=-1)$:\quad $\tilde t=t$, $\tilde
x=x$, $\tilde u=e^xu.$
 }

\begin{remark}
Tables \ref{TableGrClasForAllHg=1}-\ref{TableGrClasHpowerg=1} are the
results of classification for class \eqref{eqVarCoefTelegraphEq}
with the gauge $g=1$ with respect to the extended equivalence
transformations group $\hat G_1^{\sim}$. The classification for
class \eqref{eqVarCoefTelegraphEq} with the gauge $g=1$ with respect
to the usual equivalence transformations group $G_1^{\sim}$ are very
complicated. In some cases, the determining equations can not be
solved explicitly. Therefore, We list these results as an Appendix.
\end{remark}

\begin{remark}
The proof of theorem
\ref{TheorOnGrClasResg=1} follows directly from the analysis of the
Section \ref{SectionOnGrClasProof}.
\end{remark}

\subsection{Classification under the gauge $g=h$}
\label{SectionOnGrClasResg=h}

\begin{theorem}\label{TheorOnKernel}
The Lie algebra of the kernel of principal groups
of~\eqref{eqVarCoefTelegraphEq} with the gauge $g=h$  is
$A^{\ker}=\langle\partial_t\rangle$.
\end{theorem}

\begin{theorem}\label{TheorOnGrClasRes}
A complete set of inequivalent
equations~\eqref{eqVarCoefTelegraphEq} with the gauge $g=h$ with
respect to the transformations from $\hat G^{\sim}$ with $A^{\max}
\neq A^{\ker}$ is exhausted by cases given in
tables~\ref{TableGrClasForAllH}--\ref{TableGrClasHpower}.
\end{theorem}

\setcounter{tbn}{0}

\begin{center}\footnotesize\renewcommand{\arraystretch}{1.15}
Table~\refstepcounter{table}\label{TableGrClasForAllH}\thetable. Case of $\forall H(u)$ (gauge $g=h$) \\[1ex]
\begin{tabular}{|l|c|c|c|l|}
\hline
N & $K(u)$ & $f(x)$ & $h(x)$ & \hfil Basis of A$^{\max}$ \\
\hline
\refstepcounter{tbn}\label{CaseForAllHForAllKForallF}\thetbn & $\forall$ & $\forall$ & $\forall$ & $\p_t$ \\
\refstepcounter{tbn}\label{CaseForAllHForAllKFexp1}\thetbn a
&$\forall$ & $e^{px}$ & $1$ & $\partial_t,\, pt\partial_t+2\partial_x$ \\
\refstepcounter{tbn}\label{CaseForAllHForAllKFexp11}\thetbn
&$1$ & $x^{-2}$ & $x^{-1}$ & $\partial_t, t\partial_t+x\partial_x$ \\
\refstepcounter{tbn}\label{CaseForAllHForAllKFexp11}\thetbn
&$0$ & $1$ & $1$ & $\partial_t,\, \partial_x, t\partial_t+x\partial_x$ \\
\hline
\end{tabular}
\end{center}
{\footnotesize Here $p\in\{0,1\}$ mod $G_h^{\sim}$ in case
\ref{CaseForAllHForAllKFexp1}. }

\bigskip

\setcounter{tbn}{0}

\begin{center}\footnotesize\renewcommand{\arraystretch}{1.15}
Table~\refstepcounter{table}\label{TableGrClasHexp}\thetable. Case of $H(u)=e^{\mu u}$ (gauge $g=h$) \\[1ex]
\begin{tabular}{|l|c|c|c|c|l|}
\hline
N & $K(u)$ & $f(x)$ & $h(x)$ & \hfil Basis of A$^{\max}$ \\
\hline \refstepcounter{tbn}\label{CaseHexpKexpnuuFxlambdag=h}\thetbn
&
 $0$ & $\forall$ & $1$ &
$\partial_t,\, t\partial_t -2\partial_u $ \\
\refstepcounter{tbn}\label{CaseHexpKexpnuuF1g=h}\thetbn & $e^{\nu
u}$ &
$|x|^{p}$& $|x|^{q}$ & $\partial_t,\, [(p-q+1)\mu-(p-q+2)\nu]t\partial_t+2(\mu-\nu)x\partial_x+2\partial_u$  \\
\refstepcounter{tbn}\label{CaseHexpK1Fxlambdag=h}\thetbn & $ue^u$ &
$e^{px^2+qx} $ & $e^{px^2}$ & $\partial_t,\,
(2p+q)t\partial_t+2\partial_x-4p\partial_u$  \\
\refstepcounter{tbn}\label{CaseHexpK1F1g=h}\thetbn &
 $e^{\nu u}$ &
$1$ & $1$ & $\partial_t, \partial_x, (\mu-2\nu)t\partial_t+2(\mu-\nu)x\partial_x+2\partial_u$  \\
\refstepcounter{tbn}\label{CaseHexpK0expuForAllFg=h}\thetbn & $0$ &
$f^{1}(x)$ & $1 $ & $\partial_t, t\partial_t-2\partial_u, \alpha
t\partial_t+2(\beta x^2+\gamma_1x+\gamma_0)\partial_x
+2\beta x\partial_u$  \\
\refstepcounter{tbn}\label{CaseHexpK0Ff1g=h}\thetbn a  & $0$ & $1$ &
$1$ & $\partial_t,\partial_x, t\partial_t-2\partial_u,
t\partial_t+x\partial_x$  \\
\thetbn b  & $0$ & $x^{-3}$ & $1$ & $\partial_t,
x\partial_x-\partial_u, x^2\partial_x+x\partial_u,
t\partial_t-2\partial_u$  \\
\hline
\end{tabular}
\end{center}
{\footnotesize Here $(\mu,~\nu)\in\{(0,1),(1,\nu)\}, \nu\neq\mu$ in
cases \ref{CaseHexpKexpnuuF1} and \ref{CaseHexpK1F1}; $\mu = 1$ and
$\nu \neq 1$ in the other cases.
}

\bigskip

\setcounter{tbn}{0}

\begin{center}\footnotesize\renewcommand{\arraystretch}{1.15}
Table~\refstepcounter{table}\label{TableGrClasHpower}\thetable. Case of $H(u)=u^{\mu}$ (gauge $g=h$) \\[1ex]
\begin{tabular}{|l|c|c|c|c|l|}
\hline
N & $\mu$ & $K(u)$ & $f(x)$ & $h(x)$  & \hfil Basis of A$^{\max}$ \\
\hline \refstepcounter{tbn}\label{CaseHpowerKpowerFpowerg=h}\thetbn
& $\neq -4$ & $0$ &
$\forall$ & $1 $ & $\partial_t,\, \mu t\partial_t-2u\partial_u$  \\
\refstepcounter{tbn}\label{CaseHpowerKpowerF1g=h}\thetbn & $\forall$
& $|u|^{\nu }$ & $|x|^p$ & $|x|^q$ & $\partial_t,\,
[(p-q+1)\mu-(p-q+2)\nu]t\partial_t$\\
&$~$& $~$ & $~$ & $~$ & $+2(\mu-\nu)x\partial_x+2u\partial_u$\\
\refstepcounter{tbn}\label{CaseHpowerK1Fpowerg=h} \thetbn &
$\forall$ & $|u|^{\mu}\ln|u|$ & $e^{px^2+qx}$ & $e^{px^2}$  &
$\partial_t,
(2p\mu+q) t\partial_t+2\partial_x-4pu\partial_u$  \\
\refstepcounter{tbn}\label{CaseHpowerK0F1}\thetbn & $0$ & $u$ &
$h^2e^{\int h}$& $(h^{-1})''=-2ph$
& $\partial_t, t\partial_t+2h^{-1}\partial_x-4p\partial_u$  \\
\refstepcounter{tbn}\label{CaseHpowerK2Fpower} \thetbn & $0$ &
$\forall$ & $h^2$ & $(h^{-1})''=0$  & $\partial_t,
h^{-1}\partial_x$  \\
\refstepcounter{tbn}\label{CaseHu-4K0HForallFg=h} \thetbn &
$\forall$ & $|u|^{\nu}$ & $1$ & $1$
& $\partial_t,\,\partial_x,(\mu-2\nu) t\partial_t+2(\mu-\nu)x\partial_x+2u\partial_u$  \\
\refstepcounter{tbn}\label{CaseHu-4K0Ff3g=h} \thetbn & $\neq-4$ &
$0$ & $f^3(x)$ & $1$ & $\partial_t,
\mu t\partial_t-2u\partial_u,$\\
&  & & & & $ \alpha t\partial_t+2[(\mu+1)\beta x^2+\gamma_1x+\gamma_0]\partial_x+2\beta xu\partial_u $  \\
\refstepcounter{tbn}\label{CaseHu-4K0F1g=h} \thetbn a & $\neq-4,
-\frac{4}{3}$ & $0$ & $1$ & $1$ & $\partial_t,
\mu t\partial_t-2u\partial_u, \partial_x, t\partial_t+x\partial_x$  \\
\thetbn b & $\neq-4, -\frac{4}{3}, -1$ & $0$ &
$|x|^{-\frac{3\mu+4}{\mu+1}}$ & $1$ & $\partial_t, \mu
t\partial_t-2u\partial_u, (\mu+2)t\partial_t-2(\mu+1)x\partial_x,$\\
&  & & & & $ (\mu+1)x^2\partial_t+xu\partial_u$  \\
\thetbn c & $-1$ & $0$ & $e^x$ & $1$ & $\partial_t,
t\partial_t+2u\partial_u, \partial_x-u\partial_u,$\\
&  & & & & $ t\partial_t+x\partial_x-xu\partial_u$  \\
\refstepcounter{tbn}\label{CaseH1-43K1g=h} \thetbn & $-4$ & $0$
&$f^3(x)|_{\mu=-4}$ & $1$ & $\partial_t, 2t\partial_t+u\partial_u, t^2\partial_t+tu\partial_u$  \\
&  & & & & $\alpha t\partial_t+2(-3\beta x^2+\gamma_1x+\gamma_0)\partial_x+2\beta xu\partial_u$  \\
\refstepcounter{tbn}\label{CaseHu-43K0F3g=h} \thetbn & $-4$ & $0$ &
$1$ & $1$ & $\partial_t,
2t\partial_t+u\partial_u, \partial_x, t^2\partial_t+tu\partial_u,$\\
&  & & & & $ 2x\partial_x-u\partial_u$  \\
\refstepcounter{tbn}\label{CaseHu-43K0F2g=h} \thetbn &
$-\frac{4}{3}$ & $0$ & $1$ & $1$ & $\partial_t,
2t\partial_t+3u\partial_u, \partial_x, t\partial_t+x\partial_x,$\\
&  & & & & $ x^2\partial_x-3xu\partial_u$  \\
\hline
\end{tabular}
\end{center}
{\footnotesize Here $\nu\neq \mu$.
 }
\\

In tables~\ref{TableGrClasForAllHg=1}--\ref{TableGrClasHpowerg=1}
and \ref{TableGrClasForAllH}--\ref{TableGrClasHpower} we list all
possible $\hat G^{\sim}$-inequivalent sets of functions $f(x)$,
$h(x)$, $H(u)$, $K(u)$ and corresponding invariance algebras under
the gauges g = 1 and g = h respectively. We give the same numbers
for the corresponding ($\hat G^{\sim}$-equivalent) cases in the
gauges $g=1$ and $g=h$. The asterisked cases from tables
\ref{TableGrClasHexpg=1} and \ref{TableGrClasHpowerg=1} are
equivalent to the cases from tables \ref{TableGrClasHexp} and
\ref{TableGrClasHpower} with the same numbers, where the
parameter-function $h$ takes the value $h = x$. The similar
non-asterisked cases correspond to the same cases from tables
\ref{TableGrClasHexp} and \ref{TableGrClasHpower}, where
$p'=\frac{p-q}{q+1}, q'=\frac{q}{q+1}$ or $p'=-\frac{1}{p+2}$ if
$q=p+1$.

In what follows, for convenience we use double numeration $T.N$ of
classification cases, where $T$ denotes the number of the table and
$N$ the number of the case in table $T$. The notation `equation
$T.N$' is used for the equation of the
form~\eqref{eqVarCoefTelegraphEq} where the parameter functions take
the values from the corresponding case.

The operators from
tables~\ref{TableGrClasForAllHg=1}--\ref{TableGrClasHpowerg=1} or
\ref{TableGrClasForAllH}--\ref{TableGrClasHpower} form bases of the
maximal invariance algebras if the corresponding sets of the
functions $f$, $h$, $H$, $K$ are $\hat G^{\sim}$-inequivalent to
ones with most extensive invariance algebras. For
example, in case
\ref{TableGrClasHpowerg=1}.\ref{CaseHpowerKpowerFpower} the adduced
operators have the above property iff $f \neq f^3$.

\begin{remark}
Case \ref{TableGrClasForAllHg=1}.\ref{CaseForAllHForAllKFexp1}a is
equivalent to case
\ref{TableGrClasForAllHg=1}.\ref{CaseForAllHForAllKFexp1}a' with
respect to transformation $ \tilde t=t, \tilde x=\ln |x|, \tilde
u=u, \tilde H=H, \tilde K=K-H, \tilde p=p+2 $ from $\hat G^{\sim}$.
We adduce case
\ref{TableGrClasForAllHg=1}.\ref{CaseForAllHForAllKFexp1}a' here for
the convenience of presentation of results only.
\end{remark}

\subsection{Proof of classification results}\label{SectionOnGrClasProof}

Now, let us use the method of  furcate split
~\cite{Nikitin&Popovych2001,Popovych&Ivanova2004NVCDCEs,Ivanova&Popovych&Sophocleous2007}
to prove the main classification theorems~\ref{TheorOnGrClasResg=1}
and \ref{TheorOnGrClasRes}. It should be noted that it seems
impossible to formulate complete results of group classification of
class \eqref{eqVarCoefTelegraphEq} with respect to usual equivalence
group $G^{\sim}$ in a closed form. This can be seen from the
classifications for the gauge $g=1$ adduced in the Appendix, while
it is quite easy to solve the problem of group classification with
respect to the extended equivalence group $\hat G^{\sim}$.

The basic idea of the method is based on the fact that the
substitution of the coefficients of any operator from the extension
of $A^{\ker}$ into the classifying equations results in nonidentity
equations for arbitrary elements
(see~\cite{Nikitin&Popovych2001,Popovych&Ivanova2004NVCDCEs,Ivanova&Popovych&Sophocleous2007}
for more details about the method). In the problem under
consideration, the procedure of looking for the possible cases
mostly depends on equation~\eqref{sysDetEqTelEq_2}. For any operator
$Q\in A^{\max}$ equation~\eqref{sysDetEqTelEq_2} gives some
equations on $H$ of the general form
\[
(au+b)H_u=cH,
\]
where $a$, $b$, $c$ are constant. In general, for all operators from
$A^{\max}$ the number $k$ of such independent equations is not
greater than $2$; otherwise they form an incompatible system on $H$.
$k$ is an invariant value for the transformations from $\hat
G^{\sim}$. Therefore, there exist three inequivalent cases for the
value of $k$:
\begin{enumerate}
    \item $k = 0:\quad H(u)$ is arbitrary,
    \item $k = 1:\quad H(u)=e^{\mu u}$ or $H(u)=u^{\mu}$ $(\mu\neq 0)$ mod $\hat G^{\sim}$,
    \item $k = 2:\quad H(u)=1$ mod $G^{\sim}$.
\end{enumerate}
Furthermore, in order to provide the final presentation of
classification results in a simple way, the choice of a gauge for
the arbitrary elements is very important for solving the determining
equations. It is more convenient to constrain the parameter-function
$g$ instead of $f$ in class \eqref{eqVarCoefTelegraphEq}. The next
problem is the choice between gauges of $g$. The case $K \bar \in
\langle 1, H \rangle$ and $k \geq 1$ is easier to be investigated in
the gauge $g = h$. In the other cases we obtain results in a simpler
explicit form and in an easier way using the gauge $g = 1$. Let us
consider these possibilities in more detail, omitting cumbersome
calculations.\\
\\
{\bf Case 1: $k=0$} (the gauges $g=1$ and $g=h$,
tables~\ref{TableGrClasForAllHg=1} and \ref{TableGrClasForAllH}). We
first consider the case gauge $g=1$. Since $H(u)$ arbitrary, this
means that the coefficients of any operator from $A^{\max}$ must
satisfy $\eta=0 $ and
\begin{gather}
2(\xi_x-\tau_t)+\frac{f_x}{f}\xi=0,\label{sysDetEqForAllH_1}\\
-K(\xi h)_x+H\xi_{xx}=0.\label{sysDetEqForAllH_2}
\end{gather}

(i) Let us suppose that $K \not\in\langle1, H\rangle $. It follows
from equation~\eqref{sysDetEqForAllH_2} that $\xi_x=0$. Therefore,
equation~\eqref{sysDetEqForAllH_1} must be in the form $f_x=\mu f $
without fail. Solving this equation yields cases
\ref{CaseForAllHForAllKFexp1}a.

(ii) Now let $K \in\langle1, H\rangle $, i.e. $K=\delta$ mod $\hat
G_1^{\sim}$ where $\delta\in \{0, 1\}$. Then
equation~\eqref{sysDetEqForAllH_2} can be decomposed into the
following ones
\begin{equation}\label{sysDetEqForAllH_3}
\xi_{xx}=0, \quad\quad \delta(\xi h)_x=0.
\end{equation}
Integrating of the latter equations up to $\hat G_1^{\sim}$ results
to cases
\ref{CaseForAllHForAllKFexp1}a'--\ref{CaseForAllHForAllKFexp111} of
table \ref{TableGrClasForAllHg=1}.

The classifications for the gauge $g=h$ can be derived in a similar
way.\\
\\
{\bf Case 2: $k=1$} (the gauges $g=1$ and $g=h$, tables~\ref
{TableGrClasHexpg=1}, \ref{TableGrClasHpowerg=1} and \ref
{TableGrClasHexp}, \ref{TableGrClasHpower} ). Here $H\in \{e^{\mu
u}, u^{\mu}, \mu\neq 0\}$ mod $\hat G^{\sim}$ and there exists $Q\in
A^{\max}$ with $\eta\neq 0$, otherwise there is no additional
extension of the maximal Lie invariance algebra in comparison with
the case $k = 0$. Below we consider the gauge $g=h$ in details, the
gauge $g=1$ can be proved in a similar way. If $H= e^{\mu u}$
we assume $\mu=1$.\\

{\it Case 2.1:} Let us investigate the first possibility $H = e^{\mu
u}$ (table~\ref{TableGrClasHexp}). Equations~\eqref{sysDetEqTelEq_2}
and~\eqref{CoefsSymIntegr} imply $\eta_u=0$, i.e. $\eta=\eta^1(t,
x)$ and $\tau_{tt}=0$. Therefore, equation \eqref{sysDetEqTelEq_4}
looks like $K_u=\nu K+\lambda H$ with respect to $K$, where $\nu,
b=\const$, otherwise $\eta \equiv 0$.

Consider first the case $K \bar \in \langle 1, H \rangle$. Under the
above suppositions, equations
\eqref{sysDetEqTelEq_2}--\eqref{sysDetEqTelEq_4} can be rewritten as
\begin{gather}
\frac{\varphi_x}{\varphi}\xi=(2\nu-\mu)\eta^1+2\tau_t,\label{sysDetEqHexp_1}\\
\eta^1_t=\eta^1_x=0,\quad \xi_{xx}=\tau_{tt}=0,\label{sysDetEqHexp_2}\\
\xi_x=(\mu-\nu)\eta^1, \quad
(\xi\frac{h_x}{h})_x=-\lambda\eta^1.\label{sysDetEqHexp_3}
\end{gather}
Here and below $\varphi=f/h$. From equation \eqref{sysDetEqHexp_1}
we can get $\varphi \in \{e^{qx}, |x|^r(r\neq 0), 1\}$ mod $\hat
G_h^{\sim}$.

For $\varphi=e^{qx}$ it follows from the determining equations
\eqref{sysDetEqHexp_1}--\eqref{sysDetEqHexp_3} that $\xi_x=0,
\nu=\mu, \lambda\neq 0, \frac{h_x}{h}=2\alpha$. Thus,
$h=h_0e^{\alpha x^2+h_1x}=e^{\alpha x^2}$ mod $\hat G_h^{\sim}$,
$\alpha \neq 0, f=h\varphi=e^{\alpha x^2+qx}, K=\lambda ue^u$ mod
$\hat G_h^{\sim}$ that falls precisely into case
\ref{TableGrClasHexp}.\ref{CaseHexpK1Fxlambdag=h}.

If $\varphi=|x|^r, r\neq0$, then $r\xi/x =
(2\nu-\mu)\eta^1+2\tau_t$. Therefore, $\xi = (\mu-\nu)\eta^1 x,
(\mu-\nu)(\frac{xh_x}{h})_x=-\lambda$. Since $\mu\neq \nu$
(otherwise, $K\in \langle 1, H\rangle$) we have $\lambda = 0$ mod
$\hat G_h^{\sim}$. Therefore, $h = |x|^q$ mod $\hat G^{\sim}$. Then
$f = |x|^p, p \neq q$, and we obtain case
\ref{TableGrClasHexp}.\ref{CaseHexpKexpnuuF1g=h}.

Value $\varphi=1$ results in $2\tau_t=(\mu-2\nu)\eta^1,
\xi=(\mu-\nu)\eta^1x+\xi_0$. If $\nu=\mu$ then $\lambda \neq 0$
(otherwise, $K\in \langle 1, H\rangle$), $\lambda= 1$ mod $\hat
G_h^{\sim}$ , $(\frac{h_x}{h})_x = 2\alpha(\eta^1=-2\alpha)$.
Therefore, $h=h_0e^{\alpha x^2+h_1x}=e^{\alpha x^2}$ mod $\hat
G_h^{\sim}$ , $K =ue^u$ that follows to case
\ref{TableGrClasHexp}.\ref{CaseHexpK1Fxlambdag=h}. If $\nu\neq\mu$,
then $\lambda = 0$ mod $\hat G_h^{\sim}$ . Therefore, $h \in
\{|x|^q, 1, e^{px}\}$ $\hat G_h^{\sim}$ that yields subcases of
\ref{TableGrClasHexp}.\ref{CaseHexpK1Fxlambdag=h},
\ref{TableGrClasHexp}.\ref{CaseHexpKexpnuuF1g=h} and case
\ref{TableGrClasHexp}.\ref{CaseHexpK1F1g=h} correspondingly.

Now, we consider the case $K \in \langle 1, H \rangle$,
$H\neq$const. In contrast to the previous case, it is more
convenient to consider this case using the gauge $g = 1$. In such
case $K = 0, 1 $ mod  $\hat G_1^{\sim}$ . Application of the above
suppositions reduces the determining equations to the system
\begin{gather*}
2\xi_x+\frac{f_x}{f}\xi=\mu\eta^1+2\tau_t,\quad \eta^1_{xx}=0,\quad K\eta^1_x=\varphi \eta^1_{tt},\\
(\xi_x+\frac{\varphi_x}{\varphi}\xi-\tau_t)K=0, \quad
\xi_{xx}=2\mu\eta^1_x.
\end{gather*}
Note that $\eta^1=\frac{1}{2}\tau_t+\alpha(x)$, thus we have
$\eta^1=\frac{1}{2}\tau_t+\beta x+\alpha_0$ and $\xi=\mu \beta
x^2+\gamma_1x+\gamma_0$. Substituting these values into the first
determining equation we obtain
\[
(\mu \beta x^2+\gamma_1x+\gamma_0)\frac{f_x}{f}=-3\mu\beta
x+\frac{1}{2}(\mu+4)\tau_t+\mu \alpha_0-2\gamma_1.
\]
This equation gives $l$ linearly independent equations for $f$ of
form $(\alpha^2
x^2+\alpha^1x+\alpha^0)\frac{f_x}{f}=\beta_1x+\beta_0$. If $l= 0$,
then $\xi= 0, \beta=0, \frac{1}{2}(\mu+4) \tau_t =-\mu\alpha_0$.
Considering case $l = 1$, we get $(\alpha^2, \beta_1) \neq (0, 0)$,
$(\alpha^0, \beta_0) \neq (0, 0)$, otherwise $l > 1$. At last, if $l
\geq 2$, then $f\in \{1, e^{px}, |x|^p, p \neq 0\}$ mod $\hat
G_1^{\sim}$.

Direct substitution of the above values to the determining equation
for $K = 0$ and obvious integration leads to the cases
\ref{TableGrClasHexp}.\ref{CaseHexpKexpnuuFxlambdag=h} (case $l =
0$), \ref{TableGrClasHexp}.\ref{CaseHexpK0expuForAllFg=h} (case $l =
1$) and \ref{TableGrClasHexp}.\ref{CaseHexpK0Ff1g=h}a,
\ref{TableGrClasHexp}.\ref{CaseHexpK0Ff1g=h}b  (case $l = 2$).

Classification in case $K = 1$ is more cumbersome, and corresponding
results can be reduced to cases
\ref{TableGrClasHexp}.\ref{CaseHexpK1Fxlambdag=h},
\ref{TableGrClasHexp}.\ref{CaseHexpKexpnuuF1g=h} and
\ref{TableGrClasHexp}.\ref{CaseHexpK1F1g=h}.\\

{\it Case 2.2: }  Consider the case $H =u^{\mu}$
(table~\ref{TableGrClasHpower}). Equations~\eqref{sysDetEqTelEq_2}
and~\eqref{CoefsSymIntegr} imply
$\eta=(\frac{1}{2}\tau_{t}+\alpha(x))u=\eta^1(t, x)u$. Therefore,
equation \eqref{sysDetEqTelEq_4} with respect to $K$ looks like
$uK_u=\nu K+\lambda H$, where $\nu, b=\const$, otherwise $\eta
\equiv 0$.

Let $K \bar \in \langle 1, H \rangle$. Using the above suppositions,
we can rewrite equations
\eqref{sysDetEqTelEq_2}--\eqref{sysDetEqTelEq_4} as
\begin{equation} \label{sysDetEqHumu_1}
\begin{array}{ll}
\frac{\varphi_x}{\varphi}\xi=(2\nu-\mu)\eta^1+2\tau_t,\\
\eta^1_t=\eta^1_x=0,\quad \xi_{xx}=\tau_{tt}=0,\\
\xi_x=(\mu-\nu)\eta^1, \quad (\xi\frac{h_x}{h})_x=-\lambda\eta^1.
\end{array}
\end{equation}
From the first equation of system \eqref{sysDetEqHumu_1} we can get
$\varphi \in \{e^{qx}, |x|^r(r\neq 0), 1\}$ mod $\hat G_h^{\sim}$.

For $\varphi=e^{qx}$ it follows from the determining equations
\eqref{sysDetEqHumu_1} that $\xi_x=0, \nu=\mu, \lambda\neq 0,
\frac{h_x}{h}=2\alpha$. Thus, $h=h_0e^{\alpha x^2+h_1x}=e^{\alpha
x^2}$ mod $\hat G_h^{\sim}$, $\alpha \neq 0, f=h\varphi=e^{\alpha
x^2+qx}, K=\lambda |u|^{\mu}\ln |u|$ mod $\hat G_h^{\sim}$ that
falls precisely into case
\ref{TableGrClasHpower}.\ref{CaseHpowerK1Fpowerg=h}.

If $\varphi=|x|^r, r\neq0$, then $r\xi/x =
(2\nu-\mu)\eta^1+2\tau_t$. Therefore, $\xi = (\mu-\nu)\eta^1 x,
(\mu-\nu)(\frac{xh_x}{h})_x=-\lambda$. Since $\mu\neq \nu$
(otherwise, $K\in \langle 1, H\rangle$) we have $\lambda = 0$ mod
$\hat G_h^{\sim}$. Therefore, $h = |x|^q$ mod $\hat G^{\sim}$. Then
$f = |x|^p, p \neq q$, $K=|u|^{\nu}$,  and we obtain case
\ref{TableGrClasHpower}.\ref{CaseHpowerKpowerF1g=h}.

Value $\varphi=1$ results in $2\tau_t=(\mu-2\nu)\eta^1,
\xi=(\mu-\nu)\eta^1x+\xi_0$. If $\nu=\mu$ then $\lambda \neq 0$
(otherwise, $K\in \langle 1, H\rangle$), $\lambda= 1$ mod $\hat
G_h^{\sim}$ , $(\frac{h_x}{h})_x = 2\alpha(\eta^1=-2\alpha)$.
Therefore, $h=h_0e^{\alpha x^2+h_1x}=e^{\alpha x^2}$ mod $\hat
G_h^{\sim}$ , $K =|u|^{\mu}\ln |u|$ that follows to case
\ref{TableGrClasHpower}.\ref{CaseHpowerK1Fpowerg=h}. If
$\nu\neq\mu$, then $\lambda = 0$ mod $\hat G_h^{\sim}$ . Therefore,
$h \in \{|x|^q, 1, e^{px}\}$ $\hat G_h^{\sim}$ that yields subcases
of \ref{TableGrClasHpower}.\ref{CaseHpowerK1Fpowerg=h},
\ref{TableGrClasHpower}.\ref{CaseHu-4K0HForallFg=h} and
\ref{TableGrClasForAllHg=1}.\ref{CaseForAllHForAllKFexp1}
correspondingly.

Now, we turn to the case $K \in \langle 1, H \rangle$, $H\neq$const.
For convenience we also consider this case using the gauge $g = 1$.
Hence $K = 0, 1 $ mod $\hat G_1^{\sim}$ . Application of the above
suppositions reduces the determining equations to the system
\begin{gather*}
2\xi_x+\frac{f_x}{f}\xi=\mu\eta^1+2\tau_t,\quad \eta^1_{xx}=0,\quad K\eta^1_x=\varphi \eta^1_{tt},\\
(\xi_x+\frac{\varphi_x}{\varphi}\xi-\tau_t)K=0, \quad
\xi_{xx}=2(\mu+1)\eta^1_x.
\end{gather*}
Solving this system and noting that
$\eta^1=\frac{1}{2}\tau_t+\alpha(x)$, we can get
$\eta^1=\frac{1}{2}\tau_t+\beta x+\alpha_0$ and $\xi=(\mu+1) \beta
x^2+\gamma_1x+\gamma_0$. Substituting these values into the first
determining equation and differentiating it with respect to the
variable $t$ we obtain
\begin{equation} \label{sysDetEqHumu_2}
\begin{array}{ll}
\frac{1}{2}(\mu+4)\tau_{tt}=0.\\
((\mu+1) \beta x^2+\gamma_1x+\gamma_0)\frac{f_x}{f}=-(3\mu+4)\beta
x+\frac{1}{2}(\mu+4)\tau_t+\mu \alpha_0-2\gamma_1.
\end{array}
\end{equation}
The first equation of the above system implies that there exist two
cases should be considered: $\tau_{tt}=0$ if $\mu \neq -4$ and
$\tau_{tt} \neq 0$ if $\mu=-4$. The second equation of system
\eqref{sysDetEqHumu_2} gives $l$ linearly independent equations for
$f$ of form $(\alpha^2
x^2+\alpha^1x+\alpha^0)\frac{f_x}{f}=\beta_1x+\beta_0$. If $l= 0$,
then $\xi= 0, \beta=0, \frac{1}{2}(\mu+4) \tau_t =-\mu\alpha_0$.
Considering case $l = 1$, we get $(\alpha^2, \beta_1) \neq (0, 0)$,
$(\alpha^0, \beta_0) \neq (0, 0)$, otherwise $l > 1$. At last, if $l
\geq 2$, then $f\in \{1, e^{px}, |x|^p, p \neq 0\}$ mod $\hat
G_1^{\sim}$.

For the case $\tau_{tt}=0, \mu \neq -4$, substituting the above
values into the determining equation directly for $K = 0$ and
obvious integration leads to the cases
\ref{TableGrClasHpower}.\ref{CaseHpowerKpowerFpowerg=h} (case $l =
0$), \ref{TableGrClasHpower}.\ref{CaseHu-4K0Ff3g=h} (case $l = 1$)
and \ref{TableGrClasHpower}.\ref{CaseHu-4K0F1g=h}a,
\ref{TableGrClasHpower}.\ref{CaseHu-4K0F1g=h}b,
\ref{TableGrClasHpower}.\ref{CaseHu-4K0F1g=h}c,
\ref{TableGrClasHpower}.\ref{CaseHu-43K0F2g=h}  (case $l = 2$). The
case $\tau_{tt}\neq 0, \mu=-4$ with $K=0$ is corresponding to the
results \ref{TableGrClasHpower}.\ref{CaseH1-43K1g=h},
\ref{TableGrClasHpower}.\ref{CaseHu-43K0F3g=h}.

The classification for $K = 1$ is corresponding to subcases
\ref{TableGrClasHpower}.\ref{CaseHpowerKpowerF1g=h},
\ref{TableGrClasHpower}.\ref{CaseHpowerK1Fpowerg=h} and
\ref{TableGrClasHpower}.\ref{CaseHu-4K0HForallFg=h}.\\
\\
{\bf Case 3: $k=2$} (the gauges $g=1$ and $g=h$, tables~\ref
{TableGrClasHexpg=1}, \ref{TableGrClasHpowerg=1} and
\ref{TableGrClasHexp}, \ref{TableGrClasHpower} ). The assumption of
two independent equations of form~\eqref{sysDetEqTelEq_2} for $H$
yields $H =$ const, i.e. $H = 1$ mod $G^{\sim}$. $K_u \neq 0$
(otherwise, equation~\eqref{eqVarCoefTelegraphEq} is linear). In
what follows, we only use the gauge $g=h$.
Equations~\eqref{sysDetEqTelEq_2}--\eqref{sysDetEqTelEq_4} can be
written as
\begin{gather}
2(\xi_x-\tau_t)+(\frac{f_x}{f}-\frac{h_x}{h})\xi=0,\label{sysDetEqHConst_1}\\
(h\eta_x)_x+h K\eta_x-f\eta_{tt}=0,\label{sysDetEqHConst_2}\\
K_u
\eta+\xi_xK+\xi\frac{h_{xx}}{h}-\xi_{xx}+2\eta^1_{x}+(\frac{\xi}{h})_xh_x=0.\label{sysDetEqHConst_3}
\end{gather}
The latter equation looks similar to $(au + b)K_u = cK + d$ with
respect to $K$, where $a, b, c, d = \const$. Therefore, to within
transformations from $G^{\sim}$, $K$ must take one of four values:
\[
K = u^{\nu},\quad \nu\neq 0, 1, \qquad K = \ln u , \qquad K = e^u,
\qquad K =u.
\]
Classification for these values is carried out in the way similar to
the above. The obtained extensions can be entered in either table
\ref{TableGrClasHexp} or table \ref{TableGrClasHpower}.  The gauge
$g=1$ can be proved in a similar way.

The problem of the group classification of
equation~\eqref{eqVarCoefTelegraphEq} is exhaustively solved.

\subsection{Classification with respect to the set of point transformations}\label{SectionOnAdditEquivTr}

Although we have performed the classification by using
extended equivalence group, we can find in the classification
results equations that some cases from
tables~\ref{TableGrClasForAllHg=1}--\ref{TableGrClasHpowerg=1} or
\ref{TableGrClasForAllH}--\ref{TableGrClasHpower} are equivalent
with respect to point transformations which obviously do not belong
to $\hat G^{\sim}$. These transformations are called {\it additional
equivalence transformations} and lead to simplification of further
application of group classification results (see reference
\cite{Popovych&Ivanova2004NVCDCEs,Huang&Ivanova2007,Ivanova&Popovych&Sophocleous2007}
for details). The simplest way to find such additional equivalences
between previously classified equations is based on the fact that
equivalent equations have equivalent maximal Lie invariance
algebras. Explicit formulas for pairs of point-equivalent extension
cases and the corresponding additional equivalence transformations
are adduced after the tables. One can check that there exist no
other point transformations between the equations from tables
~\ref{TableGrClasForAllHg=1}--\ref{TableGrClasHpowerg=1} or tables
\ref{TableGrClasForAllH}--\ref{TableGrClasHpower}. Using this we can
formulate the following theorem.

\begin{theorem}\label{TheoOnClaOfNonWavePointTrans}
Up to point transformations, a complete list of extensions of the
maximal Lie invariance group of equations from class
\eqref{eqVarCoefTelegraphEq} is exhausted by the cases from tables
~\ref{TableGrClasForAllHg=1}--\ref{TableGrClasHpowerg=1} or tables
\ref{TableGrClasForAllH}--\ref{TableGrClasHpower} numbered with
Arabic numbers without Roman letters and subcases $``a"$ of each
multi-case.
\end{theorem}

As one can see, the above additional equivalence transformations have multifarious structure.
This displays a complexity of a structure of the set of admissible transformations.
Usually the problems of finding of all possible admissible transformations are very difficult to solve,
see, e.g.,~\cite{Kingston&Sophocleous1998,Kingston&Sophocleous2001,
Popovych&Ivanova&Eshraghi2004Gamma,Popovych2006}.
We will try to discuss the structure of the set of admissible
transformations of class~\eqref{eqVarCoefTelegraphEq} in a sequel paper.

\section{Lie reduction and similarity solutions}\label{SectionOnLieRedu}

In this section new Lie exact solutions for the equations from the initial class
are constructed to just illustrate possible applications of the classification
results obtained. We mainly perform group analysis of three classes of equations
possessing nontrivial symmetry properties from the obtained classification lists by the reduction
method and then apply to finding similarity solutions. For this purpose, we first construct
the optimal sets of subalgebras for each kind of maximal Lie invariance
algebras arising from group classification, then perform the reductions with respect to obtained
subalgebras. The method of reduction
with respect to subalgebras of Lie invariance algebras is well-known
and quite algorithmic to use in most cases; we refer to the standard
textbooks on the subject \cite{Olver1986,Ovsiannikov1982}.

We first consider the case
\ref{CaseHu-4K0HForallFg=h} of Table \ref{TableGrClasHpower}, i.e., the equation
\begin{equation}\label{eqVarCoefTelegraphEqfgh1}
u_{tt}=(u^{\mu}u_x)_x+u^{\nu}u_x,
\end{equation}
which admits the three-dimensional  Lie
invariance algebra $\mathfrak{g}$ generated by the operators
\[
Q_1=\partial_t,\quad Q_2=\partial_x,\quad
Q_3=(\mu-2\nu)t\partial_t+2(\mu-\nu)x\partial_x+2u\partial_u.
\]
These operators satisfy the commutations relations
\[
[Q_1,~Q_2]=0,\quad [Q_1,~Q_3]=(\mu-2\nu)Q_1,\quad[Q_2,~Q_3]=2(\mu-\nu)Q_2.
\]
An optimal set of subalgebras of the algebra $\mathfrak{g}$ can be easily constructed with application
of the standard technique \cite{Olver1986,Ovsiannikov1982}. Another way is to take the set from \cite{Patera&Winternitz1977}, where optimal
sets of subalgebras are listed for all three- and four-dimensional algebras. A complete list of
inequivalent one-dimensional subalgebras of the algebra g is exhausted by the subalgebras
$\langle Q_1 \rangle,~ \langle Q_2 \rangle,~ \langle Q_3 \rangle,
~ \langle Q_2-Q_1 \rangle, ~\langle Q_2+Q_1 \rangle.
$
This list can be reduced if we additionally use the
discrete symmetry $(t, x, v)\longrightarrow(t,-x, v)$, which maps
$\langle Q_2+Q_1 \rangle$ to $\langle Q_2-Q_1 \rangle$, thereby reducing
the number of inequivalent subalgebras to four.

The optimal set of two-dimensional subalgebras is formed by the subalgebras $\langle Q_3,~Q_1 \rangle$ ,
$\langle Q_3,~Q_2 \rangle$, and $\langle Q_1,~Q_2 \rangle$. Lie reduction to algebraic equations with the latter two two-dimensional
subalgebra leads only to the trivial zero solution. Below we list all the other subalgebras from
the optimal set as well as the corresponding ansatze and reduced equations in Table
\ref{TableRedODEsWaveEqfgh1}. Solutions of some
reduced equations are adduced.

\setcounter{tbn}{0}

\begin{center}\footnotesize\renewcommand{\arraystretch}{1.1}
Table~\refstepcounter{table}\label{TableRedODEsWaveEqfgh1}\thetable.
Reduced ODEs and algebraic equation for equation~\eqref{eqVarCoefTelegraphEqfgh1}.\\[1ex]
\begin{tabular}{|l|c|c|c|l|}
\hline
N & Subalgebra & Ansatz   & \hfil Reduced ODE \\
\hline \refstepcounter{tbn}\label{ReduODEForCasefgh11}\thetbn &
$\langle Q_1\rangle $ & $u=\varphi(\omega),~~\omega=x$
& $(\varphi^{\mu}\varphi_{\omega})_\omega+\varphi^{\nu}\varphi_{\omega}=0$  \\
\refstepcounter{tbn}\label{ReduODEForCasefgh12}\thetbn &
$\langle Q_2\rangle $ &
$u=\varphi(\omega),~~\omega=t$ &$\varphi_{\omega\omega}=0$  \\
\refstepcounter{tbn}\label{ReduODEForCasefgh13}\thetbn &
$\langle Q_3\rangle $ &
$u=t^{\alpha}\varphi(\omega),~~\omega=\frac{x}{t^{\beta}},$ &$\alpha(\alpha-1)\varphi
+(\beta^2-2\alpha\beta+2\beta)\omega\varphi_{\omega}+\beta^2\omega^2\varphi_{\omega\omega}$\\
$~ $& $~$ &$\alpha=\frac{2}{\mu-2\nu},~\beta=\frac{2(\mu-\nu)}{\mu-2\nu}$& $-(\varphi^{\mu}\varphi_{\omega})_\omega-\varphi^{\nu}\varphi_{\omega}=0$  \\
\refstepcounter{tbn}\label{ReduODEForCasefgh14}\thetbn &
$\langle Q_2-Q_1 \rangle $ &
$u=\varphi(\omega),~~\omega=x+t$ &$(\varphi^{\mu}\varphi_{\omega})_\omega+\varphi^{\nu}\varphi_{\omega}-\varphi_{\omega\omega}=0$  \\
\refstepcounter{tbn}\label{ReduODEForCasefgh15}\thetbn &
$\langle Q_3, Q_1\rangle $ &
$u=Cx^{\frac{1}{\mu-\nu}}$
& $(1+\nu)C^{\mu+1}+(\mu-\nu)C^{\nu+1}=0$  \\
\hline
\end{tabular}
\end{center}
Two kinds of explicit solutions can be constructed for arbitrary values of $\mu$ and $\nu$: the $x$-free solution
$u=c_0+c_1t$  and the stationary solution
$u=(\frac{1+\nu}{\nu-\mu})^{\frac{1}{\nu-\mu}}x^{\frac{1}{\mu-\nu}}$. We can also construct two implicit solutions
from the first and the fourth ODEs in table \ref{TableRedODEsWaveEqfgh1}:
\[
u=\varphi(x),\quad u=\psi(x+t),
\]
where $\varphi$ and $\psi$ sarisfy
\[
\frac{1}{\mu-\nu}\varphi^{\mu-\nu}+\frac{1}{\nu+1}x+\int \frac{C_1}{\varphi^{\nu+1}} d x=0,\quad
\frac{1}{\mu-\nu}\psi^{\mu-\nu}+\frac{1}{\nu}\psi^{-\nu}+\frac{1}{\nu+1}\omega+\int \frac{C_2}{\psi^{\nu+1}} d \omega=0,
\]
and $\omega=x+t$, $C_1, C_2$ are arbitrary constants. Let $\mu$ and $\nu$  be particular values,
we can get some number of explicit exact solutions from the above implicit solutions. For example, from the latter
implicit solution, we can get four triangular function exact solutions if  $\mu=1, \nu=2$:
\[
\begin{array}{ll}
u=-\frac{1}{2}-\frac{\sqrt{3}}{2}\tan[\frac{\sqrt{3}}{6}(x+t)],\\ u=-\frac{1}{2}+\frac{\sqrt{3}}{2}\cot[\frac{\sqrt{3}}{6}(x+t)],\\
u=-\frac{1}{2}-\frac{\sqrt{3}}{2}\tan[\frac{\sqrt{3}}{3}(x+t)]\pm \frac{1}{2}\sqrt{1+3\tan^2[\frac{\sqrt{3}}{3}(x+t)]},\\
u=-\frac{1}{2}+\frac{\sqrt{3}}{2}\cot[\frac{\sqrt{3}}{3}(x+t)]\pm
\frac{1}{2}\sqrt{1+3\cot^2[\frac{\sqrt{3}}{3}(x+t)]};
\end{array}
\]
and a rational solutions if $\mu=-1, \nu=-2$:
\[
u=x+t.
\]

Lie reduction and exact solutions of `truly' variable-coefficient nonlinear telegraph waves
are most interesting. We consider two
cases \ref{TableGrClasHexpg=1}.\ref{CaseHexpKexpnuuF1}$^{*}$
and \ref{TableGrClasHpower}.\ref{CaseHpowerKpowerF1g=h}, i.e. equations
\begin{equation}\label{eqVarCoefTelegraphEqExpp}
e^{px}u_{tt}=(u^{\mu}u_x)_x+\epsilon e^{qx}u^{\nu}u_x,
\end{equation}
\begin{equation}\label{eqVarCoefNonLWaveEquPowp}
|x|^pu_{tt}=(|x|^qu^{\mu}u_x)_x+|x|^qu^{\nu}u_x.
\end{equation}
For each from these cases we denote the basis symmetry operators
adduced in Table \ref{TableGrClasHexpg=1} and  \ref{TableGrClasHpower} by $\mathfrak{g_1}=\langle
Q_1=\partial_t,\,Q_2=[(p-q)\mu-q\nu]\partial_t+2(\mu-\nu)
\partial_x+2qu\partial_u \rangle$ and $\mathfrak{g_2}=\langle
Q_1=\partial_t,\, Q_2=[(p-q+1)\mu-(p-q+2)\nu]t\partial_t+2(\mu-\nu)x
\partial_x+2u\partial_u \rangle$, which are all non-commutative algebra. A complete list of inequivalent
non-zero subalgebras of $\mathfrak{g_1}$ or $\mathfrak{g_2}$ is exhausted by the algebras
$\langle Q_1 \rangle,~\langle Q_2 \rangle$ and $\langle Q_1,~Q_2
\rangle$.

Lie reduction of the equations \eqref{eqVarCoefTelegraphEqExpp} and \eqref{eqVarCoefNonLWaveEquPowp} to
ordinary differential equations (ODEs) and an algebraic equation can
be respectively made with the one-dimensional subalgebra $\langle
Q_1 \rangle,~\langle Q_2 \rangle$ and the two-dimensional subalgebra
$\langle Q_1,~Q_2 \rangle$ which coincides with the whole algebra
$\mathfrak{g_1}$ or $\mathfrak{g_2}$. The associated ansatzes and
reduced equations are listed in Table \ref{TableRedODEsWaveEqexpMu1} and \ref{TableRedODEsWaveEqexpMu}.

\setcounter{tbn}{0}

\begin{center}\footnotesize\renewcommand{\arraystretch}{1.1}
Table~\refstepcounter{table}\label{TableRedODEsWaveEqexpMu1}\thetable.
Reduced ODEs and algebraic equation for equation~\eqref{eqVarCoefTelegraphEqExpp}.\\[1ex]
\begin{tabular}{|l|c|c|c|l|}
\hline
N & Subalgebra & Ansatz   & \hfil Reduced ODE \\
\hline \refstepcounter{tbn}\label{ReduODEForCaseHexpK0F11}\thetbn &
$\langle Q_1\rangle $ &
$u=(\varphi(\omega))^{\frac{1}{\mu+1}},~~\omega=x$
& $\varphi_{\omega\omega}+\epsilon e^{q\omega}\varphi_{\omega}\varphi^{\frac{\nu+1}{\mu+1}-1}=0~~
\mbox{if}~~ \mu\neq -1$\\
$$ & & $u=\exp(\varphi(\omega)),~~\omega=x$ &
$\varphi_{\omega\omega}+\epsilon e^{(v+1)\varphi+q\omega}\varphi_{\omega}=0~~\mbox{if}~~ \mu=-1$\\
\refstepcounter{tbn}\label{ReduODEForCaseHexpK0F12}\thetbn &
$\langle Q_2\rangle $ & $u=|t|^{\alpha}\varphi(\omega),~~\omega=x+\beta\ln |t|$ &
$\alpha(\alpha-1)\varphi
+(2\alpha\beta-\beta)\varphi_{\omega}+\beta^2\varphi_{\omega\omega}$
\\$~ $& $~$ &$\alpha=-\frac{2q}{(q-p)\mu+p\nu},~\beta=\frac{2p(\nu-\mu)}{p[(p-q)\mu-p\nu]}$& $-\omega^{-p}(\varphi^{\mu}\varphi_{\omega})_\omega-\epsilon\omega^{q-p}\varphi^{\nu}\varphi_{\omega}=0$  \\
\refstepcounter{tbn}\label{ReduODEForCaseHexpK0F13}\thetbn &
$\langle Q_1, Q_2\rangle $ & $u=Ce^{\frac{qx}{\mu-\nu}}$
& $q(\mu+1)C^{\mu+1}+\epsilon (\mu-\nu) C^{\nu+1}=0$  \\
\hline
\end{tabular}
\end{center}

\setcounter{tbn}{0}

\begin{center}\footnotesize\renewcommand{\arraystretch}{1.1}
Table~\refstepcounter{table}\label{TableRedODEsWaveEqexpMu}\thetable.
Reduced ODEs and algebraic equation for equation~\eqref{eqVarCoefNonLWaveEquPowp}.\\[1ex]
\begin{tabular}{|l|c|c|c|l|}
\hline
N & Subalgebra & Ansatz   & \hfil Reduced ODE \\
\hline \refstepcounter{tbn}\label{ReduODEForCaseHexpK0F11}\thetbn &
$\langle Q_1\rangle $ &
$u=(\varphi(\omega))^{\frac{1}{\mu+1}},~~\omega=x$
& $\varphi\varphi_{\omega\omega}-\varphi_{\omega}^2
+q\varphi\varphi_{\omega}+\varphi_{\omega}\varphi^{\frac{\nu+1}{\mu+1}+1}=0~~ \mbox{if}~~ \mu\neq -1$\\
$$ & & $u=\exp(\varphi(\omega)),~~\omega=x$ &
$(\varphi_{\omega}\omega^q)_{\omega}+\omega^qe^{(\nu+1)\varphi}\varphi_{\omega}=0~~\mbox{if}~~ \mu=-1$\\
\refstepcounter{tbn}\label{ReduODEForCaseHexpK0F12}\thetbn &
$\langle Q_2\rangle $ & $u=|t|^{\alpha}\varphi(\omega),~~\omega=x|t|^{\beta}$ &
$\alpha(\alpha-1)\varphi
+(\beta^2+2\alpha\beta-\beta)\omega\varphi_{\omega}+\beta^2\omega^2\varphi_{\omega\omega}$
\\$~ $& $~$ &$\alpha=\frac{2}{(p-q)(\mu-\nu)+(\mu-2\nu)},$& $-\omega^{-p}(\omega^q\varphi^{\mu}\varphi_{\omega})_\omega-\omega^{q-p}\varphi^{\nu}\varphi_{\omega}=0$  \\
$~ $& $~$ &$\beta=\frac{2(\nu-\mu)}{(p-q)(\mu-\nu)+(\mu-2\nu)}$& $$  \\
\refstepcounter{tbn}\label{ReduODEForCaseHexpK0F13}\thetbn &
$\langle Q_1, Q_2\rangle $ & $u=Cx^{\frac{1}{\mu-\nu}}$
& $[(q-1)(\mu-\nu)+\mu+1]C^{\mu+1}+(\mu-\nu) C^{\nu+1}=0$  \\
\hline
\end{tabular}
\end{center}
Reduction to algebraic equations gives the following solutions of
the initial equations \eqref{eqVarCoefTelegraphEqExpp} and
\eqref{eqVarCoefNonLWaveEquPowp} respectively:
\[
u=\bigg{[}\epsilon\frac{q(\mu+1)}{\mu-\nu}\bigg{]}^{\frac{\mu+1}{\nu+1}}e^{\frac{qx}{\mu-\nu}}; ~~~~
u=\bigg{[}\frac{(q-1)(\mu-\nu)+(\mu+1)}{\nu-\mu}\bigg{]}^{\frac{\mu+1}{\nu+1}}x^{\frac{1}{\mu-\nu}}.
\]
Furthermore, some of the reduced ordinary differential equations in
tables \ref{TableRedODEsWaveEqexpMu1} and \ref{TableRedODEsWaveEqexpMu}
are the modification of the Emden-Fowler and the
Lane-Emden equations
\cite{Polyanin2003,Goenner&Havas2000}. For example, the first
equation corresponding to case 1 of table
\ref{TableRedODEsWaveEqexpMu} are the standard Emden-Fowler
equation, while the second one to case 1 of table
\ref{TableRedODEsWaveEqexpMu} is the generalized Lane-Emden
equation. Solutions of these equations are known for a number of
parameter values (see e.g. \cite{Polyanin2003,Goenner&Havas2000}).
As a result, classes of exact solutions can be constructed for wave
equations \eqref{eqVarCoefTelegraphEqExpp} and
\eqref{eqVarCoefNonLWaveEquPowp} for a wide set of the parameters
$\mu$ and $q$. We omit these results in order to avoid a cumbersome
enumeration.

\section{On nonclassical symmetries}\label{SectionOnNonclassicalSym}

In 1969, Bluman and Cole introduced an essential generalization of
Lie symmetry in the study symmetry reduction of the linear heat
equation \cite{Bluman&Cole1969}. These generalized symmetries are
often called nonclassical symmetries (called also conditional or
$Q$-conditional symmetries) nevertheless it was not used in
\cite{Bluman&Cole1969}. A precise and rigorous definition of this
notion was suggested noticeably later
\cite{Fushchych&Tsy1987,Zhdanov&Tsy&Pop1999} (see also
\cite{Kun&Pop2009} for a recent discussion on definition of
nonclassical symmetries). Since then there is an explosion of
research activity in the area of investigation of nonclassical
symmetries of PDEs arising from different fields of physics, biology
and chemistry
\cite{Clarkson&Mansfield1994,Clarkson&Winternitz1994,Grundland&Tafel1995,
Levi&Winternitz1989,Nucci&Clarkson1992,Nucci2003,Nucci&Leach2000}.
Some of these works concern with nonlinear wave equations. See, for
example, \cite{Ibragimov1994V1,Foursov&Vor1996}.

Generally speaking, there are two main features of nonclassical
symmetries of differential equation difference from Lie symmetries.
The first one is that they can yield solutions not obtainable from
the classical Lie symmetries. The second feature is the deriving systems of determining equations for
nonclassical symmetries which crucially depends on the interplay between
the operators and the equations under consideration, and thus are different from the Lie
symmetries. Due to these facts, when studying the general form of nonclassical
symmetry operators
$Q=\tau(t,x,u)\partial_t+\xi(t,x,u)\partial_x+\eta(t,x,u)\partial_u ((\tau,\xi)\neq(0,0))$ for the
linear heat equation $u_t = u_{xx}$,  there are two essentially
different cases of nonclassical symmetries should be considered: the {\it regular} case
$\tau \neq 0$ and the {\it singular} case $\tau = 0$. The
factorization up to the equivalence of operators 
gives the two respective cases for the further investigation: 1)
$\tau=1$ and 2) $\tau = 0, \xi= 1$. In particular, for the singular case the system
of determining equations for nonclassical symmetries consists of a
single (1+2)-dimensional nonlinear evolution equation for the
unknown function $\eta$ and, therefore, is not overdetermined. The
determining equation is reduced by a nonlocal transformation to the
initial equation with an additional implicit independent variable
which can be assumed as a parameter
\cite{Fushchych&Shtelen&Serov&Popovych1992}. The linearity of the
heat equation is inessential here.

Recently, based on the above discoveries, Kunzinger and Popovych \cite{Kun&Pop2008} raise a number of interesting
questions, to wit:
Is the partition of sets of nonclassical symmetry operators of the linear heat
equation with the conditions of vanishing and nonvanishing
coefficients of operators to regular and singular cases universal
for any differential equations or is it appropriate only for certain
classes of differential equations?  What is the proper partition of
sets of nonclassical symmetry operators different from the conventional one?
What are possible causes for the existence of singular cases for
nonclassical symmetry operators?
 The answer to these questions has some
fundamental importance in the research of nonclassical symmetries and will make
finding an optimal way of obtaining the determining
equation for nonclassical symmetries become possible. They gave a detail investigation on these
questions and present a novel framework of singular reduction
operators to clarify the main idea \cite{Kun&Pop2008}. Here and below, following \cite{Kun&Pop2008} we use sometimes the shorter and more natural term `reduction operators' instead of `nonclassical symmetry operators' or `operators of nonclassical symmetry'.

They also show that for any
(1+1)-dimensional evolution equation \cite{Kun&Pop2008}
\[
u_t=H(t,x,u,u_1,...,u_k),~~k>1,~~u_k=\partial^k u/\partial x^k,
~~H_{u_k}\neq 0,
\]
the conventional partition of the set of
its reduction operators with the conditions $\tau \neq 0$ and
$\tau=0$ is natural since it coincides with the partition of
the set into the singular and regular subsets.
After factorizing the subsets of the reduction operator set with
respect to the usual equivalence relation of reduction operators (see
Definition \ref{DefEquiOpera} below), there exist two different cases of
inequivalent reduction operators: the regular case $\tau = 1$ and
the singular case $\tau = 0$ and $\xi = 1$, which should be investigated
separately. {\it However, this is a specific property of evolution
equations which does not hold for general partial differential
equations in two independent variables}. In particular, they show
that for the class of nonlinear wave equations
\[
u_{tt}-u_{xx}=F(u),
\]
where $F$ is an arbitrary smooth function of $u$, which possesses
two singular sets of reduction operators, singled out by the
conditions $\tau=\xi$ and $\tau=-\xi$, and one regular set of
reduction operators, associated with the condition $\tau\neq\pm
\xi$. The singular sets are mapped to each other by alternating the
sign of $x$ and hence one of them can be excluded from the
consideration. After factorization with respect to the equivalence
relation of vector fields, there are two cases for further study:
the singular case $\tau= \xi= 1$ and the regular case $\tau\neq ¡À1,
\xi = 1$.

However, for more general nonlinear wave equations there exist no
general results. Therefore, it is shown that the structures of
condition symmetries of hyperbolic type nonlinear partial
differential is more complicated than general evolution equation.

In what follows we extend this new framework of singular reduction operators to the
(1+1)-dimensional variable coefficient nonlinear telegraph
equations \eqref{eqVarCoefTelegraphEq}. With the aid of the
transformation $\tilde t = t, \tilde x =\int \frac{dx}{g(x)}, \tilde
u = u$ from $G^{\sim}/G^{\sim g}$ in section
\ref{SectiononEquivaTrans}, we can reduce equation
\eqref{eqVarCoefTelegraphEq} to one which has the same form with
equation \eqref{g=1eqVarCoefTelegraphEq}. Thus, without loss of
generality we can restrict ourselves to investigation in detail the
equation
\begin{equation}\label{eqVarCoefTelegraphEq_1}
f(x)u_{tt}-(H(u)u_x)_x-h(x)K(u)u_x=0,
\end{equation}
where $f=f(x),
h=h(x), H=H(u)$ and $K=K(u)$ are arbitrary smooth functions of the
corresponding variables, $fH>0$.

For the sake of completeness, let us first
review some necessary definitions and statements on nonclassical
symmetries
\cite{Fushchych&Zhdanov1992,Vaneeva&Popovych&Sophocleous2010,Zhdanov&Tsy&Pop1999,Kun&Pop2009}
and singular reduction operator \cite{Kun&Pop2008}.

\subsection{Brief review of reduction operators of differential equation}
Consider an $r$th order differential equation $\mathcal{L}$ of the
form $L(t, x, u_{(r)}) = 0$ for the unknown function $u$ of the two
independent variables $t$ and $x$, where $L = L[u] = L(t, x,
u_{(r)})$ is a fixed differential function of order $r$ and
$u_{(r)}$ denotes the set of all the derivatives of the function $u$
with respect to $t$ and $x$ of order not greater than $r$, including
$u$ as the derivative of order zero.

In order to discuss the conditional symmetries of equation
$\mathcal{L}$, we will first treat equation $\mathcal{L}$ from a
geometric point of view as an algebraic equation in the jet space
$J^r$ of order  $r$ and is identified with the manifold of its
solutions in $J^r$ \cite{Olver1986}: $
\mathcal{L}=\{(t,x,u_{(r)})\in J^r|L(t, x, u_{(r)}) = 0\}.$
Let $\mathcal{Q}$ denote the set of vector fields of the general
form
\begin{equation}\label{VectorFieldofNonClaSymm}
Q=\tau(t,x,u)\partial_t+\xi(t,x,u)\partial_x+\eta(t,x,u)\partial_u,~~~(\tau,\xi)\neq(0,0),
\end{equation}
which is a first-order differential operator on the space
$\mathbb{R}^2\times \mathbb{R}^1$ with coordinates $t , x,$ and $u$.
Then all functions invariant under $Q$ and only such functions
satisfy a first order differential equation
\begin{equation}\label{AddiCon}
Q[u]:=\tau u_t+\xi u_x-\eta=0
\end{equation}
called the the characteristic equation (also known as invariant
surface condition).

Denote the manifold defined by the set of all the differential
consequences of the characteristic equation $Q[u] = 0$ in $J^r$ by
$\mathcal{Q}_{(r)}$, i.e.,
\[
\mathcal{Q}_{(r)}=\{(t, x, u_{(r)})\in J^r|D_t^\alpha D_x^\beta
Q[u]=0,~\alpha, \beta\in \mathbb{N}\cup \{0\},~\alpha+\beta<r\},
\]
where $D_t =\partial_t
+u_{\alpha+1,\beta}\partial_{u_{\alpha,\beta}}$ and $D_x =\partial_x
+u_{\alpha,\beta+1}\partial_{u_{\alpha,\beta}}$ are the operators of
total differentiation with respect to the variables $t$ and $x$, and
the variable $u_{\alpha,\beta}$ of the jet space $J^r$ corresponds
to the derivative $\frac{\partial^{\alpha+\beta} u}{\partial
t^\alpha \partial x^\beta}$. Denote also by $Q_{(r)}$ the standard $r$th prolongation of $Q$ to
the space $J^r$:
\[
Q_{(r)}=Q+\sum_{0<\alpha+\beta\leq r}
\eta^{\alpha\beta}\partial_{u_{\alpha,\beta}},~~
\eta^{\alpha\beta}:=D_t^\alpha D_x^\beta Q[u]+\tau
u_{\alpha+1,\beta}+\xi u_{\alpha,\beta+1}.
\]

\begin{definition}
The differential equation  $\mathcal{L}$ is called conditionally
invariant with respect to the operator $Q$ if the relation
\begin{equation}\label{DeterEq}
Q_{(r)}[L(t, x, u_{(r)})]\bigg{|}_{\mathcal{L}\cap
\mathcal{Q}_{(r)}}= 0
\end{equation}
holds, which is called the conditional invariance criterion. Then
$Q$ is called conditional symmetry (or nonclassical symmetry,
Q-conditional symmetries or reduction operator) of the equation
$\mathcal{L}$.
\end{definition}

We denote the set of reduction operators of the equation
$\mathcal{L}$ by $\mathcal{Q}(\mathcal{L})$ which is a subset of
$\mathcal{Q}$. Any Lie symmetry operator of  $\mathcal{L}$  belongs
to $\mathcal{Q}(\mathcal{L})$. Sometimes $\mathcal{Q}(\mathcal{L})$
is exhausted by the operators equivalent to Lie symmetry ones in the
sense of the following definition.

\begin{definition}\label{DefEquiOpera}
Two differential operators $Q$ and $\tilde Q$ in $\mathcal{Q}$ are
called equivalent $(Q\sim \tilde Q)$ if they differ by a multiplier
which is a non-vanishing function of $t, x$ and $u:$ $\tilde Q =
\lambda Q$, where $\lambda = \lambda(t, x, u), \lambda \neq 0$.
\end{definition}

Factoring $\mathcal{Q}$  with respect to this equivalence relation
we arrive at $\mathcal{Q}_f$. Elements of $\mathcal{Q}_f$ will be
identified with their representatives in $\mathcal{Q}$. The property
of conditional invariance is compatible with this equivalence
relation on $\mathcal{Q}$
\cite{Fushchych&Zhdanov1992,Zhdanov&Tsy&Pop1999}:

\begin{lemma}
If the equation $\mathcal{L}$ is conditionally invariant with
respect to the operator $Q$ then it is conditionally invariant with
respect to any operator which is equivalent to $Q$.
\end{lemma}

In view of this lemma, we can see that $Q \in
\mathcal{Q}(\mathcal{L})$ and $\tilde Q\sim Q$ imply $\tilde Q\in
\mathcal{Q}(\mathcal{L})$, i.e. $\mathcal{Q}(\mathcal{L})$ is closed
under the equivalence relation on $\mathcal{Q}$. Therefore, the
equivalence relation on $\mathcal{Q}$ induces a well-defined
equivalence relation on $\mathcal{Q}(\mathcal{L})$; and the
factorization of $\mathcal{Q}$ with respect to this equivalence
relation can be naturally restricted to $\mathcal{Q}(\mathcal{L})$
that results in the subset $\mathcal{Q}_f(\mathcal{L})$ of
$\mathcal{Q}_f$. As in the whole set $\mathcal{Q}_f$, we identify
elements of $\mathcal{Q}_f(\mathcal{L})$ with their representatives
in $\mathcal{Q}(\mathcal{L})$. In this approach the problem of
completely describing all reduction operators for $\mathcal{L}$ is
equivalent to finding $\mathcal{Q}_f(\mathcal{L})$. In fact,
nonclassical symmetries should be studied up to the above
equivalence relation. The elements of $\mathcal{Q}(\mathcal{L})$
which are not equivalent to Lie invariance operators of
$\mathcal{L}$ will be called {\it pure nonclassical} symmetries of
$\mathcal{L}$.

The conditional invariance criterion admits the following useful
reformulation \cite{Zhdanov&Tsy&Pop1999}.

\begin{lemma}
Given a differential equation $\mathcal{L}: L[u] = 0$ of order $r$
and differential functions $\tilde L[u]$ and $\lambda[u] \neq 0$ of
an order not greater than $r$ such that $L|_{\mathcal{Q}_{(r)}} =
\lambda \tilde L|_{\mathcal{Q}_{(r)}}$, an operator $Q$ is a
reduction operator of $\mathcal{L}$ if and only if the relation $
Q_{(\tilde r)}\tilde L|_{\mathcal{L}\cap\mathcal{Q}_{(r)}}=0$ holds,
where $\tilde r=\ord\tilde L\leq r$ is the order of the differential
function $\tilde L[u]$ and the manifold $\mathcal{\tilde L}$ is
defined in $J^{\tilde r}$ by the equation $\tilde L[u] = 0$.
\end{lemma}

Consider a vector field $Q$ in the form
\eqref{VectorFieldofNonClaSymm} and a differential function $L =
L[u]$ of order $\ord L = r$ (i.e., a smooth function of variables
$t, x, u$ and derivatives of $u$ of orders up to $r$).

\begin{definition}\label{DefSinguandReguOpe}
The vector field $Q$ is called {\it singular} for the differential
function $L$ if there exists a differential function $\tilde L =
\tilde L[u]$ of an order less than $r$ such that
$L|_{\mathcal{Q}_{(r)}} = \tilde L|_{\mathcal{Q}_{(r)}}$ . Otherwise
$Q$ is called a {\it regular} vector field for the differential
function $L$. If the minimal order of differential functions whose
restrictions on ${\mathcal{Q}_{(r)}}$ coincide with
$L|_{\mathcal{Q}_{(r)}}$ equals $k~(k < r)$ then the vector field
$Q$ is said to be of {\it singularity co-order} $k$ for the
differential function $L$. The vector field $Q$ is called {\it
ultra-singular} for the differential function $L$ if
$L|_{\mathcal{Q}_{(r)}}\equiv 0$.
\end{definition}

For convenience, the singularity co-order of ultra-singular vector
fields and the order of identically vanishing differential functions
are defined to equal $-1$. Regular vector fields for the
differential function $L$ are defined to have singularity co-order
$r = \ord L$. The singularity co-order of a vector field $Q$ for a
differential function $L$ will be denoted by $\sco_L Q$.

If $Q$ is a singular vector field for $L$ then any vector field
equivalent to $Q$ is singular for $L$ with the same co-order of
singularity.

We will say that a vector field $Q$ is {\it (strongly) singular for
a differential equation} $\mathcal{L}$ if it is singular for the
differential function $L[u]$ which is the left hand side of the
canonical representation $L[u] = 0$ of the equation $\mathcal{L}$.
Usually we will omit the attribute ``strongly".

Since left hand sides of differential equations are defined up to
multipliers which are nonvanishing differential functions, the
conditions from Definition \ref{DefSinguandReguOpe} can be weakened
when considering differential equations.

\begin{definition}\label{DefSinguandReguOpeOfDE}
A vector field $Q$ is called {\it weakly singular } for the
differential equation $\mathcal{L}: L[u] = 0$ if there exist a
differential function $\tilde L = \tilde L[u]$ of an order less than
$r$ and a nonvanishing differential function $\lambda= \lambda[u]$
of an order not greater than $r$ such that $L|_{\mathcal{Q}_{(r)}}
=\lambda \tilde L|_{\mathcal{Q}_{(r)}}$  . Otherwise $Q$ is called a
{\it weakly regular} vector field for the differential equation
$\mathcal{L}$. If the minimal order of differential functions whose
restrictions on $\mathcal{Q}_{(r)}$ coincide, up to nonvanishing
functional multipliers, with $L|_{\mathcal{Q}_{(r)}}$ is equal to $k
(k < r)$ then the vector field $Q$ is said to be {\it weakly
singular of co-order} $k$ for the differential equation
$\mathcal{L}$.
\end{definition}

The notions of ultra-singularity in the weak and the strong sense
coincide. Analogous to the case of strong regularity, weakly regular
vector fields for the differential equation $\mathcal{L}$ are
defined to have weak singularity co-order $r = \ord L$. The weak
singularity co-order of a vector field $Q$ for an equation
$\mathcal{L}$ will be denoted by $\wsco_{\mathcal{L}} Q$.

\begin{definition}\label{DefSinguReguOpeOfDE}
A vector field $Q$ is called a {\it singular reduction operator} of
a differential equation $\mathcal{L}$ if $Q$ is both a reduction
operator of $\mathcal{L}$ and a weakly singular vector field of
$\mathcal{L}$.
\end{definition}

After the factorization of the reduction operator under the usual
equivalence relation of reduction operators in Definition
\ref{DefEquiOpera} to singular and regular cases, the classification
of reduction operators can be considerably enhanced and simplified
by considering Lie symmetry and equivalence transformations of
(classes of) equations.

\begin{lemma}\label{LemmaPointTranReguOpeOfDE}
Any point transformation of $t,~x$ and $u$ induces a one-to-one
mapping of $\mathcal{Q}$ into itself. Namely, the transformation $g:
\tilde t = T(t, x, u),~\tilde x=X(t, x, u), \tilde u = U(t, x, u)$
generates the mapping $g_{\ast}: \mathcal{Q} \rightarrow
\mathcal{Q}$ such that the operator $Q$ is mapped to the operator
$g_{\ast}Q=\tilde \tau\partial_{\tilde t}+\tilde \xi\partial_{\tilde
x}+\tilde \eta\partial_{\tilde u}$, where $\tilde \tau(\tilde t,
\tilde x, \tilde u)=QT(t, x, u),~~\tilde \xi(\tilde t, \tilde x,
\tilde u)=QX(t, x, u),~~\tilde \eta(\tilde t, \tilde x, \tilde
u)=QU(t, x, u)$. If $Q'\sim Q$ then $g_{\ast}Q' \sim g_{\ast}Q$.
Therefore, the corresponding factorized mapping $g_f :
\mathcal{Q}_f\rightarrow \mathcal{Q}_f$ also is well defined and
bijective.
\end{lemma}

Lemma \ref{LemmaPointTranReguOpeOfDE} results in appearing
equivalence relation between operators, which differs from usual one
described in Definition \ref{DefEquiOpera}.

\begin{definition}\label{DefGroupEquiOpera}
Two differential operators $Q$ and $\tilde Q$ in $\mathcal{Q}$ are
called equivalent with respect to a group $G$ of point
transformations if there exists $g\in G$ for which the operators $Q$
and $g_{\ast}\tilde Q$ are equivalent. We denote this equivalence by
$Q \sim \tilde Q$ mod $G$.
\end{definition}

The problem of finding reduction operators is more complicated than
the similar problem for Lie symmetries because the first problem is
reduced to the integration of an overdetermined system of nonlinear
PDEs, whereas in the case of Lie symmetries one deals with a more
overdetermined system of linear PDEs. The question occurs: could we
use equivalence and gauging transformations in investigation of
reduction operators as we do for finding Lie symmetries? The
following statements give the positive answer.

\begin{lemma}\label{LemmaInducePointTranReguOpeOfDE}
Given any point transformation $g$ of an equation $\mathcal{L}$ to
an equation $ \mathcal{\tilde L},~ g_{\ast}$ maps
$\mathcal{Q}(\mathcal{L})$ to $\mathcal{Q}( \mathcal{\tilde L})$
bijectively. The same is true for the factorized mapping $g_f$ from
$\mathcal{Q}_f(\mathcal{L})$ to $\mathcal{Q}_f( \mathcal{\tilde
L})$.
\end{lemma}

\begin{corollary}\label{CoroInducePointTranReguOpeOfDE}
Let $G$ be the point symmetry group of an equation $\mathcal{L}$.
Then the equivalence of operators with respect to the group $G$
generates equivalence relations in $\mathcal{Q}(\mathcal{L})$ and in
$\mathcal{Q}_f(\mathcal{L})$.
\end{corollary}

Consider the class $\mathcal{L}|_{\mathcal{S}}$ of equations
$\mathcal{L}_\theta$: $L(t, x, u_{(r)}, \theta) = 0$ parameterized
with the parameter-functions $\theta= \theta(t, x, u_{(r)})$. Here $L$
is a fixed function of $t, x, u_{(r)}$ and $\theta$. The symbol
$\theta$ denotes the tuple of arbitrary (parametric) differential
functions $\theta(t, x, u_{(r)}) = (\theta^1(t, x, u_{(r)}), . . . ,
\theta^k(t, x, u_{(r)}))$ running through the set $\mathcal{S}$ of
solutions of the system $S(t, x, u_{(r)}, \theta_{(q)}(t, x, u_{(r)})) =
0$. This system consists of differential equations on $\theta$,
where $t, x$ and $u_{(r)}$ play the role of independent variables and
$\theta_{(q)}$ stands for the set of all the derivatives of $\theta$
of order not greater than $q$. In what follows we call the functions
$\theta$ arbitrary elements. Denote the point transformation group
preserving the form of the equations from
$\mathcal{L}|_{\mathcal{S}}$ by $G^{\sim}$.

Let $P$ denote the set of the pairs consisting of an equation
$\mathcal{L}_\theta$ from $\mathcal{L}|_{\mathcal{S}}$ and an
operator $Q$ from $\mathcal{Q}(\mathcal{L}_\theta)$. In view of
Lemma \ref{LemmaInducePointTranReguOpeOfDE}, the action of
transformations from the equivalence group $G^{\sim}$ on
$\mathcal{L}|_{\mathcal{S}}$  and
$\{\mathcal{Q}(\mathcal{L}_\theta)| \theta\in\mathcal{S} \}$
together with the pure equivalence relation of differential
operators naturally generates an equivalence relation on $P$.

\begin{definition}\label{DefGroupEquiOperaOfDiffEq}
Let $\theta, \theta'\in \mathcal{S}$, $Q\in
\mathcal{Q}(\mathcal{L}_\theta), Q'\in
\mathcal{Q}(\mathcal{L}_{\theta'})$. The pairs $(\mathcal{L}_\theta,
Q)$ and $(\mathcal{L}_{\theta'}, Q')$ are called
$G^{\sim}$-equivalent if there exists $g\in G^{\sim}$ such that $g$
transforms the equation $\mathcal{L}_\theta$ to the
equation$\mathcal{L}_{\theta'}$, and $Q'\sim g_{\ast}Q$.
\end{definition}

The classification of reduction operators with respect to $G^{\sim}$
will be understood as the classification in $P$ with respect to this
equivalence relation, a problem which can be investigated similar to
the usual group classification in classes of differential equations.
Namely, we construct firstly the reduction operators that are
defined for all values of $\theta$. Then we classify, with respect
to $G^{\sim}$, the values of $\theta$ for which the equation
$\mathcal{L}_\theta$ admits additional reduction operators.\\

\subsection{Singular reduction operators}

Using the above notion and the procedure given by Kunzinger and Popovych in \cite{Kun&Pop2008},
we can obtain the following assertion.

\begin{proposition}\label{ProSinguOperaOfTeleEq}
A vector field
$Q=\tau(t,x,u)\partial_t+\xi(t,x,u)\partial_x+\eta(t,x,u)\partial_u$
is singular for the differential function $L =
f(x)u_{tt}-(H(u)u_x)_x-h(x)K(u)u_x$ if and only if~
$\xi^2f(x)=\tau^2H(u)$.
\end{proposition}

\begin{proof}
Suppose that $\tau\neq0$. According to the characteristic equation
$\tau u_t+\xi u_x-\eta=0$, we can get
\begin{gather*}
u_t=\frac{\eta}{\tau}-\frac{\xi}{\tau} u_x,\\
u_{tt}=(\frac{\eta}{\tau})_t-(\frac{\xi}{\tau})_tu_x+[(\frac{\eta}{\tau})_u-(\frac{\xi}{\tau})_uu_x](\frac{\eta}{\tau}-\frac{\xi}{\tau}
u_x)
-(\frac{\xi}{\tau})[(\frac{\eta}{\tau})_x+(\frac{\eta}{\tau})_uu_x-(\frac{\xi}{\tau})_xu_x\\-(\frac{\xi}{\tau})_uu_x^2-(\frac{\xi}{\tau})u_{xx}].
\end{gather*}
Substituting the formulas of $u_{tt}$ from above formulaes into $L$,
we obtain a differential function
\begin{gather*}
\tilde
L=[f(x)(\frac{\xi}{\tau})^2-H(u)]u_{xx}+f(x)\bigg{\{}(\frac{\eta}{\tau})_t-(\frac{\xi}{\tau})_tu_x
+[(\frac{\eta}{\tau})_u-(\frac{\xi}{\tau})_uu_x](\frac{\eta}{\tau}-\frac{\xi}{\tau}
u_x)
\\-(\frac{\xi}{\tau})[(\frac{\eta}{\tau})_x+(\frac{\eta}{\tau})_uu_x-(\frac{\xi}{\tau})_xu_x
-(\frac{\xi}{\tau})_uu_x^2]\bigg{\}}-H_uu_x^2-hKu_x.
\end{gather*}
According to the definition \ref{DefSinguandReguOpe} of singular
vector field, we have $\ord\tilde L < 2$ if and only if
$f(x)(\frac{\xi}{\tau})^2-H(u)=0$.
\end{proof}

Therefore, for any  $f, h, H$ and $K$ with $fH>0$  the differential
function $L = f(x)u_{tt}-(H(u)u_x)_x-h(x)K(u)u_x$ possesses exactly
two set of singular vector fields in the reduced form, namely, $S=
\{ \partial_t +\sqrt{H/f}\partial_x+\hat \eta
\partial_u\}$ and $S^*= \{\partial_t
-\sqrt{H/f}\partial_x+\hat \eta \partial_u\}$, where $\hat
\eta=\frac{\eta}{\tau}$. Any singular vector field of $L$ is
equivalent to one of the above fields. Moreover, it is easy known
that form Theorem \ref{TheorOnEquivGrouUsaual} each equation of the
form \eqref{eqVarCoefTelegraphEq_1} admits the discrete involutive
transformation $\{x,K\}$ and $\{x,h\}$. According to Corollary
\ref{CoroInducePointTranReguOpeOfDE}, these two transformation
generates a one-to-one mapping between $S$ and $S^*$. Hence it
suffices, up to equivalence of vector fields (and permutation of $x$
and $-x$ ), to investigate only singular reduction operators from
the set $S$.

\begin{proposition}\label{Poro2SinguOperaOfTeleEq}
For any  variable coefficient nonlinear telegraph equations in the
form \eqref{eqVarCoefTelegraphEq_1} the differential function $L =
f(x)u_{tt}-(H(u)u_x)_x-h(x)K(u)u_x$  possesses exactly one set of
singular vector fields in the reduced form, namely, $S= \{
\partial_t +\sqrt{H/f}\partial_x+\hat \eta \partial_u\}$.
\end{proposition}

Thus taking into accountant the conditional invariance criterion for
an equation from class \eqref{eqVarCoefTelegraphEq_1} and the
operator $\partial_t +\sqrt{H/f}\partial_x+\eta \partial_u$, we can
get

\begin{theorem}\label{TheoDESinguOperaOfTeleEq}
Every singular reduction operator of an equation from class
\eqref{eqVarCoefTelegraphEq_1} is equivalent to
\[
Q=\partial_t +\sqrt{H(u)/f(x)}\partial_x+ \eta(t,x,u)
\partial_u,
\]
where the real-valued function $\eta(t,x,u)$ satisfies the
determining equations
\begin{gather}
(-h_xK-2f \eta_{tu}+\frac{1}{2}hKf_x/f+\frac{3}{4}Hf_x^2/f^2 -H_u
\eta \eta_uf/H-2f \eta
\eta_{uu}+\frac{1}{4}fH_u^2 \eta^2/H^2 \nonumber \\
-\frac{1}{2}Hf_{xx}/f-\frac{1}{2}f \eta^2H_{uu}/H)\sqrt{H/f} -H_u
\eta_x-h \eta K_u-\frac{1}{2} \eta H_uf_x/f-2H \eta_{xu}=0,
\label{DESinguReduOpe}\\ - \eta \eta_xH_u+\sqrt{H/f}(f \eta_{tt}-
\eta_{xx}H-h \eta_xK+2f \eta \eta_{tu} +f \eta^2 \eta_{uu})=0.
\nonumber
\end{gather}
\end{theorem}

\subsection{Regular reduction operators}

The above investigation of singular reduction operators of nonlinear
telegraph equation of the form \eqref{eqVarCoefTelegraphEq_1} shows
that for these equation the regular case of the natural partition of
the corresponding sets of reduction operators is singled out by the
conditions $\xi\neq \pm \sqrt{H(u)/f}\tau$. After factorization with
respect to the equivalence relation of vector fields, we obtain the
defining conditions of regular subset of reduction operator:
$\tau=1, \xi\neq \pm \sqrt{H(u)/f}$. Hence we have

\begin{proposition}\label{PoroReduOperaOfTeleEq}
For any  variable coefficient nonlinear telegraph equations in the
form \eqref{eqVarCoefTelegraphEq_1} the differential function $L =
f(x)u_{tt}-(H(u)u_x)_x-h(x)K(u)u_x$ possesses exactly one set of
regular vector fields in the reduced form, namely, $S= \{
\partial_t +\hat \xi \partial_x+\hat \eta \partial_u\}$ with $\hat \xi\neq \pm \sqrt{H(u)/f}$.
\end{proposition}

Consider the conditional invariance criterion for an equation from
class \eqref{eqVarCoefTelegraphEq_1} and the operator $\partial_t +
\xi \partial_x+ \eta \partial_u$ with $ \xi=\neq \pm \sqrt{H(u)/f}$,
we can get

\begin{theorem}\label{TheoDEReguOperaOfTeleEq}
Every regular reduction operator of an equation from class
\eqref{eqVarCoefTelegraphEq_1} is equivalent to
\[
Q=\partial_t + \xi(t,x)\partial_x+ \eta(t,x,u)
\partial_u, \quad  \xi(t,x) \neq \pm
\sqrt{H(u)/f}
\]
where the real-valued functions $ \xi(t,x),  \eta(t,x,u)$ satisfy
the determining equations
\begin{gather}
\xi f_xH/f-\eta H_u+2H\xi_x+2f\xi\xi_t=0, \nonumber \\
2H_u\xi_x+f\eta_{uu}\xi^2-\eta_uH_u-\eta H_{uu}-H\eta_{uu}+\xi
H_uf_x/f=0, \nonumber \\
 -2f\xi\eta\eta_{uu}+\xi hKf_x/f-h\eta
K_u-2H_u\eta_x+hK\xi_x+2f\xi_t\xi_x-\xi h_xK \label{DEReguReduOpe}
\\-2f\xi\eta_{tu}-2f\xi_t\eta_u-2H\eta_{xu}-f\xi_{tt}+H\xi_{xx}=0,
\nonumber \\
-hK\eta_x-H\eta_{xx}+f\eta_{tt}+2f\eta_{tu}\eta+f\eta^2\eta_{uu}-2f\xi_t\eta_x=0.
\nonumber
\end{gather}
\end{theorem}

Solving the above system with respect to the coefficient functions $
\xi, \eta, f, H, h$ and $K$ under the equivalence group
$G_1^{\sim}$, we can get a classification of regular reduction
operator for the class \eqref{eqVarCoefTelegraphEq_1}. However, due
to the strong nonlinearity of system \eqref{DEReguReduOpe}, it is
difficult to get an explicit classification. Hence, we omit the
detail investigation for the general case and concentrate on some
special cases.
\\
\\
{\bf Example.} We study the regular reduction operator of equations
\eqref{eqVarCoefTelegraphEq_1} with $H(u)=K(u)=u, f(x)=h(x)=1$, i.e,
nonlinear telegraph equations
\begin{equation}\label{eqConCoefNonLWaveEq}
u_{tt}=(uu_x)_x+uu_x.
\end{equation}
From table 3, it is easy to know that equation
\eqref{eqConCoefNonLWaveEq} admits three-dimensional Lie algebra
$\mathfrak{g}$ of its infinitesimal Lie symmetries with a basis:
\begin{equation}\label{OperOfeqConCoefNonLWaveEq}
X_1=\partial_t,~~~~X_2=t\partial_t-2u\partial_u,~~~~X_3=\partial_x.
\end{equation}
The corresponding one-parameter groups are time translations and
scale transformations.

We first discuss a special case of the regular reduction operator,
i.e., consider  the conditional symmetry operator in the form:
\begin{equation}\label{OperOfConSymFirstKind}
Q=\partial_x+ \eta(t,x,u)\partial_u.
\end{equation}
With the assumptions $\tau=0, \xi=1$ the determining equations
(\ref{DeterEq}) for the nonlinear telegraph equation
\eqref{eqVarCoefTelegraphEq_1} are as follows:
\begin{equation}\label{DeEqFirstType}
\begin{array}{ll}
\eta_{tu}=0,~~~~\eta_{uu}=0,\\
f\eta_{tt}-H_{uu}\eta^3+H_u(\eta^2f_x/f-3\eta\eta_x-2\eta^2\eta_u)\\-H(\eta_{xx}+2\eta_{xu}\eta
-\eta\eta_uf_x/f-\eta_xf_x/f)-K_uh\eta^2+K(h\eta
f_x/f-h_x\eta-h\eta_x)=0.
\end{array}
\end{equation}
From the first two equations we obtain that
\begin{equation}\label{eta-t-x-u}
\eta(t,x,u)=A(x)u+B(t,x).
\end{equation}
Substituting the latest equation with $H(u)=K(u)=u, f(x)=g(x)=1$
into the last equation of system \eqref{DeEqFirstType}, we can see
that the functions $A(x)$ and $B(t,x)$ satisfy the overdetermined
system:
\begin{equation}\label{sysNonclassical}
\begin{array}{ll}
A_{xx}+5AA_x+A_x+2A^3+A^2=0,\\
B_{xx}+5A_xB+3AB_x+B_x+4A^2B+2AB=0,\\
B_{tt}-3BB_x-2AB^2-B^2=0.
\end{array}
\end{equation}
The last two equations of system \eqref{sysNonclassical} imply the
compatibility condition
\begin{equation}\label{ComCon1}
\begin{array}{ll}
(14A+4)B_x^2+(36A^2+46A_x+18A)BB_x+(2A-37AA_x-4A_x\\
-18A^3-5A^2)B^2=0
\end{array}
\end{equation}
obtained by cross-differential. The second equation of
\eqref{sysNonclassical} is a differential consequence of
\eqref{ComCon1} provided the equation
\begin{equation}\label{ComCon2}
\begin{array}{ll}
(60A_x-48A^2-34A-8)B_x^2-(510AA_x+122A_x+348A^3+234A^2+30A)BB_x\\
+(-267A_x^2-233A^2A_x-135AA_x+6A_x-99A^3-32A^2-70A^4)=0
\end{array}
\end{equation}
is satisfied.

Eliminate $B_x$ from \eqref{ComCon1} and \eqref{ComCon2} we come to
the equation
\begin{equation}\label{ComCon3}
\begin{array}{ll}
B^2(272\,A A_x +215\, A   ^{3}+591\,  A_x   A   ^{2}+ 262\,
A^{4}-2\,A  + 10\,A_x +230\,  A_x   ^{2}+38\,  A^{2} ) (1872\,  A
^{3}\\-32\, A_x ^{2}+3336\,  A ^{4}-148608\,  A   ^{5}-194152
\,  A_x     A  ^{3}-11392\,A  A_x  ^{2}+115668656\, A ^{6}A_x\\
+363896822\, A   ^{7}A_x +11989256\,  A   ^{5}A_x +539345898\,  A
^{5}  A_x   ^{2}- 1263352\,  A   ^{4}A_x \\+5505328\,  A  ^{3} A_x
^{2}+107208\, A_x   ^{3}A +37404240\,  A_x ^{3}  A^{2}+119385824\,
A_x   ^{2}  A   ^{4}\\-942408\,  A_x   ^{2}  A   ^{2 }+301875108\,
A_x   ^{3} A   ^{3}+45437328\,A A_x   ^{4}+508439090\, A   ^{8}A_x
\\+979255660\,  A   ^{6}  A_x   ^{2}+800801280\,  A_x^{3}  A   ^{4}+247114530\,  A   ^{2} A_x^{4}
+638696765 \,  A  ^{7}  A_x ^{2}\\+695655030\,  A ^{5}  A_x^{3}+
337735335\,  A_x   ^{4} A   ^{3}+59322060\,A A_x^{5}+268692620\, A_x
A^{9}\\+34451808\,  A   ^ {8}+83063493\,  A   ^{9}-238410\, A
^{6}+5772399\,  A   ^{7}+6096\,  A_x   ^{3} +686688\,  A_x
^{4}\\+94806420\,  A   ^ {10}+16949160\,  A_x   ^{5} +42378980\,  A
^{11}-128\,  A^{2}+14960\,  A_x     A   ^{2}+ 128\,A A_x ) =0.
\end{array}
\end{equation}

If we take the third factor in \eqref{ComCon3} and the first
equation for the function $A(x)$ in system \eqref{sysNonclassical},
then that overdetermined system for the function $A(x)$ admits
solution $A(x)=0$. Hence, solving system \eqref{sysNonclassical}
with  $A(x)=0$, we can obtain
\[
B(t,x)=a(t),
\]
where the function $a(t)$ is given by:
\begin{equation}\label{SystemODE1}
\pm \int ^{a(t)} \frac{3}{\sqrt{6u^3+3c_1}} du-t-c_2=0
\end{equation}
Taking a special form $a(t)=6/t^2$ from \eqref{SystemODE1} yield an
exact explicit solution of equation \eqref{eqConCoefNonLWaveEq}
obtainable by solving equation \eqref{AddiCon}, which is an ordinary
differential equation in the variable $x$ for the symmetries of
second type, and subsequent solution of equation
\eqref{eqConCoefNonLWaveEq} for the 'constants' of integration
actually depending on the variable $t$:
\[
u(t,x)=\,{\frac {150\,x+25c_2+25c_1\,{t}^{5}-180\,\ln
 \left( t \right) -36}{25{t}^{2}}},
\]
where $c_i (i=1,2)$ are parameters.

If we set $B(t,x)=0$, then last two equations of
\eqref{sysNonclassical} are satisfied and we arrive at an
infinitesimal conditional symmetry
\[
Q=\partial_x+A(x)u\partial_u
\]
with the function $A(x)$ satisfying the ODE
$A_{xx}+5AA_x+A_x+2A^3+A^2=0$. Particular solution $A(x)=1/4[\tanh
(1/2x)-1]$ of the latter equation yields exact solutions of the
nonlinear telegraph equation \eqref{eqConCoefNonLWaveEq}
\[
u(t,x)=[-\cosh(\frac{1}{2}x)]^{\frac{1}{2}}\exp(-\frac{x}{4})(c_1t+c_2),
\]
while particular solution $A(x)=1/4[\coth (1/2x)-1]$ yields the
exact solution
\[
u(t,x)=[\sinh(\frac{1}{2}x)]^{\frac{1}{2}}\exp(-\frac{x}{4})(c_1t+c_2).
\]

Finally, if we take the second factor $272\,A A_x +215\, A
^{3}+591\,  A_x   A   ^{2}+ 262\, A^{4}-2\,A  + 10\,A_x +230\,  A_x
^{2}+38\,  A^{2} =0$ in \eqref{ComCon3}, which together with the
first equation of \eqref{sysNonclassical} imply $A(x)=-1/2$. Thus,
we arrive at an infinitesimal conditional symmetry in the form
\[
Q=\partial_x+[-\frac{u}{2}+\alpha t+\beta]\partial_u,
\]
where $\alpha, \beta$ are arbitrary constants. Solving equation
\eqref{AddiCon}, we obtain an exact solutions
\[
\begin{array}{ll}
u \left( t,x \right) =2\,\alpha\,t+2\,\beta+{e^{-1/2\,x}}c_1\,{\it
AiryAi} \left( -1/ 2\,{\frac {{2}^{2/3} \left( \alpha\,t+\beta
\right) }{{\alpha}^{2/3}}}
 \right) \\+{e^{-1/2\,x}}c_2\,{\it AiryBi} \left( -1/2\,{\frac {{
2}^{2/3} \left( \alpha\,t+\beta \right) }{{\alpha}^{2/3}}} \right)
\end{array}
\]
of equation \eqref{eqConCoefNonLWaveEq}, where $AiryAi(t)$ and
$AiryBi(t)$ are the associated Airy functions of the first kind and
the second kind respectively.

Now we consider the conditional symmetry operator of equation
\eqref{eqConCoefNonLWaveEq} in the form:
\begin{equation}\label{OperOfConSymSecondKind}
Q=\partial_t+\xi(t,x)\partial_x+\eta(t,x,u)\partial_u.
\end{equation}
In this way, solving the determining equations \eqref{DEReguReduOpe}
with $H(u)=K(u)=u, f(x)=h(x)=1$, we can get $\xi=\const, \eta=0$.

From the above discussion we can arrival at
\begin{theorem}
Equation \eqref{eqVarCoefTelegraphEq_1} with $H(u)=K(u)=u,
f(x)=h(x)=1$ is conditionally invariant under the following
operators:
\[
\begin{array}{ll}
(1)~~Q=\partial_x+a(t)\partial_u;\\
(2)~~Q=\partial_x+1/4[\tanh
(1/2x)-1]u\partial_u;\\
(3)~~Q=\partial_x+1/4[\coth
(1/2x)-1]u\partial_u;\\
(4)~~Q=\partial_x+[-\frac{u}{2}+\alpha t+\beta]\partial_u;\\
(5)~~Q=\partial_t+\partial_x;
\end{array}
\]
where $a(t)$ is given by equation \eqref{SystemODE1}.
\end{theorem}

\section{Conservation laws}\label{SectionOnConsLaws}

Apart from exact solutions, classical and nonclassical symmetry
classifications, we know that another important subject of group
analysis is the construction of conservation laws of (systems of)
differential equations, which play an important role in mathematical
physics\cite{Olver1986}. In fact, the knowledge of conservation laws
is useful in the numerical integration of partial differential
equations, for example, to control numerical errors. Also, the
investigation of conservation laws of the Korteweg-de Vries equation
became a starting point of the discovery of a number of techniques
(such as Lax pair, inverse scattering transformation, bi-Hamiltonian
structures, etc.) to solve nonlinear evolution equations. The
existence of a large number of conservation laws of a evolutionary
partial differential equation (system) is a strong indication of its
integrability. Conservation laws have also significant uses in the
theory of non-classical transformations and in the theory of normal
forms and asymptotic integrability.

In this section we classify local conservation laws of
equations~\eqref{eqVarCoefTelegraphEq} with characteristics
depending, at mostly, on $t$, $x$ and~$u$. For classification we use
the direct method described
in~\cite{Popovych&Ivanova2004ConsLaws,Ivanova&Popovych&Sophocleous2007-3}.
To begin with, we adduce a necessary theoretical background on
conservation laws, following,
e.g.,~\cite{Olver1986,Popovych&Ivanova2004ConsLaws,Ivanova&Popovych&Sophocleous2007-3}
and considering for simplicity the case of two independent
variables~$t$ and~$x$. See the above references for the general
case.

Let~$\mathcal{W}$ be a system~$W(t, x, u_{(\rho)})=0$ of $l$ PDEs $W^1=0$, \ldots, $W^l=0$
for $m$ unknown functions $u=(u^1,\ldots,u^m)$
of two independent variables $t$ and $x$.
Here $u_{(\rho)}$ denotes the set of all the partial derivatives of the functions $u$
of order not greater than~$\rho$, including $u$ as the derivatives of the zero order.
Let $\mathcal{W}_{(k)}$  denote the set of all algebraically
independent differential consequences that have, as differential
equations, orders no greater than $k$. We identify
$\mathcal{W}_{(k)}$ with the manifold determined by $\mathcal{W}_{(k)}$ in the jet
space $J^{(k)}$.

\begin{definition}\label{DefinitionOfConsVector}
A {\em conserved vector} of the system $\mathcal{W}$ is a
$2$-tuple $F = (F^1(t, x, u_{(r)}),F^2(t, x, u_{(r)}))$ for which
the divergence $\Div F := D_tF^1+D_xF^2$ vanishes for all solutions
of $\mathcal{W}$ (i.e., $\Div F|_W=0$).
\end{definition}

In Definition \ref{DefinitionOfConsVector} and later, $D_t$ and $D_{x}$ denotes the operator
of total differentiation with respect to the variables $t$ and $x$ respectively.
The notation $V|_{\mathcal{W}}$ means that values of $V$ are
considered only on solutions of the system $\mathcal{W}$.

The crucial notion of the theory of conservation laws is one of  triviality and equivalence of conservation laws.

\begin{definition}\label{DefinitionOfTrivialOfConsVector}
A conserved vector $F$ is called {\em trivial} if
$F^i =\hat F^i+\tilde F^i , i = 1, 2$, where $\hat F^i$ and $ \tilde F^i$  are
functions of $t, x$ and derivatives of $u$ (i.e., differential
functions), $\hat F^i$ vanish on the solutions of $\mathcal{W}$, and the
$2$-tuple $\tilde F = (\tilde F^1, \tilde F^n)$ is a null divergence (i.e., its
divergence vanishes identically).
\end{definition}


\begin{definition}\label{DefinitionOfConsVectorEquivalence}
Two conserved vectors $F$ and $F'$ are called {\em equivalent}
if the vector-function $F'-F$ is a trivial
conserved vector.
\end{definition}

The notion of linear dependence of conserved vectors is introduced in a similar way.
Namely, a set of conserved vectors is {\em linearly dependent}
iff a linear combination of them is a trivial conserved vector.

The above definitions of triviality and equivalence of conserved vectors are natural in view
of the usual ``empiric'' definition of conservation laws of a system of differential equations as
divergences of its conserved vectors, i.e. divergence expressions which vanish for all solutions
of this system. For example, equivalent conserved vectors correspond to the same conservation
law. However, for deeper understanding of the problem and absolutely correct calculations
a more rigorous definition of conservation laws should be used.

For any system~$\mathcal{W}$ of differential equations the set~$\CV(\mathcal{W})$ of conserved vectors of
its conservation laws is a linear space,
and the subset~$\CV_0(\mathcal{W})$ of trivial conserved vectors is a linear subspace in~$\CV(\mathcal{W})$.
The factor space~$\CL(\mathcal{W})=\CV(\mathcal{W})/\CV_0(\mathcal{W})$
coincides with the set of equivalence classes of~$\CV(\mathcal{W})$ with respect to the equivalence relation adduced in
definition~\ref{DefinitionOfConsVectorEquivalence}.

\begin{definition}\label{DefinitionOfConsLaws}
The elements of~$\CL(\mathcal{W})$ are called {\em conservation laws} of the system~$\mathcal{W}$,
and the whole factor space~$\CL(\mathcal{W})$ is called {\em the space of conservation laws} of~$\mathcal{W}$.
\end{definition}

That is why description of the set of conservation laws can be assumed
as finding~$\CL(\mathcal{W})$ that is equivalent to construction of either a basis if
$\dim \CL(\mathcal{W})<\infty$ or a system of generatrices in the infinite dimensional case.
The elements of~$\CV(\mathcal{W})$ which belong to the same equivalence class giving a conservation law~${\cal F}$
are considered all as conserved vectors of this conservation law,
and we will additionally identify elements from~$\CL(\mathcal{W})$ with their representatives
in~$\CV(\mathcal{W})$.
For $F\in\CV(\mathcal{W})$ and ${\cal F}\in\CL(\mathcal{W})$
the notation~$F\in {\cal F}$ will denote that $F$ is a conserved vector corresponding
to the conservation law~${\cal F}$.
In contrast to the order $r_{F}$ of a conserved vector~$F$ as the maximal order of derivatives
explicitly appearing in $F$,
the {\em order of the conservation law}~$\cal F$
is called $\min\{r_{F}\,| F\in{\cal F}\}$.
Under linear dependence of conservation laws we understand linear dependence of them as elements of~$\CL(\mathcal{W})$.
Therefore, in the framework of ``representative'' approach
conservation laws of a system~$\mathcal{W}$ are considered as {\em linearly dependent} if
there exists linear combination of their representatives, which is a trivial conserved vector.

Let the system~$\cal W$ be totally nondegenerate~\cite{Olver1986}.
Then application of the Hadamard lemma to the definition of conservation law and integrating by parts imply that
the left hand side of any conservation law of~$\mathcal W$ can be always presented up to the equivalence relation
as a linear combination of left hand sides of independent equations from $\mathcal W$
with coefficients~$\lambda^\mu$ being functions of $t$, $x$ and derivatives of~$u$:
\begin{equation}\label{CharFormOfConsLaw}
D_tF^1+D_xF^2=\lambda^1 W^1+\dots+\lambda^l W^l.
\end{equation}

\begin{definition}\label{DefCharForm}
Formula~\eqref{CharFormOfConsLaw} and the $l$-tuple $\lambda=(\lambda^1,\ldots,\lambda^l)$
are called the {\it characteristic form} and the {\it characteristic}
of the conservation law~$D_tF^1+D_xF^2=0$ correspondingly.
\end{definition}

The characteristic~$\lambda$ is {\em trivial} if it vanishes for all solutions of $\cal W$.
Since $\cal W$ is nondegenerate, the characteristics~$\lambda$ and~$\tilde\lambda$ satisfy~\eqref{CharFormOfConsLaw}
for the same conserved vector~$F$ and, therefore, are called {\em equivalent}
iff $\lambda-\tilde\lambda$ is a trivial characteristic.
Similarly to conserved vectors, the set~$\Ch(\mathcal{W})$ of characteristics
corresponding to conservation laws of the system~$\cal W$ is a linear space,
and the subset~$\Ch_0(\mathcal{W})$ of trivial characteristics is a linear subspace in~$\Ch(\mathcal{W})$.
The factor space~$\Ch_{\rm f}(\mathcal{W})=\Ch(\mathcal{W})/\Ch_0(\mathcal{W})$
coincides with the set of equivalence classes of~$\Ch(\mathcal{W})$
with respect to the above characteristic equivalence relation.

In most of cases, we can essentially simplify and order classification of conservation laws, taking into account
additionally symmetry transformations of a system or equivalence transformations of a whole
class of systems. Such problem is similar to one of group classification of differential equations.

\begin{proposition}\label{PropEqTrCL}
Any point transformation $g$ maps a class of equations in the conserved form
into itself. More exactly, the transformation $g$: $\tilde t=t^g(t,x,u)$, $\tilde x=x^g(t,x,u)$, $\tilde u=u^g(t,x,u)$
prolonged to the jet space $J^{(r)}$ transforms the equation $D_tF^1+D_xF^2=0$ to the equation \mbox{$D_tF^1_g+D_xF^2_g=0$}.
The transformed conserved vector~$F_g=(F^1_g, F^2_g)$ is determined
by the formula
\begin{gather}
F^1_g(\tilde x,\tilde u_{(r)})=\frac{F^1(x,u_{(r)})D_t\tilde t+
F^2(x,u_{(r)})D_x\tilde t}{D_t\tilde t\,D_x\tilde x-D_x\tilde t\,D_t\tilde x},\nonumber
\\[1ex]
F^2_g(\tilde x,\tilde u_{(r)})=\frac{F^1(x,u_{(r)})D_t\tilde x+
F^2(x,u_{(r)})D_x\tilde x}{D_t\tilde t\,D_x\tilde x-D_x\tilde t\,D_t\tilde x}.\label{eqtrvarconslaw}
\end{gather}
\end{proposition}

\begin{remark}
In the case of one dependent variable ($m=1$) $g$ can be a contact transformation:
$\tilde t=t^g(t,x,u_{(1)})$, $\tilde x=x^g(t,x,u_{(1)})$, $\tilde u_{(1)}=u^g_{(1)}(t,x,u_{(1)})$.
Similar note is true for the below statement.
\end{remark}

\begin{definition}\label{DefinitionOfGEquiOfConsLaws}
Let $G$ be a symmetry group of the system
$\mathcal{W}$. Two conservation laws with the conserved vectors
$F$ and $F'$ are called $G$-equivalent if there exists
a transformation $g\in G$ such that the conserved
vectors $F_g$ and $F'$ are equivalent in the sense of
Definition \ref{DefinitionOfConsVectorEquivalence}.
\end{definition}

Any transformation $g\in G$  induces a linear one-to-one mapping $g_*$ in $\CV(\mathcal{W})$, transforms
trivial conserved vectors only to trivial ones (i.e. $\CV_0(\mathcal{W})$ is invariant with respect to $g_*$ ) and
therefore induces a linear one-to-one mapping $g_f$ in  $\CL(\mathcal{W})$. It is obvious that $g_f$ preserves linear
(in)dependence of elements in $\CL(\mathcal{W})$ and maps a basis (a set of generatrices) of $\CL(\mathcal{W})$ in a
basis (a set of generatrices) of the same space. In such way we can consider the $G$-equivalence
relation of conservation laws as well-determined on $\CL(\mathcal{W})$ and use it to classify conservation
laws.

\begin{proposition}\label{PropositionOnInducedConservedVector}
If system $\mathcal{W}$ admits a one-parameter
group of transformations, then the infinitesimal generator
$Q =\xi^i\partial_i+\eta^a\partial_{u^a}$ of this group can be used for construction
of new conservation laws from known ones.
Namely, differentiating equation \eqref{eqtrvarconslaw} with respect to the
parameter $\epsilon$ and taking the value $ \epsilon= 0$, we obtain the
new conserved vector
\begin{equation}\label{EqOnInducedConservedVector}
\tilde F^i =-Q_{(r)}F^i+ (D_j \xi^i )F^j-(D_j\xi^j)F^i .
\end{equation}
Here $Q_{(r)}$ denotes the $r$-th prolongation of the operator
$Q$.
\end{proposition}

\begin{remark}
Formula \eqref{EqOnInducedConservedVector} can be directly extended to generalized symmetry operators (see, for
example, [4]). A similar statement for generalized symmetry operators in evolutionary form
($\xi^i = 0$) was known earlier [13, 30]. It was used in [20] to introduce a notion of basis of
conservation laws as a set which generates a whole set of conservation laws with action of
generalized symmetry operators and operation of linear combination.
\end{remark}

\begin{proposition}\label{PropositionOnInducedMapping}
Any point transformation $g$ between systems~$\mathcal{W}$ and~$\tilde{\mathcal{W}}$
induces a linear one-to-one mapping $g_*$ from~$\CV(\mathcal{W})$ into~$\CV(\tilde{\mathcal{W}})$,
which maps $\CV_0(\mathcal{W})$ into~$\CV_0(\tilde{\mathcal{W}})$
and generates a linear one-to-one mapping $g_{\rm f}$ from~$\CL(\mathcal{W})$ into~$\CL(\tilde{\mathcal{W}})$.
\end{proposition}

\begin{corollary}
Any point transformation $g$ between systems~$\mathcal{W}$ and~$\tilde{\mathcal{W}}$ induces a linear one-to-one
mapping $\hat g_f$ from $Ch_f (\mathcal{W})$ into $Ch_f (\tilde{\mathcal{W}})$.
\end{corollary}


Consider the class $\mathcal{W}|_{\mathcal{S}}$ of systems
$\mathcal{W}_\theta$: $W(t, x, u_{(\rho)}, \theta) = 0$
parameterized with the parameter-functions $\theta= \theta(t, x,
u_{(\rho)})$. Here $W$ is a fixed function of $t, x, u_{(\rho)}$ and
$\theta$. The symbol $\theta$ denotes the tuple of arbitrary
(parametric) differential functions $\theta(t, x, u_{(\rho)}) =
(\theta^1(t, x, u_{(\rho)}), . . . , \theta^k(t, x, u_{(\rho)}))$ running
through the set $\mathcal{S}$ of solutions of the system $S(t, x,
u_{(\rho)}, \theta_{(q)}(t, x, u_{(\rho)})) = 0$. This system consists
of differential equations on $\theta$, where $t, x$ and $u_{(\rho)}$
play the role of independent variables and $\theta_{(q)}$ stands for
the set of all the derivatives of $\theta$ of order not greater than
$q$. In what follows we call the functions $\theta$ arbitrary
elements. Denote the point transformation group preserving the form
of the equations from $\mathcal{W}|_{\mathcal{S}}$ by
$G^{\sim}=G^{\sim}(W, S)$.

Consider the set $P = P(W, S)$ of all pairs each of which consists
of a system $\mathcal{W}_\theta$ from $\mathcal{W}|_{\mathcal{S}}$
and a conservation law $F$ of this system. In view of Proposition
\ref{PropositionOnInducedMapping}, action of transformations from
$G^{\sim}$ on $\mathcal{W}|_{\mathcal{S}}$ and
$\{\CV(\mathcal{W}_\theta) | \theta \in \mathcal{S}\}$ together with
the pure equivalence relation of conserved vectors naturally
generates an equivalence relation on $P$.

\begin{definition}\label{DefGroupEquiConserLawOfDiffEq}
Let $\theta, \theta'\in \mathcal{S}$, $\mathcal{F}\in
\CL(\mathcal{W}_\theta), \mathcal{F'}\in
\CL(\mathcal{W}_{\theta'}), F\in \mathcal{F}, F'\in \mathcal{F'}$. The pairs $(\mathcal{W}_\theta,
\mathcal{F})$ and $(\mathcal{W}_{\theta'}, \mathcal{F'})$ are called
$G^{\sim}$-equivalent if there exists a transformation $g\in G^{\sim}$
which transforms
the system $\mathcal{W}_\theta$ to the system $\mathcal{W}_{\theta'}$ and such that the conserved vectors $F_g$ and $F'$ are equivalent
in the sense of Definition \ref{DefinitionOfConsVectorEquivalence}.
\end{definition}

In such a way, classification of conservation laws with respect to $G^{\sim}$ will
be understood as classification in $P$ with respect to the above
equivalence relation. This problem can be investigated in the way
that is similar to group classification in classes of systems of
differential equations, especially it is formulated in terms of
characteristics. Namely, we construct firstly the conservation laws
that are defined for all values of the arbitrary elements. (The
corresponding conserved vectors may depend on the arbitrary
elements.) Then we classify, with respect to the equivalence group,
arbitrary elements for each of that the system admits additional
conservation laws.

For more detail and rigorous proof of the correctness of the above
definitions and statements
see~\cite{Popovych&Ivanova2004ConsLaws,Ivanova&Popovych&Sophocleous2007-3}.\\

In what follows, we use the most direct method described in~\cite{Popovych&Ivanova2004ConsLaws} to derive
the conservation law of class
\eqref{eqVarCoefTelegraphEq}. Due to using
the transformation $\tilde t = t, \tilde x =\int \frac{dx}{g(x)},
\tilde u = u$ from $G^{\sim}/G^{\sim g}$ in section \ref{SectiononEquivaTrans}, we can reduce equation
\eqref{eqVarCoefTelegraphEq} to one which has the same form with equation \eqref{g=1eqVarCoefTelegraphEq}.
Thus, without loss of generality we can restrict ourselves to investigation conservation law of the equation \eqref{g=1eqVarCoefTelegraphEq}.

\begin{theorem}\label{TheoremOfCLsClasWithG1}
A complete list of $G_1^{\sim}$-inequivalent equations~\eqref{g=1eqVarCoefTelegraphEq} having nontrivial
conservation laws with characteristics of the zeroth order is exhausted by ones given in table~\ref{TableCLsClas}.
\end{theorem}

\setcounter{tbn}{0}

\setcounter{clnumber}{0}
\begin{center}\footnotesize\renewcommand{\arraystretch}{1.15}
Table~\refstepcounter{table}\label{TableCLsClas}\thetable.
Conservation laws of equations~\eqref{g=1eqVarCoefTelegraphEq} \\[1ex]
\begin{tabular}{|l|c|c|c|c|l|}
\hline
N & $H(u)$ & $K(u)$ & $f(x)$ & $h(x)$ & \hfil Basis conservation laws \\
\hline
\refstepcounter{tbn}\label{CaseCLsForAllHForAllKForAllfForh1}\thetbn & $\forall$ & $\forall$ &  $\forall$ & $1$ &
\refstepcounter{clnumber}CL$^{\theclnumber}$, \refstepcounter{clnumber}CL$^{\theclnumber}$ \\
\refstepcounter{tbn}\label{CaseCLsForH1Kneq0hh-1Forallh}\thetbn & $1$ & $K_u\neq 0$ & $h(h^{-1})_{xx}$ & $\forall$ &
\refstepcounter{clnumber}CL$^{\theclnumber}$, \refstepcounter{clnumber}CL$^{\theclnumber}$\\
\refstepcounter{tbn}\label{CaseCLsForAllHKH+1hyForallh}\thetbn & $\forall$ & $\epsilon H+1$ & $-h_y$ &  $\forall$ &
\refstepcounter{clnumber}CL$^{\theclnumber}$, \refstepcounter{clnumber}CL$^{\theclnumber}$\\
\refstepcounter{tbn}\label{CaseCLsForAllHKH+1hyhy-1ForAllh}\thetbn & $\forall$ & $\epsilon H+1$ & $-h_y-hy^{-1}$  &  $\forall$ &
\refstepcounter{clnumber}CL$^{\theclnumber}$, \refstepcounter{clnumber}CL$^{\theclnumber}$\\
\refstepcounter{tbn}\label{CaseCLsForAllHKH+1hw-1hw-1}\thetbn & $\forall$ & $\epsilon H+1$ & $-h\omega^{-1}$  &  $\omega^{-1/2}\exp\big{(}-\int \frac{1}{2}(a_{00}+a_{11}\big{)}\omega^{-1})$ &
\refstepcounter{clnumber}CL$^{\theclnumber}$, \refstepcounter{clnumber}CL$^{\theclnumber}$, \refstepcounter{clnumber}CL$^{\theclnumber}$, \refstepcounter{clnumber}CL$^{\theclnumber}$ \\
\refstepcounter{tbn}\label{CaseCLsForAllHKHForAllfForAllH}\thetbn & $\forall$ & $\epsilon H$ &   $\forall$ &  $\forall$ &
\refstepcounter{clnumber}CL$^{\theclnumber}$, \refstepcounter{clnumber}CL$^{\theclnumber}$, \refstepcounter{clnumber}CL$^{\theclnumber}$, \refstepcounter{clnumber}CL$^{\theclnumber}$ \\
\refstepcounter{tbn}\label{CaseCLsForH1ForAllfForAllh}\thetbn & $1$ & $0$ &  $\forall$ & $\forall$ &
\refstepcounter{clnumber}CL$^{\theclnumber}$ \\
\hline
\end{tabular}
\end{center}
\smallskip
Here the conserved densities $F^1$ and fluxes $F^2$ of the presented conservation laws have the following forms:
\\[1ex] \setcounter{clnumber}{0}
\refstepcounter{clnumber} CL$^{\theclnumber}$:\quad $fu_t$, $-(Hu_x+\int K)$;\\[0.5ex]
\refstepcounter{clnumber} CL$^{\theclnumber}$:\quad $f(tu_t-u)$, $-t(Hu_x+\int K)$;\\[0.5ex]
\refstepcounter{clnumber} CL$^{\theclnumber}$:\quad $e^{-t}(h^{-1})_{xx}(u_t+u)$, $-e^{-t}[h^{-1}u_x-(h^{-1})_{x}u+\int K]$;\\[0.5ex]
\refstepcounter{clnumber} CL$^{\theclnumber}$:\quad $e^{t}(h^{-1})_{xx}(u_t-u)$, $-e^{t}[h^{-1}u_x-(h^{-1})_{x}u+\int K]$;\\[0.5ex]
\refstepcounter{clnumber} CL$^{\theclnumber}$:\quad $-e^{t-\epsilon \int h}h_{y}(u_t-u)$, $-e^{t}(e^{-\epsilon \int h}H u_y+hu)$;\\[0.5ex]
\refstepcounter{clnumber} CL$^{\theclnumber}$:\quad $-e^{-t-\epsilon \int h}h_{y}(u_t+u)$, $-e^{-t}(e^{-\epsilon \int h}H u_y+hu)$;\\[0.5ex]
\refstepcounter{clnumber} CL$^{\theclnumber}$:\quad $-e^{t-\epsilon \int h}y(h_{y}+y^{-1}h)(u_t-u)$, $-e^{t}(e^{-\epsilon \int h}yH u_y-\int H+yhu)$;\\[0.5ex]
\refstepcounter{clnumber} CL$^{\theclnumber}$:\quad $-e^{-t-\epsilon \int h}y(h_{y}+y^{-1}h)(u_t+u)$, $-e^{-t}(e^{-\epsilon \int h}yH u_y-\int H+yhu)$;\\[0.5ex]
\refstepcounter{clnumber} CL$^{\theclnumber}$:\quad $e^{-\epsilon \int h}f[(\alpha^{11}y+\alpha^{10})u_t-(\alpha_t^{11}y+\alpha_t^{10})u]$,
 $-(\alpha^{11}y+\alpha^{10})(e^{-\epsilon \int h}Hu_y+hu)+\alpha^{11}\int H$;\\[0.5ex]
\refstepcounter{clnumber} CL$^{\theclnumber}$:\quad $e^{-\epsilon \int h}f[(\alpha^{21}y+\alpha^{20})u_t-(\alpha_t^{21}y+\alpha_t^{20})u]$,
 $-(\alpha^{21}y+\alpha^{20})(e^{-\epsilon \int h}Hu_y+hu)+\alpha^{21}\int H$;\\[0.5ex]
\refstepcounter{clnumber} CL$^{\theclnumber}$:\quad
  $e^{-\epsilon \int h}f[(\alpha^{31}y+\alpha^{30})u_t-(\alpha_t^{31}y+\alpha_t^{30})u]$,
 $-(\alpha^{31}y+\alpha^{30})(e^{-\epsilon \int h}Hu_y+hu)+\alpha^{31}\int H$;\\[0.5ex]
\refstepcounter{clnumber} CL$^{\theclnumber}$:\quad
 $e^{-\epsilon \int h}f[(\alpha^{41}y+\alpha^{40})u_t-(\alpha_t^{41}y+\alpha_t^{40})u]$,
 $-(\alpha^{41}y+\alpha^{40})(e^{-\epsilon \int h}Hu_y+hu)+\alpha^{41}\int H$;\\[0.5ex]
\refstepcounter{clnumber} CL$^{\theclnumber}$:\quad
 $e^{-\epsilon \int h}fu_t$, $-e^{-\epsilon \int h}H u_y$;\\[0.5ex]
\refstepcounter{clnumber} CL$^{\theclnumber}$:\quad
 $e^{-\epsilon \int h}f(tu_t-u)$, $-e^{-\epsilon \int h}tH u_y$;\\[0.5ex]
\refstepcounter{clnumber} CL$^{\theclnumber}$:\quad
  $e^{-\epsilon \int h}yfu_t$, $-e^{-\epsilon \int h}yH u_y+\int H$;\\[0.5ex]
\refstepcounter{clnumber} CL$^{\theclnumber}$:\quad
 $e^{-\epsilon \int h}yf(tu_t-u)$, $-e^{-\epsilon \int h}tyH u_y+t\int H$;\\ [0.5ex]
\refstepcounter{clnumber} CL$^{\theclnumber}$:\quad
 $f(\sigma u_t-\sigma_tu)$, $-\sigma u_x+\sigma_xu$;\\ [0.5ex]
\smallskip
the variable $y$ is implicitly determined by the formula $x=\int
e^{-\epsilon \int h(y)dy} dy$; $\epsilon, a_{ij}=\const, i,j=0,1$;
$(\alpha^{k1}, \alpha^{k0})=(\alpha^{k1}(t), \alpha^{k0}(t)), k =
1,2,3,4$ is one of fundamental solution system of the system of ODEs
\[
\alpha^1_{tt}=-(a_{11}\alpha^1+a_{01}\alpha_0),\quad\quad \alpha^0_{tt}=-(a_{10}\alpha^1+a_{00}\alpha_0)
\]
$\omega= a_{01}y^2 + (a_{00}-a_{11})y-a_{10}$; $\sigma= \sigma(t,
x)$ is an arbitrary solution of the linear equation
$f\sigma_{tt}-\sigma_{xx}=0$. Hereafter $\int H=\int H du, \int
K=\int K du$. In case \ref{CaseCLsForAllHKHForAllfForAllH}
$\epsilon\in \{0, 1\} \mod G_1^{\sim}$.
\\

In Theorem \ref{TheoremOfCLsClasWithG1}, the classification of conservation laws is performed
with respect to the usual equivalence group $G_1^{\sim}$ and thus some obtained results are explicit. Hence, we can further formulate
the classification result in an implicit form, and indeed we can split case \ref{CaseCLsForAllHKH+1hw-1hw-1} of table \ref{TableCLsClas}
into a number of inequivalent cases depending on values of $a_{ij}$ . At the same time, using the
extended equivalence group $\hat G_1^{\sim}$, we can present the result of classification in a closed and simple
form with a smaller number of inequivalent equations having nontrivial conservation laws.

\begin{theorem}\label{TheoremOfCLsClasWithGeneG1}
A complete list of $\hat G_1^{\sim}$-inequivalent equations~\eqref{g=1eqVarCoefTelegraphEq} having nontrivial
conservation laws with characteristics of the zeroth order is exhausted by ones given in table~\ref{TableCLsClasWithGeneG1}.
\end{theorem}

\setcounter{tbn}{0}

\setcounter{clnumber}{0}
\begin{center}\footnotesize\renewcommand{\arraystretch}{1.15}
Table~\refstepcounter{table}\label{TableCLsClasWithGeneG1}\thetable.
Conservation laws of equations~\eqref{g=1eqVarCoefTelegraphEq} \\[1ex]
\begin{tabular}{|l|c|c|c|c|l|}
\hline
N & $H(u)$ & $K(u)$ & $f(x)$ & $h(x)$ & \hfil Basis conservation laws \\
\hline
\refstepcounter{tbn}\label{CaseCLsForAllHForAllKForAllfForh1Gene}\thetbn & $\forall$ & $\forall$ &  $\forall$ & $1$ &
\refstepcounter{clnumber}CL$^{\theclnumber}$, \refstepcounter{clnumber}CL$^{\theclnumber}$ \\
\refstepcounter{tbn}\label{CaseCLsForH1Kneq0hh-1ForallhGene}\thetbn & $1$ & $K_u\neq 0$ & $h(h^{-1})_{xx}$ & $\forall$ &
\refstepcounter{clnumber}CL$^{\theclnumber}$, \refstepcounter{clnumber}CL$^{\theclnumber}$\\
\refstepcounter{tbn}\label{CaseCLsForAllHK1fhxForAllh}\thetbn & $\forall$ & $1$ & $-h_x$ &  $\forall$ &
\refstepcounter{clnumber}CL$^{\theclnumber}$, \refstepcounter{clnumber}CL$^{\theclnumber}$\\
\refstepcounter{tbn}\label{CaseCLsForAllHK1fhxhx-1ForAllh}\thetbn &
$\forall$ & $1$ & $-h_x-hx^{-1}$  &  $\forall$ &
\refstepcounter{clnumber}CL$^{\theclnumber}$, \refstepcounter{clnumber}CL$^{\theclnumber}$\\
\refstepcounter{tbn}\label{CaseCLsForAllHK0ForAllfForAllh}\thetbn a & $\forall$ & $0$ & $\forall$  &  $\forall$ &
\refstepcounter{clnumber}CL$^{\theclnumber}$, \refstepcounter{clnumber}CL$^{\theclnumber}$, \refstepcounter{clnumber}CL$^{\theclnumber}$, \refstepcounter{clnumber}CL$^{\theclnumber}$ \\
\thetbn b & $\forall$ & $1$ & $1$  &  $1$ &
\refstepcounter{clnumber}CL$^{\theclnumber}$, \refstepcounter{clnumber}CL$^{\theclnumber}$, \refstepcounter{clnumber}CL$^{\theclnumber}$, \refstepcounter{clnumber}CL$^{\theclnumber}$ \\
\thetbn c & $\forall$ & $1$ & $e^x$  &  $e^x$ &
\refstepcounter{clnumber}CL$^{\theclnumber}$, \refstepcounter{clnumber}CL$^{\theclnumber}$, \refstepcounter{clnumber}CL$^{\theclnumber}$, \refstepcounter{clnumber}CL$^{\theclnumber}$ \\
\thetbn d & $\forall$ & $1$ & $x^{\mu-1}$  &  $x^{\mu}$ &
\refstepcounter{clnumber}CL$^{\theclnumber}$, \refstepcounter{clnumber}CL$^{\theclnumber}$, \refstepcounter{clnumber}CL$^{\theclnumber}$, \refstepcounter{clnumber}CL$^{\theclnumber}$ \\
\refstepcounter{tbn}\label{CaseCLsForAllHK1fxexphxexp}\thetbn a & $\forall$ & $1$ &   $\frac{1}{x^3}e^{-\frac{1}{x}}$ &  $\frac{1}{x}e^{-\frac{1}{x}}$ &
\refstepcounter{clnumber}CL$^{\theclnumber}$, \refstepcounter{clnumber}CL$^{\theclnumber}$, \refstepcounter{clnumber}CL$^{\theclnumber}$, \refstepcounter{clnumber}CL$^{\theclnumber}$ \\
\thetbn b & $\forall$ & $1$ &   $\frac{1}{x^3}$ &  $\frac{1}{x}$ &
\refstepcounter{clnumber}CL$^{\theclnumber}$, \refstepcounter{clnumber}CL$^{\theclnumber}$, \refstepcounter{clnumber}CL$^{\theclnumber}$, \refstepcounter{clnumber}CL$^{\theclnumber}$ \\
\refstepcounter{tbn}\label{CaseCLsForAllHK1fx-1x+1hx-1x+1}\thetbn & $\forall$ & $1$ &   $|x-1|^{\mu-\frac{3}{2}}|x+1|^{-\mu-\frac{3}{2}}$
 &  $|x-1|^{\mu-\frac{1}{2}}|x+1|^{-\mu-\frac{1}{2}}$ &
\refstepcounter{clnumber}CL$^{\theclnumber}$, \refstepcounter{clnumber}CL$^{\theclnumber}$, \refstepcounter{clnumber}CL$^{\theclnumber}$, \refstepcounter{clnumber}CL$^{\theclnumber}$ \\
\refstepcounter{tbn}\label{CaseCLsForAllHK1fexparctanhexparctan}\thetbn
& $\forall$ & $1$ &   $e^{\mu \arctan x}(x^2+1)^{-\frac{3}{2}}$
 & $e^{\mu \arctan x}(x^2+1)^{-\frac{1}{2}}$ &
\refstepcounter{clnumber}CL$^{\theclnumber}$, \refstepcounter{clnumber}CL$^{\theclnumber}$, \refstepcounter{clnumber}CL$^{\theclnumber}$, \refstepcounter{clnumber}CL$^{\theclnumber}$ \\
\refstepcounter{tbn}\label{CaseCLsForH1K0ForAllfForAllh}\thetbn & $1$ & $0$ &  $\forall$ & $\forall$ &
\refstepcounter{clnumber}CL$^{\theclnumber}$ \\
\hline
\end{tabular}
\end{center}
\smallskip
Here the conserved densities $F^1$ and fluxes $F^2$ of the presented conservation laws have the following forms:
\\[1ex] \setcounter{clnumber}{0}
\refstepcounter{clnumber} CL$^{\theclnumber}$:\quad $fu_t$, $-(Hu_x+\int K)$;\\[0.5ex]
\refstepcounter{clnumber} CL$^{\theclnumber}$:\quad $f(tu_t-u)$, $-t(Hu_x+\int K)$;\\[0.5ex]
\refstepcounter{clnumber} CL$^{\theclnumber}$:\quad $e^{-t}(h^{-1})_{xx}(u_t+u)$, $-e^{-t}[h^{-1}u_x-(h^{-1})_{x}u+\int K]$;\\[0.5ex]
\refstepcounter{clnumber} CL$^{\theclnumber}$:\quad $e^{t}(h^{-1})_{xx}(u_t-u)$, $-e^{t}[h^{-1}u_x-(h^{-1})_{x}u+\int K]$;\\[0.5ex]
\refstepcounter{clnumber} CL$^{\theclnumber}$:\quad $e^{t}f(u_t-u)$, $-e^{t}(H u_x+h\int K)$;\\[0.5ex]
\refstepcounter{clnumber} CL$^{\theclnumber}$:\quad $e^{-t}f(u_t+u)$, $-e^{-t}(H u_x+h\int K)$;\\[0.5ex]
\refstepcounter{clnumber} CL$^{\theclnumber}$:\quad $e^{t}xf(u_t-u)$, $-e^{t}(xH u_x-\int H+xhu)$;\\[0.5ex]
\refstepcounter{clnumber} CL$^{\theclnumber}$:\quad $e^{-t}xf(u_t+u)$, $-e^{-t}(xH u_x-\int H+xhu)$;\\[0.5ex]
\refstepcounter{clnumber} CL$^{\theclnumber}$:\quad $fu_t$, $-Hu_x$;\\[0.5ex]
\refstepcounter{clnumber} CL$^{\theclnumber}$:\quad $f(tu_t-u)$, $-tHu_x$;\\[0.5ex]
\refstepcounter{clnumber} CL$^{\theclnumber}$:\quad $xfu_t$, $-xHu_x+\int H$;\\[0.5ex]
\refstepcounter{clnumber} CL$^{\theclnumber}$:\quad $xf(tu_t-u)$, $-t(xHu_x-\int H)$;\\[0.5ex]
\refstepcounter{clnumber} CL$^{\theclnumber}$:\quad $u_t$, $-Hu_x-u$;\\[0.5ex]
\refstepcounter{clnumber} CL$^{\theclnumber}$:\quad $tu_t-u$, $-t(Hu_x+u)$;\\[0.5ex]
\refstepcounter{clnumber} CL$^{\theclnumber}$:\quad $(x-\frac{1}{2}t^2)u_t+tu$, $-(x-\frac{1}{2}t^2)(Hu_x+u)+\int H$;\\[0.5ex]
\refstepcounter{clnumber} CL$^{\theclnumber}$:\quad $(tx-\frac{1}{6}t^3)u_t-(x-\frac{1}{2}t^2)u$, $-(tx-\frac{1}{6}t^3)(Hu_x+u)+t\int H$;\\ [0.5ex]
\refstepcounter{clnumber} CL$^{\theclnumber}$:\quad $e^x(u_t\cos t+u\sin t)$, $-\cos t(Hu_x+e^xu)$;\\ [0.5ex]
\refstepcounter{clnumber} CL$^{\theclnumber}$:\quad $e^x(u_t\sin t-u\cos t)$, $-\sin t(Hu_x+e^xu)$;\\ [0.5ex]
\refstepcounter{clnumber} CL$^{\theclnumber}$:\quad $(x\sin t+\frac{1}{2} t\cos t)e^xu_t-(x\cos t+\frac{1}{2}\cos t)e^x u$, $-(x\sin t+\frac{1}{2} t\cos t)(Hu_x+e^xu)+\sin t\int H$;\\ [0.5ex]
\refstepcounter{clnumber} CL$^{\theclnumber}$:\quad $[(x-\frac{1}{2}) \cos t-\frac{1}{2} t\sin t]e^xu_t+(x\sin t+\frac{1}{2}t\cos t)e^x u$, $-[(x-\frac{1}{2}) \cos t-\frac{1}{2} t\sin t](Hu_x+e^xu)+\cos t\int H$;\\ [0.5ex]
\refstepcounter{clnumber} CL$^{\theclnumber}$:\quad $x^{\mu-1}[u_t\cos(\sqrt{\mu}t)+ \sqrt{\mu} u\sin(\sqrt{\mu}t)]$, $-\cos(\sqrt{\mu}t)(Hu_x+x^{\mu}u)$;\\ [0.5ex]
\refstepcounter{clnumber} CL$^{\theclnumber}$:\quad $x^{\mu-1}[u_t\sin(\sqrt{\mu}t)- \sqrt{\mu} u\cos(\sqrt{\mu}t)]$, $-\sin(\sqrt{\mu}t)(Hu_x+x^{\mu}u)$;\\ [0.5ex]
\refstepcounter{clnumber} CL$^{\theclnumber}$:\quad
$x^{\mu}[u_t\cos(\sqrt{\mu+1}t)+ \sqrt{\mu+1}
u\sin(\sqrt{\mu+1}t)]$,
$-\cos(\sqrt{\mu+1}t)(xHu_x+x^{\mu+1}u)\\+\cos(\sqrt{\mu+1}t)\int
H$;\\ [0.5ex]
\refstepcounter{clnumber} CL$^{\theclnumber}$:\quad
$x^{\mu}[u_t\sin(\sqrt{\mu+1}t)- \sqrt{\mu+1}
u\cos(\sqrt{\mu+1}t)]$,
$-\sin(\sqrt{\mu+1}t)(xHu_x+x^{\mu+1}u)\\+\sin(\sqrt{\mu+1}t)\int
H$;\\ [0.5ex]
\refstepcounter{clnumber} CL$^{\theclnumber}$:\quad $\frac{1}{x^2}e^{-\frac{1}{x}}(u_t\cos t+  u\sin t)$, $-\cos t(xHu_x+\frac{1}{x}e^{-\frac{1}{x}}u)+\cos t\int H$;\\ [0.5ex]
\refstepcounter{clnumber} CL$^{\theclnumber}$:\quad $\frac{1}{x^2}e^{-\frac{1}{x}}(u_t\sin t-  u\cos t)$, $-\sin t(xHu_x+\frac{1}{x}e^{-\frac{1}{x}}u)+\sin t\int H$;\\ [0.5ex]
\refstepcounter{clnumber} CL$^{\theclnumber}$:\quad $\frac{1}{x^3}e^{-\frac{1}{x}}[(-\frac{1}{2}tx\cos t+ \sin t)u_t-(-\frac{1}{2}x\cos t+\frac{1}{2}tx \sin t+\cos t)u]$, $-(-\frac{1}{2}tx\cos t+ \sin t)(Hu_x+\frac{1}{x}e^{-\frac{1}{x}}u)-\frac{1}{2}t\cos t\int H$;\\ [0.5ex]
\refstepcounter{clnumber} CL$^{\theclnumber}$:\quad $\frac{1}{x^3}e^{-\frac{1}{x}}[(\frac{1}{2}tx\sin t+\frac{1}{2}x \cos t+\cos t)u_t-(\frac{1}{2}tx\cos t-\sin t)u]$, $-(\frac{1}{2}tx\sin t+\frac{1}{2}x\cos t+ \cos t)(Hu_x+\frac{1}{x}e^{-\frac{1}{x}}u)+\frac{1}{2}(t\sin t+\cos t)\int H$;\\ [0.5ex]
\refstepcounter{clnumber} CL$^{\theclnumber}$:\quad
$\frac{1}{x^2}u_t$, $-xHu_x-u+\int H$;\\ [0.5ex]
\refstepcounter{clnumber} CL$^{\theclnumber}$:\quad
$\frac{1}{x^2}(tu_t-u)$, $-t(xHu_x+u-\int H)$;\\ [0.5ex]
\refstepcounter{clnumber} CL$^{\theclnumber}$:\quad
$(\frac{1}{2}\frac{t^2}{x^2}+\frac{1}{x^3})u_t-\frac{t}{x^2}u$,
$-(\frac{1}{2}t^2x+1)(Hu_x+\frac{1}{x}u)+\frac{1}{2}t^2\int H$;\\
[0.5ex]
\refstepcounter{clnumber} CL$^{\theclnumber}$:\quad
$(\frac{1}{6}\frac{t^3}{x^2}+\frac{t}{x^3})u_t-(\frac{1}{2}\frac{t^2}{x^2}+\frac{t}{x^3})u$,
$-(\frac{1}{6}t^3x+t)(Hu_x+\frac{1}{x}u)+\frac{1}{6}t^3\int H$;\\
[0.5ex]
\refstepcounter{clnumber} CL$^{\theclnumber}$:\quad
$(x+1)f[\sin(\sqrt{2\mu-1}t)u_t-\sqrt{2\mu-1}\cos(\sqrt{2\mu-1}t)u]$,
$-(x+1)\sin(\sqrt{2\mu-1}t)(Hu_x+hu)+\sin(\sqrt{2\mu-1}t)\int H$;\\
[0.5ex]
\refstepcounter{clnumber} CL$^{\theclnumber}$:\quad
$(x+1)f[\cos(\sqrt{2\mu-1}t)u_t+\sqrt{2\mu-1}\sin(\sqrt{2\mu-1}t)u]$,
$-(x+1)\cos(\sqrt{2\mu-1}t)(Hu_x+hu)+\cos(\sqrt{2\mu-1}t)\int H$;\\
[0.5ex]
\refstepcounter{clnumber} CL$^{\theclnumber}$:\quad
$(x-1)f[\sin(\sqrt{2\mu+1}t)u_t-\sqrt{2\mu+1}\cos(\sqrt{2\mu+1}t)u]$,
$-(x-1)\sin(\sqrt{2\mu+1}t)(Hu_x+hu)+\sin(\sqrt{2\mu+1}t)\int H$;\\
[0.5ex]
\refstepcounter{clnumber} CL$^{\theclnumber}$:\quad
$(x-1)f[\cos(\sqrt{2\mu+1}t)u_t+\sqrt{2\mu+1}\sin(\sqrt{2\mu+1}t)u]$,
$-(x-1)\cos(\sqrt{2\mu+1}t)(Hu_x+hu)+\cos(\sqrt{2\mu+1}t)\int H$;\\
[0.5ex]
\refstepcounter{clnumber} CL$^{\theclnumber}$:\quad
$e^{-\sqrt{-\mu-i}t}(x-i)f(u_t+\sqrt{-\mu-i}u)$,
$-e^{-\sqrt{-\mu-i}t}(x-i)(Hu_x+hu)+e^{-\sqrt{-\mu-i}t}\int H$;\\
[0.5ex]
\refstepcounter{clnumber} CL$^{\theclnumber}$:\quad
$e^{\sqrt{-\mu-i}t}(x-i)f(u_t-\sqrt{-\mu-i}u)$,
$-e^{\sqrt{-\mu-i}t}(x-i)(Hu_x+hu)+e^{\sqrt{-\mu-i}t}\int H$;\\
[0.5ex]
\refstepcounter{clnumber} CL$^{\theclnumber}$:\quad
$e^{-\sqrt{-\mu+i}t}(x+i)f(u_t+\sqrt{-\mu+i}u)$,
$-e^{-\sqrt{-\mu+i}t}(x+i)(Hu_x+hu)+e^{-\sqrt{-\mu+i}t}\int H$;\\
[0.5ex]
\refstepcounter{clnumber} CL$^{\theclnumber}$:\quad
$e^{\sqrt{-\mu+i}t}(x+i)f(u_t-\sqrt{-\mu+i}u)$,
$-e^{\sqrt{-\mu+i}t}(x+i)(Hu_x+hu)+e^{\sqrt{-\mu+i}t}\int H$;\\
[0.5ex]
\refstepcounter{clnumber} CL$^{\theclnumber}$:\quad
 $f(\sigma u_t-\sigma_tu)$, $-\sigma u_x+\sigma_xu$;\\ [0.5ex]
\smallskip
$\mu = \const, \sigma= \sigma(t, x)$ is an arbitrary solution of the
linear equation $f\sigma_{tt}-\sigma_{xx} = 0$.

\begin{proof}
We search the first-order conservation laws for the equations from
class \eqref{g=1eqVarCoefTelegraphEq} in the form
\begin{equation}\label{FirstOrderCL}
D_tF(t,x,u,u_t,u_x)+D_xG(t,x,u,u_t,u_x)=0,
\end{equation}
where $D_t$ and $D_x$ are the operators of the total differentiation
with respective to $t$ and $x$ correspondingly, namely,
$D_t=\partial_t+u_t\partial_u+\cdots,
D_x=\partial_x+u_x\partial_u+\cdots.$ Substituting the expression
for $u_{tt}$ deduced from \eqref{g=1eqVarCoefTelegraphEq} into
\eqref{FirstOrderCL} and decompose the obtained equation with
respect to $u_{xt}$ and $u_{xx}$, we obtain
\begin{gather}
F_{u_x}+G_{u_t}=0, \nonumber \\
F_{u_t}\frac{H}{f}+G_{u_x}=0, \label{System1FirstOrderCL} \\
F_{u_t}\frac{H_u}{f}u_x^2+(F_{u_t}\frac{hK}{f}+G_u)u_x+F_uu_t+F_t+G_x=0.\nonumber
\end{gather}
Up to conserved vectors equivalence for the first equation of system
\eqref{System1FirstOrderCL}, we can assume $F_{u_x} = G_{u_t}= 0$,
which together with the second equation imply
\begin{equation}\label{Expression1ForFG}
F=F^3(t,x,u)u_t+F^2(t,x,u),\quad
G=-F^3(t,x,u)\frac{H}{f}u_x+G^1(t,x,u).
\end{equation}
Substituting these expression into the last equation of system
\eqref{System1FirstOrderCL} and splitting it with respect to the
powers of $u_x$ and $u_t$, we obtain the system of PDEs for the functions
$F^3, F^2$ and $G^1$ of the form
\begin{gather}
F^3_{u}=0,\quad F^2_u+F^3_t=0, \nonumber \\
F^3\frac{hK}{f}+G^1_{u}-\frac{H}{f}F^3_x+F^3\frac{f_x}{f^2}H=0, \label{System2FirstOrderCL} \\
F^2_t+G^1_x=0.\nonumber
\end{gather}
Solving first three equations of \eqref{System2FirstOrderCL} yields
\begin{equation}\label{Expression2ForFG}
F^3=F^1(t,x),\quad F^2=-F^1_tu+F^0(t,x),\quad
G^1=(\frac{F^1}{f})_x\int H-\frac{hF^1}{f}\int K+G^0(t,x).
\end{equation}
Substituting the latter expression for $F^2$ and $G^1$ into the last
equation of system \eqref{System2FirstOrderCL}, we can know that the
major role for classification is played by a differential
consequence
\begin{equation}\label{ClassfyingCondForCL}
(\frac{F^1}{f})_{xx}H-(\frac{hF^1}{f})_xK-F^1_{tt}=0.
\end{equation}
Indeed, it is the unique classifying condition for this problem. In
all classification cases we obtain the equation $F^0_t +G^0_x = 0$.
Therefore, up to conserved vectors equivalence we can assume $F^0 =
G^0 = 0$, and additionally $F^1\neq 0$ for conservation laws to be
non-trivial. Thus,  taking into accountant \eqref{Expression1ForFG}
and \eqref{Expression2ForFG} the conservation density and flux can
be rewritten as
\begin{equation}\label{Expression3ForFG}
F=F^1(t,x)u_t-F^1_tu,\quad
G=-\frac{H}{f}F^1u_x+(\frac{F^1}{f})_x\int H-\frac{hF^1}{f}\int K.
\end{equation}
Equation \eqref{ClassfyingCondForCL} implies that there exist no
non-trivial conservation laws in the general case. Let us classify
the special values of the parameter-functions for which equation
\eqref{g=1eqVarCoefTelegraphEq} possesses non-trivial conservation
laws. There exist four different possibilities for values of $H$ and
$K$.

1. $\dim \langle h, H, K \rangle= 3$. It follows from
\eqref{ClassfyingCondForCL} that $F^1_{tt} = (F^1/f)_{xx} =
(hF^1/f)_x = 0$ and therefore $F^1 = \alpha(t)f/h, (1/h)xx = 0,
\alpha_{tt}=0$ i.e., obviously $h\in \{1, x^{-1}\}, \alpha=c_1+c_2t
\mod G_1^{\sim}$. Moreover, $h = 1\sim h = x^{-1} \mod G_1^{\sim}$
(the corresponding transformation is $\tilde x = \ln |x|$ and $\tilde
f = x^2f$, the other variables and parameter-functions are not
changed). As a result, we obtain case
\ref{CaseCLsForAllHForAllKForAllfForh1Gene}.

2. $H\in \langle 1 \rangle, K \bar \in \langle 1 \rangle$. Then $H=
1 \mod G_1^{\sim}$ and $(hF^1/f)_x = 0, F^1_{tt}-(F1/f)_{xx}= 0$,
i.e., $F^1 =\alpha(t)f/h$, where
$\alpha_{tt}/\alpha=\lambda=\const$(otherwise we have case
\ref{CaseCLsForAllHForAllKForAllfForh1Gene}) and so $\lambda= 1, f =
h(h^{-1})_{xx} \mod G_1^{\sim}$ (case
\ref{CaseCLsForH1Kneq0hh-1ForallhGene}).

3. $H \in \langle 1 \rangle, K\in  \langle H, 1 \rangle$. Then $K\in
\langle 1 \rangle \mod \hat G_1^{\sim}$ and $(F^1/f)_{xx} = 0,
F^1_{tt} = -K(hF^1/f)_x$, i.e. $F^1 = (\alpha^1(t)x + \alpha^0(t))f$
and $\alpha^1_{tt}xf +\alpha^0_{tt}f = -K(\alpha^1(xh)_x +
\alpha^0h_x)$. For $K= 0$ we obtain case
\ref{CaseCLsForAllHK0ForAllfForAllh}a at once. Suppose $K \neq 0$.
Then $K = 1 \mod G_1^{\sim}$ and the dimension $m = \dim \langle f,
xf, h_x, (xh)_x \rangle$ can have only the values $2$ and $3$.

If $m = 3$ then there exist constants $a_{ij}, b_i, i, j = 0, 1$,
and a function $\theta= \theta(x)$ such that $(b_0, b_1) \neq (0,
0)$, $\dim \langle f, xf, \theta \rangle = 3$ and $h_x = a_{00}f +
a_{01}xf + b_0\theta, (xh)_x = a_{10}f + a_{11}xf + b_1\theta$.
Therefore, $\alpha^1_{tt} =- a_{11}\alpha^1-a_{01}\alpha^0,
\alpha^0_{tt} =- a_{10}\alpha^1-a_{00}\alpha^0,
b_1\alpha^1+b_0\alpha^0 = 0$, i.e. $\alpha^1=c_1e^{\lambda
t}+c_2e^{-\lambda t}, \alpha^0=c_3e^{\delta t}+c_4e^{-\delta t}$
where $c_i (i=1,\cdots, 4), \lambda, \delta = \const$ and $\lambda,
\delta \neq 0$ (otherwise, this case is reduced to a subcase of
\ref{CaseCLsForAllHForAllKForAllfForh1Gene}), hence $\lambda, \delta
= 1 \mod G^{\sim}$. Depending on values (either vanishing or
non-vanishing) of $c_1, c_2$ and $c_3, c_4$ we obtain cases
\ref{CaseCLsForAllHK1fhxForAllh} and
\ref{CaseCLsForAllHK1fhxhx-1ForAllh} correspondingly.

If $m = 2$ then $h_x = a_{00}f + a_{01}xf, (xh)_x = a_{10}f +
a_{11}xf$ for some constants $a_{ij} , i, j = 0, 1$. Therefore,
$\alpha^1_{tt} =- a_{11}\alpha^1-a_{01}\alpha^0, \alpha^0_{tt} =-
a_{10}\alpha^1-a_{00}\alpha^0, f = -h/\omega, h_x/h =
-(a_{01}x+a_{00})/\omega$, where $\omega =
a_{01}x^2+(a_{00}-a_{11})x-a_{10}$, i.e., $h =
\omega^{-1/2}\exp\big{(}-\int
\frac{1}{2}(a_{00}+a_{11}\big{)}\omega^{-1})$. As a results, we
obtain four conservation laws with the conserved vectors
\[
(\alpha^{\mu 1}(t)x+\alpha^{\mu 0}(t))fu_t-(\alpha_t^{\mu
1}(t)x+\alpha_t^{\mu 0}(t))fu, -(\alpha^{\mu 1}(t)x+\alpha^{\mu
0}(t))(Hu_x+hu)+\alpha^{\mu 1}\int H),
\]
where $(\alpha^{\mu1}, \alpha^{\mu0}), \mu= 1,\cdots, 4$ form a
fundamental set of solutions of the system $\alpha^1_{tt} =-
a_{11}\alpha^1-a_{01}\alpha^0, \alpha^0_{tt} =-
a_{10}\alpha^1-a_{00}\alpha^0$. Separate consideration of possible
inequivalent values of the constants $a_{ij}$ leads to cases
\ref{CaseCLsForAllHK0ForAllfForAllh}b-\ref{CaseCLsForAllHK0ForAllfForAllh}d
and
\ref{CaseCLsForAllHK1fxexphxexp}-\ref{CaseCLsForAllHK1fexparctanhexparctan}.

4. $H,K \in \langle 1 \rangle$. Therefore, $H = 1, K= 0 \mod
G^{\sim}$ and $F^1_{tt} -(F^1/f)_{xx} = 0$. Let $F^1=\sigma(t,x)f$,
we have $\sigma_{xx}-f\sigma_{tt}=0$, which corresponding case
\ref{CaseCLsForH1K0ForAllfForAllh}.

\end{proof}

The above conservation laws can be used for construction of
potential systems, potential symmetries and potential conservation
laws. We will present such analysis elsewhere.

\section{Conclusion and Remarks}\label{SectionOnConclusion}

In summary, we have present an enhanced classical and nonclassical
symmetries and conservation laws analysis of the class of
equations~\eqref{eqVarCoefTelegraphEq} in the framework of modern
group analysis of differential equations.

We have performed a complete and extended symmetry group
classification of the class of
equations~\eqref{eqVarCoefTelegraphEq} with the two ``best" gauges
$g = 1$ and $g = h$. The main results on classification are
collected in tables~
\ref{TableGrClasForAllHg=1}--\ref{TableGrClasHpowerg=1} ($g=1$) and
\ref{TableGrClasForAllH}--\ref{TableGrClasHpower} ($g=h$) where we
list inequivalent cases of extensions with the corresponding Lie
invariance algebras. The success in the classification and the clear
presentation of the final results are relied heavily on the regular
applications of four original tools presented in
\cite{Ivanova&Popovych&Sophocleous2007}, i.e., the equivalence
relation with respect to the extended equivalence group instead of
the usual one, the choice of true gauges, furcate split and
systematic usage of additional equivalences. Among them the first
two kinds of techniques (the extended equivalence group and true
gauges) are of crucial importance for obtaining a closed and
explicit classification list. The extended equivalence group of
class \eqref{eqVarCoefTelegraphEq} is the extension of the usual one
with the non-trivial group of gauge equivalence transformations
including transformations which are nonlocal in arbitrary elements.
Neglecting this transformations leads to critical swelling and
complication of both calculations and results. This can be seen from
the classification results of equation \eqref{eqVarCoefTelegraphEq}
with the gauge $g=1$ under the usual equivalence transformations
(adduced in the Appendix). As an application of the classification
results, exact solutions of some classification models are given by
using the method of classical Lie reduction.

Nonclassical symmetries of equation \eqref{eqVarCoefTelegraphEq} are
discussed within the framework of singular reduction operator.
Determining equations related the general singular and regular
reduction operators of \eqref{eqVarCoefTelegraphEq} with $g=1$ are
given. Several reduction operators of a special nonlinear telegraph
equation (equation \eqref{eqVarCoefTelegraphEq} with $H(u)=K(u)=u,
f(x)=h(x)=1$) are present. This enabled to obtain some exact
solutions of the corresponding equation which are invariant under
these conditional symmetries.

Using the most direct method, we have also investigated two
classifications of local conservation laws up to equivalence
relations which are generated by both usual and extended equivalence
groups. Equivalence with respect to these groups and correct choice
of gauge coefficients of equations play the major role for simple
and clear formulation of the final results.

It should be noted that all above results of class \eqref{eqVarCoefTelegraphEq}
can be applied to the specific models
listed in section \ref{SectionOnPhysExam}.
We do not discuss here because there are nothing but tendinous computations.

One of the natural continuation for further investigation of
different properties of class~\eqref{eqVarCoefTelegraphEq} is to
perform group classifications of the variable gauges , i.e., when
the value of $g$ ($1$ or $h$) depends on values of other arbitrary
elements,  in a way similar to
\cite{Ivanova&Popovych&Sophocleous2007-2}. We can also make a
further studies of non-Lie exact solutions of
class~\eqref{eqVarCoefTelegraphEq} by means of the classification of
singular and regular reduction operators. Furthermore, the proposed
classifications of local conservation laws of
class~\eqref{eqVarCoefTelegraphEq} can be used to find all possible
inequivalent potential systems and potential conservation laws (see
\cite{Ivanova&Popovych&Sophocleous2007-4} for detail) associated to
the given system of differential equations. These problems will be investigated
in subsequent publication.

\subsection*{Acknowledgements}
D.j. Huang express his sincerely thanks
to professor Nataliya Ivanova for stimulating discussions
and correction of the results of this paper.
This work was partially supported by the National Key Basic Research Project
of China under Grant No. 2010CB126600, the National Natural Science
Foundation of China under Grant No. 60873070, Shanghai Leading
Academic Discipline Project No. B114, the Postdoctoral Science
Foundation of China under Grant No. 20090450067, Shanghai
Postdoctoral Science Foundation under Grant No. 09R21410600 and
the Fundamental Research Funds for the Central Universities under Grant No. WM0911004.

\begin{appendix}

\section{Note on classification with respect to the usual equivalence group}

In tables~\ref{TableUsualGrClasForAllH}--\ref{TableUsualGrClasHpower} we
list all possible $G^{\sim}$-inequivalent sets of functions
$f(x)$, $g(x)$, $H(u)$, $K(u)$ and corresponding invariance
algebras.

\setcounter{tbn}{0}

\begin{center}\footnotesize\renewcommand{\arraystretch}{1.15}
Table~\refstepcounter{table}\label{TableUsualGrClasForAllH}\thetable. Case of $\forall H(u)$ \\[1ex]
\begin{tabular}{|l|c|c|c|l|}
\hline
N & $K(u)$ & $f(x)$ & $g(x)$ & \hfil Basis of A$^{\max}$ \\
\hline
\refstepcounter{tbn}\label{CaseUsualForAllHForAllKForallF}\thetbn & $\forall$ & $\forall$ & $\forall$ & $\p_t$ \\
\refstepcounter{tbn}\label{CaseUsualForAllHForAllKFexp1}\thetbn
&$\forall$ & $|x|^\lambda$ & $x^{-1}$ & $\partial_t,\, \frac{1}{2}(\lambda+2)t\partial_t+x\partial_x$ \\
\refstepcounter{tbn}\label{CaseUsualForAllHForAllKFexp11}\thetbn
&$\forall$ & $1$ & $(\alpha+\beta x)^{-1}$ & $\partial_t,\, \partial_x, t\partial_t+x\partial_x$ \\
\refstepcounter{tbn}\label{CaseUsualForAllHForAllKFexp111}\thetbn
&$H+\beta$ & $(\alpha+\beta x)^{-\frac{2}{\beta}}$ &
$(\alpha+\beta x)^{-1}$
& $\partial_t,\, \partial_x, (\beta-1)t\partial_t+\beta x\partial_x$ \\
\refstepcounter{tbn}\label{CaseUsualForAllHForAllKFexp2}\thetbn
 & $H$ &$\forall$ & $-f'/2f+\alpha \sqrt{|f|}$& $\partial_t,  |f|^{-\frac{1}{2}}\partial_x$ \\
\refstepcounter{tbn}\label{CaseUsualForAllHForAllKFexp3}\thetbn & $H$ &
$\forall$ & $-\frac{f'}{2f}+\alpha\frac{\sqrt{|f|}}{\int
\sqrt{|f|} dx}$ & $\partial_t,
\frac{1}{2}t\partial_t+\frac{\int \sqrt{|f|} dx }{2\sqrt{|f|}}\partial_x$ \\
\refstepcounter{tbn}\label{CaseUsualForAllHForAllKFexp4}\thetbn & $H$ & $\forall$ &$-\frac{f'}{2f}$
& $\partial_t, \frac{1}{2}t\partial_t+[\frac{1}{2}|f|^{-\frac{1}{2}}\int|f|^{\frac{1}{2}}dx]\partial_x,
|f|^{-\frac{1}{2}}\partial_x$ \\
\refstepcounter{tbn}\label{CaseUsualForAllHK0F1}\thetbn & 1 & 1 & $x^{-1}$ & $\partial_t, t\partial_t+x\partial_x$\\
\hline
\end{tabular}
\end{center}
{\footnotesize Here $\lambda\neq 0$ mod $G^{\sim}$, $\alpha,
\beta\neq 0$.\\ }

\bigskip

\setcounter{tbn}{0}

\begin{center}\footnotesize\renewcommand{\arraystretch}{1.15}
Table~\refstepcounter{table}\label{TableUsualGrClasHexp}\thetable. Case of $H(u)=e^{\mu u}$ \\[1ex]
\begin{tabular}{|l|c|c|c|c|l|}
\hline
N & $\mu$ & $K(u)$ & $f(x)$ & $g(x)$ & \hfil Basis of A$^{\max}$ \\
\hline \refstepcounter{tbn}\label{CaseUsualHexpKexpnuuFxlambda}\thetbn
& $1$ &  $e^{\nu u}$ & $|x|^p$ & $|x|^q$ &
$\partial_t,\, \frac{1}{2}(p-p\nu-2\nu-q+1)t\partial_t $ \\
& $~$ &  $~$ & $~$ & $~$ &
$+(1-\nu)x\partial_x+(q+1)\partial_u $ \\
\refstepcounter{tbn}\label{CaseUsualHexpKexpnuuF1}\thetbn & $1$ &
$e^{\nu u}$  &
$e^{px}$& $\epsilon e^{qx}$ & $\partial_t,\, \frac{1}{2}(p-p\nu-q)t\partial_t+(1-\nu)\partial_x+q\partial_u$  \\
\refstepcounter{tbn}\label{CaseUsualHexpK1Fxlambda}\thetbn & $1$ & $1$
& $|x|^p $ & $|x|^q$ & $\partial_t,\,
\frac{1}{2}(1+p-q)t\partial_t+x\partial_x+(q+1)\partial_u$  \\
\refstepcounter{tbn}\label{CaseUsualHexpK1F1}\thetbn &
$1$ & $1$ &
$e^{px}$ & $\epsilon e^{qx} $ & $\partial_t, \frac{1}{2}(p-q)t\partial_t+\partial_x+q\partial_u$  \\
\refstepcounter{tbn}\label{CaseUsualHexpK0expuForAllF}\thetbn & $1$ &
$e^{\nu u}+h_1e^u$ & $\alpha \sqrt{|f|}-\frac{f'}{2f}$ &
$\frac{1}{h_1}(\frac{f'}{2f\sqrt{|f|}})' $ & $\partial_t,
\frac{1}{2}t\partial_t+(1-\nu)\frac{1}{\sqrt{|f|}}\partial_x-\partial_u$  \\
& $~$ & $~$ & $=(\frac{f'}{2f\sqrt{|f|}})'$ & $~$ & $~$\\
\refstepcounter{tbn}\label{CaseUsualHexpK0Ff1}\thetbn & $1$ & $e^{\nu
u}+h_1e^u$ & $\alpha\frac{\sqrt{|f|}}{\int \sqrt{|f|}
dx}-\frac{f'}{2f}$ &
$\frac{1}{h_1}(\frac{f'\int\sqrt{|f|}dx}{2f\sqrt{|f|}})' $ &
$\partial_t,\frac{1}{2}(1-2\nu)t\partial_t+(1-\nu)\frac{\int \sqrt{|f|} dx}{\sqrt{|f|}}\partial_x+\partial_u$  \\
& $~$ & $~$ & $=(\frac{f'\int\sqrt{|f|}dx}{2f\sqrt{|f|}})'$ & $~ $
&
$~$  \\
\refstepcounter{tbn}\label{CaseUsualHexpK0F1}\thetbn & $1$ & $ue^{
u}+h_1e^u$ & $\frac{f'}{f}=-2(g^{-1})_xg$ & $(g^{-1})_{xx}=\beta g
$ &
$\partial_t,-\frac{1}{2}\beta t\partial_t+g^{-1}\partial_x+\beta \partial_u$  \\
\refstepcounter{tbn}\label{CaseUsualHexpKexpu1}\thetbn & $1$ & $e^{ u}$
& $\forall$ & $\forall$ &
$\partial_t,\frac{1}{2}t\partial_t-\partial_u$  \\
\refstepcounter{tbn}\label{CaseUsualHexpKexpu2}\thetbn & $1$ & $e^{ u}$
& $f^{1}(x)$ & $g^{1}(x)$ &
$\tau=c_1+\frac{1}{2}c_2t, \xi=\xi^{1}(x), $  \\
& $~$ & $~$ & $~$ & $~$ &
$\eta=-c_2+(-\frac{f^{1}_x}{f^1}-\frac{4}{3}g^{1})\xi^{1}(x)$  \\
\refstepcounter{tbn}\label{CaseUsualHexpKexpu3}\thetbn & $1$ & $e^{ u}$
& $f^{2}(x)$ & $g^{2}(x)$ &
$\tau=c_1+\frac{1}{2}c_2t, \xi=\xi^{2}(x), $  \\
& $~$ & $~$ & $~$ & $~$ & $
\eta=-c_2+2c_3+(-\frac{f^{2}_x}{f^2}-\frac{4}{3}g^{2})\xi^{2}(x)$\\
\hline
\end{tabular}
\end{center}
{\footnotesize Here $\epsilon=\pm 1$ and $q, \alpha, \beta\neq 0$.
$f^1(x), g^1(x)$ and $\xi^1(x)$ satisfy the relation
\[
\Phi_{xx}+(2\Phi-\Psi)\Phi_x+\Psi_{xx}+(\Phi-2\Psi)\Psi_x-\Phi\Psi(\Phi+\Psi)=0,\quad
\xi^1_x-\Phi\xi^1=0,
\]
where
\[
\Phi=-\frac{1}{3}(2\frac{f^1_x}{f^1}+g^1),\quad
\Psi=(\frac{f^1_x}{3f^1}-g^1);
\]
and $f^2(x), g^2(x)$ and $\xi^2(x)$ satisfy the relation
\[
\Theta_x+\Lambda_x-\Theta\Lambda-\Lambda^2=0,\quad
\xi^2_x-\Theta\xi^2=c_3,
\]
where
\[
\Theta=-\frac{1}{3}(2\frac{f^2_x}{f^2}+g^2),\quad
\Lambda=(\frac{f^2_x}{3f^2}-g^2).
\]
}

\bigskip

The operators from
tables~\ref{TableUsualGrClasForAllH}--\ref{TableUsualGrClasHpower} form
bases of the maximal invariance algebras if the corresponding sets
of the functions $f$, $g$, $H$, $K$ are $G^{\sim}$-inequivalent to
ones with most extensive invariance algebras. For example, in case
$\ref{TableUsualGrClasHpower}.\ref{CaseUsualHpowerKpowerFpower}~ (\mu,
\nu)\neq(0, 0)$. Similarly, in case
\ref{TableUsualGrClasHexp}.\ref{CaseUsualHexpKexpnuuFxlambda} the constraint
set on the parameters $\mu, \nu$ and $\lambda$ coincides with the
one for case \ref{TableUsualGrClasHpower}.\ref{CaseUsualHpowerKpowerFpower},
and $\mu = 1$ if $\nu = 0$.

\setcounter{tbn}{0}

\begin{center}\footnotesize\renewcommand{\arraystretch}{1.15}
Table~\refstepcounter{table}\label{TableUsualGrClasHpower}\thetable. Case of $H(u)=u^{\mu}$ \\[1ex]
\begin{tabular}{|l|c|c|c|c|l|}
\hline
N & $\mu$ & $K(u)$ & $f(x)$ & $g(x)$  & \hfil Basis of A$^{\max}$ \\
\hline
\refstepcounter{tbn}\label{CaseUsualHpowerKpowerFpower}\thetbn &
$\forall$ & $u^{\nu}$ &
$|x|^p$ & $|x|^q $ & $\partial_t,\, \frac{1}{2}[(p-q+1)\mu-(p+2)\nu]t\partial_t$ \\
& $ $ & $ $ &
$ $ & $ $ & $+(\mu-\nu)x\partial_x+(q+1)u\partial_u$  \\
\refstepcounter{tbn}\label{CaseUsualHpowerKpowerF1}\thetbn & $\forall$
& $u^{\nu }$ & $e^{px}$ & $\epsilon e^{qx} $ & $\partial_t,\,
\frac{1}{2}[(p-q)\mu-p\nu]\partial_t+(\mu-\nu)\partial_x+qu\partial_u$\\
\refstepcounter{tbn}\label{CaseUsualHpowerK1Fpower} \thetbn & $\forall$
& $u^{\nu}+h_1u^{\mu}$ & $\alpha \sqrt{|f|}-\frac{f'}{2f}$ &
$\frac{1}{h_1}(\frac{f'}{2f\sqrt{|f|}})' $ & $\partial_t,
\frac{1}{2}\mu t\partial_t+(\mu-\nu)\frac{1}{\sqrt{|f|}}\partial_x-u\partial_u$  \\
& $~$ & $~$ & $=(\frac{f'}{2f\sqrt{|f|}})'$ & $~$ & $~$\\
\refstepcounter{tbn}\label{CaseUsualHpowerK1F1} \thetbn & $\forall$ &
$u^{\nu}+h_1u^{\mu}$ & $\alpha\frac{\sqrt{|f|}}{\int \sqrt{|f|}
dx}-\frac{f'}{2f}$ &
$\frac{1}{h_1}(\frac{f'\int\sqrt{|f|}dx}{2f\sqrt{|f|}})' $ &
$\partial_t,\frac{1}{2}(\mu-2\nu)t\partial_t+(\mu-\nu)\frac{\int \sqrt{|f|} dx}{\sqrt{|f|}}\partial_x+u\partial_u$  \\
& $~$ & $~$ & $=(\frac{f'\int\sqrt{|f|}dx}{2f\sqrt{|f|}})'$ & $~ $
& $~$  \\
\refstepcounter{tbn}\label{CaseUsualHpowerK0HForAllF}\thetbn &
 $\forall$ & $1$ & $|x|^p$
 & $|x|^q$ &$\partial_t,\, \frac{1}{2}\mu(1+p-q)t\partial_t+\mu x\partial_x+(1+q)u\partial_u$  \\
\refstepcounter{tbn}\label{CaseUsualHpowerK0Ff3}\thetbn & $\forall$ &
$1$ & $e^{px}$ & $\epsilon e^{qx} $ & $\partial_t,\,
\frac{1}{2}\mu(p-q)t\partial_t+\mu \partial_x+qu\partial_u$  \\
\refstepcounter{tbn}\label{CaseUsualHpowerK0F1}\thetbn & $\forall$ &
$u^{\mu}\ln u+h_1u^{\mu}$ & $\frac{f'}{f}=-2(g^{-1})'g$ &
$(g^{-1})''=\beta g $
& $\partial_t,\,-\frac{1}{2}\beta \mu t\partial_t+g^{-1}\partial_x+\beta u\partial_u$  \\
\refstepcounter{tbn}\label{CaseUsualHu-4K0HForallF} \thetbn & $\forall$
& $u^{\mu}$ & $\forall$ & $\forall$
& $\partial_t,\,\frac{1}{2} \mu t\partial_t-u\partial_u$  \\
\refstepcounter{tbn}\label{CaseUsualHu-4K0Ff3} \thetbn &
$\neq-\frac{4}{3}$ & $u^{\mu}$ & $f^3(x)$ & $g^3(x)$ &
$\tau=\frac{1}{2}\mu c_2t+c_1,
\xi=\xi^3(x)$\\
&  & & & & $ \eta=[(4\Phi_3+\Psi_3)\xi^3-c_2]u$  \\
\refstepcounter{tbn}\label{CaseUsualHu-4K0F1} \thetbn &
$\neq-\frac{4}{3}$ & $u^{\mu}$ & $f^4(x)$ & $g^4(x)$ &
$\tau=\frac{1}{2}\mu c_2t+c_1,
\xi=\xi^4(x)$\\
&  & & & & $ \eta=[(4\Phi_4+\Psi_4)\xi^4-c_2+2c_3]u$  \\
\refstepcounter{tbn}\label{CaseUsualHu-43K0F2} \thetbn & $-\frac{4}{3}$
& $u^{-\frac{4}{3}}$ & $f^5(x)$ & $g^5(x)$ & $\tau= c_2t+c_1,
\xi=c_3/\Phi_5$\\
&  & & & & $ \eta=\frac{3}{4\Phi_5^2}[2c_3\Phi_{5x}+(c_2-c_3)\Phi_5^2+2c_3\Phi_5g^5]u$  \\
\refstepcounter{tbn}\label{CaseUsualHu-43K0F3} \thetbn & $-\frac{4}{3}$
& $u^{-\frac{4}{3}}$ & $\forall$ & $-\frac{f_x}{2f}$ & $\tau=
c_2t+c_1,
\xi=\xi^6(x)$\\
&  & & & & $ \eta=\frac{3}{4}[c_2-2\xi^6_{x}-\xi^6\frac{f_x}{f}]u$  \\
\refstepcounter{tbn}\label{CaseUsualH1-43K1} \thetbn & $0$ & $u$
& $|x|^p$ & $|x|^q$ & $\partial_t, \frac{1}{2}(p+1)t\partial_t+x\partial_x-(q+1)u\partial_u$  \\
\refstepcounter{tbn}\label{CaseUsualH1-43K6} \thetbn & $0$ & $u$ &
$e^{px}$ & $\epsilon e^{qx}$ & $\partial_t,
\frac{1}{2}(p+2)t\partial_t+\partial_x-qu\partial_u$\\
\refstepcounter{tbn}\label{CaseUsualH1-43K2} \thetbn & $0$ & $\ln u
+h_0$
& $g^2\exp(\int g dx)$ & $(g^{-1})''=\beta g$ & $\partial_t, \frac{1}{2}t\partial_t+g^{-1}\partial_x+\beta u\partial_u$\\
\refstepcounter{tbn}\label{CaseUsualH1-43K3} \thetbn & $0$ & $u$
& $g^2\exp(\int g dx)$ & $(g^{-1})''=\beta g$ & $\partial_t, \frac{1}{2}t\partial_t+g^{-1}\partial_x+\beta\partial_u$\\
\refstepcounter{tbn}\label{CaseUsualH1-43K4} \thetbn & $0$ & $e^u+h_0$
& $g^2(\int g dx)^p$ & $(g^{-1}\int g dx)''=h_0 g$ & $\partial_t,
-\frac{1}{2}(p+2)t\partial_t-g^{-1}\int g dx\partial_x+\partial_u$\\
\refstepcounter{tbn}\label{CaseUsualH1-43K5} \thetbn & $0$ &
$u^{\nu}+h_0$ & $g^2(\int g dx)^p$ & $(g^{-1}\int g dx)''=h_0 g$ &
$\partial_t,
-\frac{1}{2}\nu(p+2)t\partial_t-\nu g^{-1}\int g dx\partial_x+u\partial_u$\\
\hline
\end{tabular}
\end{center}
{\footnotesize Here $\epsilon=\pm 1$ and $q, \alpha, \beta, h_0,
h_1\neq 0$. $f^3(x), g^3(x)$ and $\xi^3(x)$ satisfy the relation
\[
\begin{array}{ll}
4\mu\Phi_{3xx}-[(16\mu+32)\Phi_3+(7\mu+8)\Psi_3]\Phi_{3x}+\mu\Psi_{3xx}-[(5\mu+8)\Phi_3+(2\mu+2)\Psi_3]\Psi_{3x}
+(24\mu+32)\Phi_3^3 \\
+(14\mu+16)\Phi_3^2\Psi_3+(2\mu+2)\Phi_3
\Psi_3^2=0,\quad \xi^3_x-\mu\Phi_3\xi^3=0,
\end{array}{ll}
\]
where
\[
\Phi_3=-\frac{\mu}{3\mu+4}g^3-\frac{2(\mu+1)f^3_x}{(3\mu+4)f^3},\quad
\Psi_3=\frac{\mu}{3\mu+4}g^3+\frac{(5\mu+6)f^3_x}{(3\mu+4)f^3};
\]
$f^4(x), g^4(x)$ and $\xi^4(x)$ satisfy the relation
\[
4\mu\Phi_{4x}+\mu
\Psi_{4x}-(12\mu+16)\Phi_4^2-(7\mu+8)\Phi_4\Psi_4-(\mu+1)\Phi_4^2=0,\quad
\xi^4_x-\mu\Phi_4\xi^4=c_3,
\]
where
\[
\Phi_4=-\frac{\mu}{3\mu+4}g^4-\frac{2(\mu+1)f^4_x}{(3\mu+4)f^4},\quad
\Psi_4=\frac{\mu}{3\mu+4}g^4+\frac{(5\mu+6)f^4_x}{(3\mu+4)f^4};
\]
$f^5(x), g^5(x)$ satisfy the relation
\[
\Phi_5^2\Phi_{5xxx}-6\Phi_5\Phi_{5x}\Phi_{5xx}+6\Phi_{5x}^3-(g^5)^2\Phi_5^2\Phi_{5x}
+\Phi_5^3g^5_{xx}-2\Phi_5^2g^5_x\Phi_{5x}+\Phi_5^3g^5g^5_x=0,
\]
where
\[
\Phi_5=\frac{f^5_x}{f}+2g^5;
\]
and $\xi^6(x)$ satisfies third order ordinary differential
equation
\[
\xi^6_{xxx}-[(\frac{f_x}{2f})^2-2(\frac{f_x}{2f})_x]\xi^6_x+[(\frac{f_x}{2f})_{xx}+(\frac{f_x}{2f})_x]\xi^6=0.
\]
 }

\end{appendix}

\end{document}